\documentclass{amsart}
\allowdisplaybreaks
\usepackage{musicography}
\usepackage[normalem]{ulem}
\usepackage{slashed} 
\usepackage{tikz}
\usepackage{mathrsfs}
\usepackage{latexsym}
\usepackage{linearb}
\usepackage[LGR, OT1]{fontenc}
\usepackage{amsmath,amsthm}
\usepackage{cancel}
\usepackage{amssymb}
\usepackage{upgreek}
\usepackage[greek,english]{babel}
\usepackage{teubner}
\usepackage{MnSymbol, wasysym}
\usepackage{graphicx}
\usepackage{marvosym}
\usepackage{eufrak}
\usepackage[margin=1.4in,dvips]{geometry}
\usepackage{color}
\usepackage{hyperref}
\usepackage{comment}
\usepackage{fancyhdr}
\usepackage{mathtools}
\mathtoolsset{showonlyrefs,showmanualtags}

\def\f12{\frac 1 2}

\def\f12{\frac 1 2}

\makeindex

\newtheorem{definition}{Definition}[section]

\newtheorem{remark}{Remark}[section]
\newtheorem{lemma}{Lemma}
\newtheorem{theorem}{Theorem}[section]
\newtheorem*{theorem*}{Theorem}
\newtheorem*{corollary*}{Corollary}

\newtheorem{convention}{Convention}[section]

\newtheorem{corollary}{Corollary}[section]

\title[Boundedness of waves on $C^1$ Stationary Axisymmetric 
 Perturbations of Kerr]{Boundedness for the wave equation on $C^1$ Stationary Axisymmetric Perturbations of Kerr}

\date\today

    \author{Yakov Shlapentokh-Rothman}
    \address{\small  Department of Mathematics, University of Toronto,  40 St.~George~Street, Toronto, ON, Canada} 
	\address{\small Department of Mathematical and Computational Sciences, 
		University of Toronto Mississauga, 3359 Mississauga Road, Mississauga, ON, Canada}
        \email{yakovsr@math.toronto.edu}
    \author{Mihai Tohaneanu}
    \address{\small  Department of Mathematics, University of Kentucky,  719 Patterson Office Tower, Lexington, KY, USA}
    \email{mihai.tohaneanu@uky.edu}
    
\begin{document}

\maketitle

\begin{abstract}
     On the full range of sub-extremal Kerr exterior spacetimes we give a new proof of energy boundedness for high-frequency projections of solutions to the wave equation onto trapped frequencies. A key feature of the new estimate is that it circumvents the use of an integrated local energy decay (ILED) statement. As an illustration of the robustness of the estimate, we use it to establish  energy boundedness for solutions to the wave equation on stationary and axisymmetric metrics which are merely $C^1$ close to a sub-extremal Kerr spacetime. We show explicitly that such perturbed metrics may possess stably trapped null geodesics, and thus one does not expect ILED statements to hold.
\end{abstract}
\tableofcontents
\section{Introduction}
In view of its role as a proxy for the more complicated equations of linearized gravity, the wave equation
\begin{equation}\label{waveEqnKerr}
\Box_{g_K}\psi = F
\end{equation}
on subextremal Kerr spacetimes $\left(\mathcal{M}_K,g_K\right)$ has been an object of intense study in the last 25 years. Following earlier works in the non-rotating~\cite{KayWald,redshiftSchw,bluesoffer,bluesterbenz} and slowly rotating cases~\cite{slowlybound,claylecturenotes,TatTohILED,AndBlueKerr,TohStric}, uniform boundedness of energy and integrated local energy decay for solutions to~\eqref{waveEqnKerr} in the full sub-extremal range was established in~\cite{realmodestab,waveKerrlargea}.\footnote{Though it will not directly concern us here, we note that the precise pointwise leading order asymptotics for the wave equation, the so-called ``Price's law,'' has also been established in~\cite{hintzprice,aagprice}.} In~\eqref{schematicIELDwithbound} below we state the key estimate for~\eqref{waveEqnKerr} that is proven in~\cite{waveKerrlargea}.

Let $\Sigma_0$ denote an asymptotically flat space-like hypersurface intersecting the event horizon transversally  and $n_{\Sigma_0}$ denote the future-pointing unit normal to $\Sigma_0$. Then let $\left\{\Sigma_{\tau}\right\}_0^{\infty}$ denote the set of time-translations of $\Sigma_0$. After choosing an appropriately normalized set of time-independent vector fields which span the tangent space of $\mathcal{M}_K$,  we choose a corresponding set of multi-indices $\alpha$, and let $\psi^{(\alpha)}$ denote the application of the associated subset of the vector fields applied to $\psi$.  Let $(M,a)$ denote the mass and specific angular momentum of a given Kerr black hole, and then set $r_+ \doteq M + \sqrt{M^2-a^2}$. For a suitable bounded interval $I \subset (r_+,\infty)$, we let $\zeta(r)$ be a function which is $0$ for $r \in I$ and identically $1$ for $r\not\in I$. 

Then, for every $k \geq 0$ and $\delta > 0$, solutions $\psi$ to~\eqref{waveEqnKerr} with $F = 0$ satisfy
\begin{equation}\label{schematicIELDwithbound}\begin{split}
&\sup_{\tau > 0}\sum_{1 \leq \left|\alpha\right| \leq k+1}\int_{\Sigma_{\tau}}\left|\psi^{(\alpha)}\right|^2 + \int_0^{\infty}\int_{\Sigma_{\tau}}\left[\zeta\left(r\right)\sum_{1 \leq |\alpha| \leq { k+1}}\left|\psi^{(\alpha)}\right|^2 + r^{-2}\sum_{0 \leq |\alpha| \leq k}\left|\psi^{(\alpha)}\right|^2\right]r^{-1-\delta}\, d\tau 
\\ &\qquad \lesssim_{k,\delta} \sum_{1 \leq \left|\alpha\right| \leq k+1}\int_{\Sigma_0}\left|\psi^{(\alpha)}\right|^2.
\end{split}\end{equation}

One unusual feature of the estimate~\eqref{schematicIELDwithbound} is that the upper bound for the energy flux through $\Sigma_{\tau}$ is required to be proved simultaneously with an integrated local energy decay statement. The lack of an ``easy'' proof for boundedness of energy is related to the presence of superradiance (which we quickly review below in Section~\ref{superduper}). This means that the techniques that go into the proof of~\eqref{schematicIELDwithbound} cannot be applied to establish boundedness of energy on spacetimes where there exists both superradiance and stable trapping. (See Section~\ref{stabletrapsec} for further discussion and motivation for considering such classes of spacetimes.) Though not physically relevant, \emph{per se}, a particular test case one can consider is to study the wave equation on spacetimes which are stationary and axisymmetric and close to a sub-extremal Kerr spacetime only in $C^1$. As we show in Appendix~\ref{stabletrapapp} these metrics may possess stably trapped null geodesics, and thus one cannot expect to prove an integrated energy decay statement. \emph{One of our main results is exactly to show that the uniform boundedness of energy for solutions to the wave equation continues to hold on such spacetimes.}

Before we state our key theorems below, we wish to particularly note the previous work~\cite{slowlybound} which showed that the boundedness of energy for solutions to the wave equation continues to hold for the wave equation on spacetimes which are $C^1$ stationary and axisymmetric perturbations of a \emph{slowly rotating Kerr spacetime} with $|a| << M$. While the proofs given in~\cite{slowlybound} use the $|a| << M$ assumption in an essential way, the key insights of the work will nevertheless play a fundamental role for us.

The rest of the introduction is structured as follows. In Section~\ref{themainresultssection} below we state our two main theorems. In Sections~\ref{superduper} and~\ref{stabletrapsec} we will review the phenomena of superradiance and stable trapping  and their relevance to the behavior of solutions to the wave equation. In Section~\ref{quasiintro} we will discuss potential applications of our new linear estimates to the study of small data solutions to quasilinear waves on sub-extremal Kerr black holes. In Section~\ref{secfurther} we will discuss various future directions for the techniques developed here. Finally in Section~\ref{sketchit} we provide a quick sketch of the proofs of our results.
    
\subsection{Statement of main results}\label{themainresultssection}
In this section we will state our main results. Our first result will involve an operator $\mathcal{P}_{HT}$ which is precisely defined in Definition~\ref{phtwithtehextend} below, which serves to project the solution $\psi$ onto the high frequencies which experience trapping. For the reader who wishes to jump directly to the main results below before reading Section~\ref{Kerrspace}, we note that $\mathcal{P}_{HT}$ is a pseudo-differential operator in $(t^*,\phi^*)$ that is classical with respect to the Fourier variable associated to $t^*$, but discrete with respect to the Fourier variable associated to $\phi^*$. Following~\cite{blackboxlargea}, we will call such operators $(t^*, \phi^*)$ pseudo-differential operators. 

Associated to the definition of the operator $\mathcal{P}_{HT}$ are two free small parameters $\delta_1 > 0$ and $\delta_2 > 0$ which determine the largeness of the frequencies in the support of $\mathcal{P}_{HT}$ and how close to superradiant frequencies we allow the support of $\mathcal{P}_{HT}$ to go. The statements of our theorems will also refer to various vector fields and hypersurfaces which are defined later in Section~\ref{Kerrspace}. However, we will go ahead and note that $E_k\left[f\right]$ will denote an energy flux along a suitable spacelike hypersurface of up to $k$ derivatives of $f$. (Since the energy flux involves already a derivative, we thus see that $E_k\left[f\right]$ will involve a total of $k+1$ derivatives of $f$.)

The following convention is convenient.
\begin{convention}
We introduce the convention that whenever a result is phrased in terms of an underlying sub-extremal Kerr spacetime, all implicit constants may depend on the choice of $a$ and $M$.
\end{convention}

The following notation will be useful for stating our main theorems. We start with a cut-off function which we will use for time-averages.
\begin{definition}
Let $\chi(t)$ be a bump function which is identically $1$ for $|t| \leq 1$ vanishes for $|t| \geq 2$. For any $\tau \in \mathbb{R}$, we define $\chi_{\tau}\left(t\right) \doteq \chi\left(t-\tau\right)$. 
\end{definition}

Corresponding to the bump function $\chi(t)$ and any $\tau > 0$, will be another bump function $\tilde{\chi}_{\tau}(t)$ defined by
\[\tilde{\chi}_{\tau}\left(t\right) \doteq \sqrt{\int_{t-\tau}^t\chi^2(s)\, ds}\]
More generally, we also define, for $A < \tau$, $\tilde{\chi}_{\tau,A}(t)$ by
\begin{equation}\label{tildechitauA}
\tilde{\chi}_{\tau,A}(t) \doteq  \sqrt{\int_{t-\tau}^{t-A}\chi^2(s)\, ds}.
\end{equation}
We note that $\tilde{\chi}_{\tau,A}(t)$ is supported in $[A-2,\tau+2]$, is identically equal to  $\sqrt{\int_{-\infty}^{\infty}\chi^2(t)}\, dt$ when $t \in [A+2,\tau-2]$, and satisfies that $\left\vert\left\vert \tilde{\chi}_{t,A}\right\vert\right\vert_{C^N_t} \lesssim_N 1$. Finally we note that $\tilde{\chi}_{\tau,A}(t)$ is monotonically increasing in $\tau$.

Next, we define a separate cut-off function which is used to separate the near horizon region from the rest of the spacetime.
\begin{definition}Let $\xi(r) \geq 0$ be a cut-off function which is identically $0$ for $r \leq r_+ + \delta_0$, is identically $1$ for $r \geq r_+ + 2\delta_0$, and satisfies $|\xi'| \lesssim \delta_0{^{-1}}$. Here $\delta_0$ is a small constant fixed later in the paper.
\end{definition}

Our first main result is an estimate for solutions to the inhomogeneous wave equation on any rotating sub-extremal Kerr spacetime whose proof does not rely on the existence of any integrated local energy decay estimate.  We refer the reader to Section 2 for the definitions of the vector fields $K$, $T$, $Y$, and $N$. 
\begin{theorem}\label{mainLinEstTheo}Let $M > 0$ and $a \in (-M,M)$. Then, for every $\epsilon > 0$, $k \geq 0$, and choice of small constants $\delta_1 > 0$ and $\delta_2 > 0$, we may choose the parameter $A_{\rm high}$ in Definition~\ref{hightrappedpart} sufficiently large  and a time independent vector field $N$ which  equals $T$ for $r$ large, so that all solutions $\psi$ to the wave equation
\[\Box_{g_K}\psi = F,\]
with $F$ smooth and compactly supported on $J^{+}(\Sigma_1)$ and $\left(n_{\Sigma_0}\psi,\psi\right)|_{\Sigma_0}$ smooth and compactly supported, satisfy
\begin{equation}\label{thekeyEstimate}\begin{split}
 \sup_{t > 0}&\int_{-\infty}^{\infty}\chi^2_{t} (s)E_k\left[\mathcal{P}_{HT}\psi\right]\left(s\right)\, ds  \lesssim_k 
\\ & A_{\rm high}^{-1}\sup_{t > 0}\int_{-\infty}^{\infty}\chi^2_{t}(s)E_k\left[\psi\right](s)\, ds + \sum_{j=0}^k\sup_{t > 0 }\left|\int_{\mathcal{M}_K}\tilde{\chi}^2_{t}{\rm Re}\left(\left( T^j \mathcal{P}_{HT}F\right)\overline{K\left(T^j \mathcal{P}_{HT}\psi\right)}\right)\right|
\\  & +\sum_{j+l=1}^k\sup_{t > 0 }\left|\int_{\mathcal{M}_K}\tilde{\chi}_{t}^2{\rm Re}\left(\left( T^jY^l \mathcal{P}_{HT} F\right)\overline{N\left(T^jY^l\mathcal{P}_{HT}\psi\right)}\right)\right|+ E_k\left[\psi\right](0)
\\ &+\sum_{j=0}^k \sup_{t>0}\left|\int_{\mathcal{M}_K}{\rho \sqrt{\frac{\Delta}{\Pi}}}\chi_{t}^2\xi{\rm Re}\left(\left( T^j \mathcal{P}_{HT}F\right)\overline{\left(T^j\mathcal{P}_{HT}\psi\right)}\right)\right|
\\ &+\sup_{t > 0}\int_{-\infty}^{\infty}\chi_{t}^2(s)E_{k-2}\left[\mathcal{P}_{HT}F\right]\left(s\right)\, ds+\sup_{t > 0}\int_{\mathcal{M}_K}\chi_{t}^2\left|\mathcal{P}_{HT}F\right|^2,
\end{split}
\end{equation}
where the final line on the right hand side is absent if $k = 0$, and, if $k =1$, the term involving $E_{k-2}$ is absent.
\end{theorem}

\begin{remark}By standard density arguments, the regularity and support assumptions on $\psi$ and $F$ in Theorem~\ref{mainLinEstTheo} may be weakened as long as the corresponding terms on the right hand side of~\eqref{thekeyEstimate} remain finite.
\end{remark}

\begin{remark}One could remove the term involving $A_{\rm high}^{-1}\sup_{\tau > 0}E_k\left[\psi\right](t)$ on the right hand side of~\eqref{thekeyEstimate} at the expense of adding a term $E_k\left[\mathcal{P}_{HT}\psi\right](0)$ to the right had side, and thus the estimate would appear more self-contained with respect to $\mathcal{P}_{HT}\psi$. However, we have chosen for our main estimate to be of the form~\eqref{thekeyEstimate} as it is more suitable for our later applications in view of the fact that we do not want to assume an a priori bound for $E_k\left[\mathcal{P}_{HT}\psi\right](0)$.
\end{remark}
\begin{remark}Using finite in time energy estimates, one can easily replace the term
\begin{equation}\label{thetimeaveragedenergy}
\sup_{t > 0}\int_{-\infty}^{\infty}\chi_{t} (s)E_k\left[\mathcal{P}_{HT}\psi\right]\left(s\right)\, ds
\end{equation}
on the left hand side of~\eqref{thekeyEstimate} with $ \sup_{t>0}E_k\left[\mathcal{P}_{HT}\psi\right]\left(t\right)$, but working with the time averaged energy~\eqref{thetimeaveragedenergy} is in fact slightly more flexible for applications.
\end{remark}
\begin{remark}\label{Nabsent}If $F$ is supported in $\{r \geq r_+ + \tilde{\delta}\}$ for some $\tilde{\delta}$ then, if we let the implied constants now depend on $\tilde{\delta}$, we may drop the term of the right hand side of~\eqref{thekeyEstimate} which involves the vector field $N$. 
\end{remark}
\begin{remark}We will actually prove a more general result, see Theorem~\ref{actuallywehatweneedformain} in Section~\ref{proofthemainlineestsec}. The more general result will allow us to obtain estimates for solutions to $\Box_g\psi = 0$ where $g\in \mathscr{A}_{\epsilon,r_{\rm low},R_{\rm high}}$ (see Definition~\ref{themetricclass} below). 
\end{remark}
\begin{convention}When $a = 0$, the above theorem (even without the projection $\mathcal{P}_{HT}$!) already follows from~\cite{slowlybound,claylecturenotes}. Thus we will assume from now on that $a \neq 0$.
\end{convention}

We now turn to discussing an application of Theorem~\ref{mainLinEstTheo} to the study of solutions to the wave equation on a certain class of perturbations of Kerr spacetimes where integrated energy decay does \emph{not} generally hold. We first define the class of spacetimes.
\begin{definition}\label{themetricclass}Let $g_{\rm ref}$ be any choice of Riemannian metric on $\mathcal{M}_K$ that satisfies $\mathcal{L}_Tg_{\rm ref} = 0$. Choose $r_+ < r_{\rm low} < R_{\rm high} < \infty$ with $|r_+-r_{\rm low}| \ll 1$ and $R_{\rm high} \gg 1$. Then we let $\mathscr{A}_{\epsilon,r_{\rm low},R_{\rm high}}$ denote the set of Lorentzian metrics $g$ on $\mathcal{M}_K$ which satisfy the following
\begin{enumerate}
    \item $\mathcal{L}_Tg = \mathcal{L}_{\Phi}g = 0$.
    \item $r \not\in [r_{\rm low},R_{\rm high}] \Rightarrow g = g_K$.
    \item $\left\vert\left\vert g - g_K\right\vert\right\vert_{C^1_{g_{\rm ref}}} \leq \epsilon$.
\end{enumerate}
\end{definition}

\begin{convention}\label{episalwayssmall}We introduce the convention that, unless said otherwise, $\epsilon$ may be assumed sufficiently small depending on any other fixed constants or parameters we introduce.
\end{convention}

In Appendix~\ref{stabletrapapp} we show that for every $\epsilon > 0$ there exists metrics $g \in \mathscr{A}_{\epsilon,r_{\rm low},R_{\rm high}}$ so that $g$ has stably trapped null geodesics. Our main result concerning this class of metrics is the following:
\begin{theorem}\label{theoc1pert}Let $k \geq 1$ be an integer, $\epsilon > 0$ be suitably small, possibly depending on $k$, and let $g \in \mathscr{A}_{\epsilon,r_{\rm low},R_{\rm high}}$. Let $\psi : J^+\left(\Sigma_0\right) \to \mathbb{C}$ solve
\[\Box_g\psi = 0,\]
and satisfy that $\left(n_{\Sigma_0}\psi,\psi\right)|_{\Sigma_0}$ is smooth and compactly supported. Then, for every integer $k \geq 1$, we have 
\begin{equation}\label{enboundc1pert}
\sup_{t \geq 0}E_k\left[\psi\right](t) \lesssim_k E_k\left[\psi\right](0).
\end{equation}
\end{theorem}
\begin{remark}By standard density arguments, the regularity and support assumptions on $\psi$ in Theorem~\ref{theoc1pert} may be weakened as long as the corresponding norm $E_k\left[\psi\right](0)$ is finite.
\end{remark}
\begin{remark}Theorem~\ref{theoc1pert} would still hold, with essentially the same proof, if one broadened the definition of $\mathscr{A}_{\epsilon,r_{\rm low},R_{\rm high}}$ to allow metric perturbations along the event horizon and for ones which fall off polynomially at infinity. We have restricted to the simpler class of Definition~\ref{themetricclass} to keep the exposition more streamlined and because this class is already rich enough to create stable trapping.
\end{remark}

\subsection{Superradiance}\label{superduper}
Many of the difficulties involved in studying the wave equation on Kerr spacetimes are due to the presence of the ergoregion and the associated phenomenon of superradiance. The ergoregion is the region of spacetime where the stationary Killing vector field $T$ is spacelike, and superradiance is the resulting fact that the event horizon energy flux associated to $T$ may be negative. As a consequence of these geometric facts, it is not possible to obtain any sort of energy boundedness result directly from the use of conservation laws. 

We refer the reader to the introduction of~\cite{waveKerrlargea} for a more thorough discussion of superradiance on the Kerr spacetime. Here we will just make two additional points about superradiance. The first is that in view of the fact that various natural equations acquire exponential superradiant instabilities on rotating Kerr spacetimes (see~\cite{KGblowup,shortrangeblowup}), the difficulties posed by superradiance should not be considered purely ``technical.'' The second point is a fundamental insight which originates in~\cite{slowlybound} in the $|a| \ll M$ setting, but also plays a fundamental role for the full sub-extremal range $|a| < M$ in~\cite{waveKerrlargea} and the present work. Namely, by using Fourier analysis in $(t^*,\phi^*)$, any solution $\psi$ to the wave equation may be split in two pieces $\psi = \psi_{\flat} + \psi_{\sharp}$ where $\psi_{\flat}$ is immune to trapping in that it is possible to establish an integrated local energy decay statement for $\psi_{\flat}$ with no regularity loss (see the discussion of trapping below in Section~\ref{stabletrapsec}) and $\psi_{\sharp}$ is immune to superradiance in that its formal flux along the event horizon is positive (see Section~\ref{sketchit} below for more details).

\subsection{Stable/unstable trapping}\label{stabletrapsec}
We say that a null geodesic on a black hole spacetime is trapped if as the affine parameter $s\to \infty$ it stays within the time-translations of a bounded spacelike region, neither escaping to infinity nor falling into the black hole. As is well-known, the proof of integrated local energy decay statements such as~\eqref{schematicIELDwithbound} requires, one way or another, that one understands quantitatively that all trapped null geodesics become non-trapped upon suitable small perturbations (see, for example, the discussions of trapping in the introduction to~\cite{waveKerrlargea}). However, even though one can still establish an integrated energy decay statement if the trapping is sufficiently unstable, the resulting estimate must lose regularity~\cite{jangauss}. This is ultimately the reason for the presence of the function $\zeta(r)$ in~\eqref{schematicIELDwithbound}.

When one has stably trapped null geodesics, then one expects at most logarithmic type decay for solutions to the wave equation; in particular, the integrated energy decay statements cannot be expected to hold in this context. The first results in mathematical relativity concerning stable-trapping and the associated slow decay are contained in the works~\cite{gannotAds,holzsmulquasi} which studied solutions to the wave equation on Kerr--AdS spacetimes. Here the stable trapping is associated with the fact that null infinity is timelike; in particular these works showed that integrated local energy decay cannot hold on the backgrounds. 

More relevant for us, however, is the work~\cite{stabletrapring} which showed that the stable trapping phenomenon is ubiquitous for higher dimensional asymptotically flat black holes. More specifically, the work~\cite{stabletrapring} showed that stable trapping occurs for a wide class of five dimensional black ring spacetimes and used this fact to establish the following dichotomy: For any of the black ring spacetimes considered in~\cite{stabletrapring}, either uniform boundedness of energy fails \emph{or} solutions to the wave equation generally decay logarithmically. 

\emph{All black ring spacetimes studied in~\cite{stabletrapring} possess an ergoregion, superradiance, and stable trapping, and it currently remains an open problem to determine whether or not energy boundedness holds for any of these spacetimes.} We hope that the techniques developed in this paper will be useful for addressing this question.

\subsection{Quasilinear problems in the full sub-extremal range $|a| < M$}\label{quasiintro}
We now discuss the relevance of our new linear estimate Theorem~\ref{actuallywehatweneedformain} to the study of quasilinear wave equations on sub-extremal Kerr backgrounds. In order to apply the results from~\cite{waveKerrlargea} to a nonlinear problem, one of course needs a version of the estimate~\eqref{schematicIELDwithbound} which allows for $F \neq 0$. One such estimate which may be derived from the main results of~\cite{waveKerrlargea}  (see Remark 3.2.1 and Appendix D of~\cite{blackboxsmalla}) is the following:
\begin{equation}\label{schematicIELDwithboundinhomog}\begin{split}
&\sup_{\tau > 0}\sum_{1 \leq \left|\alpha\right| \leq k+1}\int_{\Sigma_{\tau}}\left|\psi^{(\alpha)}\right|^2 + \int_0^{\infty}\int_{\Sigma_{\tau}}\left[\zeta\left(r\right)\sum_{1 \leq |\alpha| \leq k+1}\left|\psi^{(\alpha)}\right|^2 + \sum_{0 \leq |\alpha| \leq k} r^{-2}\left|\psi^{(\alpha)}\right|^2\right]r^{-1-\delta}\, d\tau \lesssim_{k,\delta}
\\ &\qquad \sqrt{\int_0^{\infty}\int_{\Sigma_{\tau} \cap \{r \leq R\}}\sum_{0\leq \left|\alpha\right| \leq k}\left|F^{(\alpha)}\right|^2\, d\tau}\sqrt{\int_0^{\infty}\int_{\Sigma_{\tau} \cap \{r \leq R\}}\sum_{\left|\alpha\right| \leq k+1}\left|\psi^{(\alpha)}\right|^2\, d\tau} 
\\ &\qquad  +\int_0^{\infty}\int_{\Sigma_{\tau}\cap \{r \geq R\}}\sum_{0\leq \left|\alpha\right| \leq k}\left|F^{(\alpha)}\left(w_1 \partial_r\psi^{(\alpha)} + w_2 \psi^{(\alpha)} + \partial_t\psi^{(\alpha)}\right)\right|\, d\tau 
\\ &\qquad + \sum_{1 \leq \left|\alpha\right| \leq k+1}\int_{\Sigma_0}\left|\psi^{(\alpha)}\right|^2+ \int_0^{\infty}\int_{\Sigma_{\tau}}\sum_{0\leq \left|\alpha\right| \leq k-1}\left|F^{(\alpha)}\right|^2\, d\tau,
\end{split}\end{equation}
for suitable weights $w_1$, $w_2$, where the final term on the right hand side is absent when $k = 0$.

In the second line of~\eqref{schematicIELDwithboundinhomog} we have a term proportional to the spacetime $L^2$-norm of $F^{(\alpha)}$. This turns out to limit the usefulness of the estimate~\eqref{schematicIELDwithboundinhomog} when establishing ``top order'' estimates for quasilinear wave equations. An essentially equivalent way to phrase the issue is that if one attempts to apply the estimate~\eqref{schematicIELDwithboundinhomog} directly to a solution $\psi$ of $\Box_g\psi = 0$, where $g$ is a perturbation of $g_K$, by writing $\Box_{g_K}\psi = \Box_{g_K}\psi - \Box_g\psi$, then the resulting estimate cannot close on its own because the right hand side will involve $k+2$ derivatives of $\psi$ while the left hand side involves only $k+1$ derivatives.\footnote{For semilinear problems the estimate~\eqref{schematicIELDwithboundinhomog} still seems to have the issue that there is a spacetime norm on the right hand side involving $k+1$ derivatives of $\psi$ while, because of the degenerate due to trapping, the spacetime norm on the left hand side only involves $k$ derivatives of $\psi$. One way to surmount this issue is to start with estimate~\eqref{schematicIELDwithbound} and add in the inhomogeneity via Duhamel's principle. In the resulting estimate, one puts the inhomogeneity $F^{(\alpha)}$ in $L^1_t$. See, for example,~\cite{LukNLM}. However, this trick does not suffice for a quasilinear problem.} For more details see the discussion in the introductions of~\cite{blackboxsmalla,blackboxlargea}.

 This difficulty with quasilinear problems  has been recently resolved in~\cite{blackboxsmalla,blackboxlargea} for the full sub-extremal range $|a| < M$. The key is a new linear estimate to be applied at top order which involves partitioning the solution into various pieces using Fourier analysis in $(t^*,\phi^*)$, applying certain physical space currents, and finally controlling certain remaining terms by exploiting the fact that the wave operator is elliptic in certain regions of phase space. For the specific piece $\mathcal{P}_{HT}\psi$ of the solution, one could, in principle, use a variant of Theorem~\ref{actuallywehatweneedformain} to establish a top order estimate. (Note that when we prove Theorem~\ref{actuallywehatweneedformain} we have to resolve an analogous top order estimate difficulty.) In the specific context of~\cite{blackboxsmalla,blackboxlargea} this would not lead to an improved result, but it may be useful to keep this flexibility in mind for other problems.

 We note also the work~\cite{szeftelma} which has provided new linear microlocal estimates for solutions to the wave equation which, when combined with the results of~\cite{waveKerrlargea}, resolve the derivative loss problem at top order. 

 Our fundamental interest in this paper are results which hold in the full subextremal regime $|a| < M$. However, before we move to the next section it is worth noting that in the case $a = 0$ or $|a| \ll M$ the specific problem we have mentioned above with the top order estimate is simpler to handle, and there has been tremendous progress both on small data stability for quasilinear wave equations and on the corresponding black hole stability problem, see~\cite{blackboxsmalla,hansmihaiquasischw,hansmihailocalenergKerr,hansmihaiweaknullkerr,fedmaxborn,schwstabdhrt,polstabschw,klainszefsmallastab,giorklainszeftwaveest}.

\subsection{Further directions}\label{secfurther}We have already mentioned in Section~\ref{stabletrapsec} open problems concerning higher dimensional asymptotically flat black holes that could possibly be addressed with the techniques developed in this paper. In this section we will discuss some other potential directions for extending these results.

For the equations governing true linearized gravitational perturbations around a Kerr black hole, the gauge-invariant part of the system are controlled by solutions to the Teukolsky equation. Energy boundedness and integrated energy decay estimates for the Teukolsky equation have been established in $a= 0$~\cite{linestabDHR}, the case of $|a| \ll M$~\cite{teukDHRsmalla,mateuksmalla}, and in the full sub-extremal case $|a| < M$~\cite{teuksubextremalI,teuksubextremalII}. (For precise leading order asymptotics see~\cite{mapriceteuk,teukprecise}.) It would be interesting to establish analogues of Theorems~\ref{mainLinEstTheo} and~\ref{theoc1pert}; however, in view of the lack of a Lagrangian structure for the Teukolsky equation, it is clear that new ideas would have to be introduced.

While the focus of this paper is on the sub-extremal case $|a| < M$, the extremal case $|a| = M$ is both physically relevant and very interesting mathematically. In fact, for the wave equation on an extremal Kerr background it remains an open question to establish an analogue of the boundedness and integrated energy decay statement~\eqref{schematicIELDwithbound}. For the state of the art see~\cite{aretakisKerr,aretakisinstabilityextreme,azimuthalgajic}. With the current proofs, the implied constants in the estimates of Theorems~\ref{mainLinEstTheo} and~\ref{theoc1pert} blow-up in the extremal limit $|a| \to M$. This is due both to the loss of the red-shift along the event horizon, and to the degeneration of the disjointness of trapping and superradiance (see the discussion of this latter point in the introduction of~\cite{waveKerrlargea}). It would be interesting to determine if a suitable modification of the estimates in Theorems~\ref{mainLinEstTheo} and~\ref{theoc1pert} could be shown when $|a| = M$.

Finally, we note that it is natural to attempt to prove analogues of our results in case of Kerr--Anti--de Sitter or Kerr--de Sitter black holes. We will not attempt here a general conjecture for either sign of the cosmological constant. However, we note that for $\Lambda < 0$, there is, in general, a big dichotomy in stable or unstable behavior of scalar waves depending on whether the Hawking--Reall bound~\cite{hawkreall} is satisfied and what type boundary conditions are posed at the timelike conformal boundary, see~\cite{doldunstable,HolzSmuldecay,holzsmulquasi,holzewarnwellposed} (and also the important precursor to these logarithmic decay results~\cite{burq} as well as the general logarithmic decay result of~\cite{moschidishlog} for asymptotically flat spacetimes). For the positive cosmological constant $\Lambda > 0$, all available results are consistent with the analogue of Theorem~\ref{theoc1pert} being true, but it remains an open problem to determine whether boundedness of energy and integrated energy decay for solutions to the wave equation hold in the full range of admissible parameters. For spacetimes with $\Lambda > 0$ and $|a| \ll M$, see~\cite{hintzvasystab,fangkerrds} and the references therein where the nonlinear stability of the Kerr--de Sitter family is proven. For the state of the art for the wave equation in the full range of admissible parameters, see~\cite{vasypeterstab,casacostmodestab,hintzmodestab,mavrokerrds}; \cite{vasypeterstab} establishes decay statements for the wave equation conditional on the validity of mode stability,~\cite{mavrokerrds} establishes the analogue of~\eqref{schematicIELDwithbound} also conditional on the validity of mode stability, and~\cite{hintzmodestab,casacostmodestab} establish mode stability for certain ranges of parameters.   

\subsection{Sketch of the proof}\label{sketchit}
In this final section of the introduction we will provide a sketch of the proof of our main results.

There are two main ingredients that go into the proof of Theorem~\ref{mainLinEstTheo}. The first ingredient is essentially algebraic; namely we formally consider a solution of the form 
\begin{equation}\label{pluginmode}
\psi = e^{-it\omega}e^{im\phi}Q\left(r,\theta\right),
\end{equation}
for $\omega \in \mathbb{R}$ and $m \in \mathbb{Z}$. If the parameters satisfy 
\begin{equation}\label{notsuper}
\omega\left(\omega - \frac{am}{2Mr_+}\right) > 0,
\end{equation}
then the horizon flux satisfies $\mathbf{J}^T_{\mu}n^{\mu}_{\mathcal{H}^+} \geq 0$, where $T$ denotes the stationary vector field. (This insight originates in~\cite{slowlybound}.) However, in order to prove Theorem~\ref{mainLinEstTheo} we ask ourselves something stronger, that is, can we construct a current which produces a positive energy along a spacelike foliation if  $\omega\left(\omega - \frac{am}{2Mr_+}\right) > 0$? It is useful to note, however, that in view of finite in time energy estimates, it suffices to have a positive energy after integration in time again a suitable bump function.  

Let $K = T + \frac{a}{2Mr_+}\Phi$ denote the Hawking vector field. Using the notations from Section 2.1,  we may compute for a function of the form~\eqref{pluginmode} along a constant $t$-hypersurface
\[\mathbf{J}^K_{\mu}n^{\mu}_{\{t = \rm const\}} = \rho^{-1}\sqrt{\frac{\Pi}{\Delta}}\underbrace{\left(\left(\omega  - \frac{a}{2Mr_+}m \right)\left(\omega  -\frac{2Mar}{\Pi} m\right) \right)}_{\doteq P\left(\omega,m,r\right)}\left|Q\right|^2+\frac{1}{2}\rho \sqrt{\frac{\Delta}{\Pi}} \partial^{\gamma}\psi \partial_{\gamma}\psi.\]
After averaging in time against a bump function the last term can be removed (up to a lower order term) by using a Lagrangian correction. Hence, we focus on whether or not $P > 0$. The key structural observation is that $P$ is indeed positive exactly in the the same frequency range~\eqref{notsuper}! Note also that in view of the fact that $K$ is timelike near the horizon, we can modify the constant $t$-hypersurface near the horizon and maintain the positivity of $\mathbf{J}^K_{\mu}n^{\mu}$ (this is why we work with the vector field $K$ as opposed to $T$). 

The second ingredient that goes into the proof of Theorem~\ref{mainLinEstTheo} is primarily analytic. Namely, so that we can take advantage of the observations above, we must show that, after averaging against bump functions in time, we have an approximate Plancherel's formula in time with acceptable errors. We also need estimates for certain lower order terms that appear. These estimates are achieved by relying on various standard aspects of the theory for the $(t^*,\phi^*)$ pseudo-differential calculus.

We now discuss the main ideas in the proof of Theorem~\ref{theoc1pert}. The basic idea is to split the solution $\psi$ into three pieces $\psi = \psi_{\rm bound} + \psi_{\rm NT} + \psi_{\rm HT}$, where the Fourier support of $\psi_{\rm bound}$ is in the region where $\omega^2+m^2 \lesssim 1$, the Fourier support of $\psi_{\rm NT}$ is in a region which does not experience trapping, and the Fourier support of $\psi_{\rm HT}$ covers high frequencies which experience trapping. For $\psi_{\rm bound}$ we are able to directly perturb the estimate for the exact Kerr metric; the potential derivative loss from terms like $\Box_{g_K} - \Box_g$ may be handled using elliptic estimates and the fact that $\omega^2+m^2 \lesssim 1$. For $\psi_{\rm HT}$ we use a generalization of Theorem~\ref{mainLinEstTheo} (see Theorem~\ref{actuallywehatweneedformain}) which allows us to prove boundedness of energy directly without an integrated energy decay statement. Finally, for $\psi_{\rm NT}$ we may exploit the disjointness from trapping by using an estimate from~\cite{blackboxlargea} to establish an integrated local energy decay statement. 

A key technical difficulty is that one cannot \emph{a priori} take the Fourier transform in time since the solution could grow exponentially. A similar difficulty occurs in~\cite{waveKerrlargea} (and the earlier~\cite{slowlybound}) and we resolve the issue in an analogous way here. Instead of working directly with the operator $\Box_g$, we study solutions to $\Box_{\tau}\psi_{\tau} = 0$, where $\Box_{\tau}$ interpolates between $\Box_g$ and $\Box_{g_K}$ in a region where $t\sim \tau$. We then prove estimates uniform as $\tau \to \infty$. This interpolation process creates a time-dependence for the operator $\Box_{\tau}$ which introduces additional errors, but these turn out to be controllable. 

 The paper is structured as follows. In Section 2 we introduce the Kerr metric, the current formalism, the main operator we will control, and the class of metrics we will study. In Section 3 we establish various theorems for $(t^*, \phi^*)$ pseudo-differential operators that will be useful later on. Section 4 contains the proof of Theorem~\ref{mainLinEstTheo}. Section 5 gives the proof of Theorem~\ref{theoc1pert}. Appendix A gives an example of a family of metrics with stable trapping for which our result applies. Finally, Appendix B gives a brief argument of why Theorem~\ref{fromtheblackboxlargea} holds.

\subsection{Acknowledgments}
YS acknowledges support from an Alfred P. Sloan Fellowship in Mathematics and from NSERC discovery grants RGPIN-2021-02562 and DGECR-2021-00093.
\section{Preliminaries}

\subsection{The Kerr spacetime}\label{Kerrspace}
In this section we will review some standard coordinate systems and hypersurfaces and fix our conventions for the Kerr spacetime. For a true introduction to Kerr family of black holes see, for example,~\cite{oneillKerr,wald}.

Let $M > 0$, $a \in (-M,M)\setminus\{0\}$, and set $r_{\pm} \doteq M \pm \sqrt{M^2-a^2}$. In Boyer-Lindquist coordinates $\left(t,r,\theta,\phi\right) \in \mathbb{R} \times (r_+,\infty) \times \mathbb{S}^2$ the Kerr metric takes the form
\begin{equation}\label{metric}
g_{BL} = -\left(1-\frac{2Mr}{\rho^2}\right)dt^2 - \frac{4Mar\sin^2\theta}{\rho^2}dtd\phi + \frac{\rho^2}{\Delta}dr^2
+ \rho^2 d\theta^2 + \sin^2\theta\frac{\Pi}{\rho^2}d\phi^2,
\end{equation}
\[\Delta := r^2 - 2Mr + a^2 = (r-r_+)(r-r_-),\]
\[\rho^2 := r^2 + a^2\cos^2\theta,\]
\[\Pi := (r^2+a^2)^2 - a^2\sin^2\theta\Delta.\]
\begin{convention}We introduce the convention for this paper that we take $a > 0$. The extension of all results to $a < 0$ will be immediate.
\end{convention}

We introduce the notation $T$ and $\Phi$ to refer to the Killing vector fields $\partial_t$ and $\partial_{\phi}$. The Hawking vector field refers to the particular linear combination of $T$ and $\Phi$ given by
\begin{equation}\label{hawk}
K \doteq T + \frac{a}{2Mr_+}\Phi.
\end{equation}

For future reference, we note also the inverse of $g_{BL}$:
\begin{equation}\label{inversemetric}\begin{split}
&g_{BL}^{-1} = \frac{-\Pi}{\rho^2\Delta}\partial_t\otimes \partial_t - \frac{2Mar}{\rho^2\Delta}\left(\partial_t\otimes \partial_{\phi} + \partial_t\otimes \partial_{\phi}\right)
\\ &\qquad\qquad + \frac{\Delta - a^2\sin^2\theta}{\Delta\rho^2\sin^2\theta}\partial_{\phi}\otimes\partial_{\phi} + \frac{\Delta}{\rho^2}\partial_r\otimes\partial_r + \rho^{-2}\partial_{\theta}\otimes\partial_{\theta},
\end{split}
\end{equation}
and also the volume form
\begin{equation}\label{volumeform}
dVol = \rho^2dtdrd\mathbb{S}^2.
\end{equation}

\begin{convention}If we do not write explicitly a volume form, our convention is that the integration is always with respect to the induced volume form.
\end{convention}

We have
\begin{equation}\label{thisisnablat}\begin{split}
&-\nabla t = -g^{tt}\partial_t - g^{t\phi}\partial_{\phi} = \frac{\Pi}{\rho^2\Delta}\partial_t + \frac{2Mar}{\rho^2\Delta}\partial_{\phi} \Rightarrow 
\\ &\qquad \qquad g_{BL}\left(-\nabla t,-\nabla t\right) = \frac{-\Pi}{\rho^2\Delta}.
\end{split}\end{equation}

In view of the inequality
\[\Pi = \Delta^2 + 4M^2r^2 + \Delta \left(4Mr - a^2\sin^2\theta\right)  > 0,\]
we see that $\nabla t$ is a timelike vector. We may thus take $-\nabla t$ to define our time orientation.

As is well-known, Boyer-Lindquist coordinates are not regular across the event horizon. A convenient coordinate system which is regular across the future event horizon $\mathcal{H}^+$ is given by the so-called Kerr star coordinates $\left(t^*,r,\phi^*,\theta\right) \in \mathbb{R} \times (r_-,\infty) \times \mathbb{S}^2$, where, when $r > r_+$, we have $t^* = t - \overline{t}\left(r\right)$ and $\phi^* = \phi - \overline{\phi}\left(r\right)$ for suitable functions $\overline{t}$ and $\overline{\phi}$ (which vanish for sufficiently large $r$). The specific form of the metric in Kerr star coordinates will not be important for us. Instead, the relevant properties are that $T$ and $\Phi$ extend as Killing vector fields to the entire domain of Kerr star coordinates, the time-orientation uniquely extends to the Kerr star coordinates, $\mathcal{H}^+ \doteq \{r = r_+\}$ is a null hypersurface with normal vector $K$, and $\left(\nabla_KK\right)|_{\mathcal{H}^+} = \kappa K$ for $\kappa > 0$ (where $\nabla$ is the connection associated to the metric).

We then fix the Kerr spacetime $\left(\mathcal{M}_K,g_K\right)$ to be a manifold with boundary with $\mathcal{M}_K \doteq \left\{(t^*,r,\phi^*,\theta) \in \mathbb{R} \times [r_+,\infty) \times \mathbb{S}^2\right\}$ and $g_K$ being the Kerr metric expressed in Kerr star coordinates. We remark that, while the interior of $\mathcal{M}_K$ coincides with the domain of outer communication, the boundary $\{r=r_+\}$ is the portion of the event horizon to the future of the bifurcation sphere. We depict the corresponding Penrose diagram below.
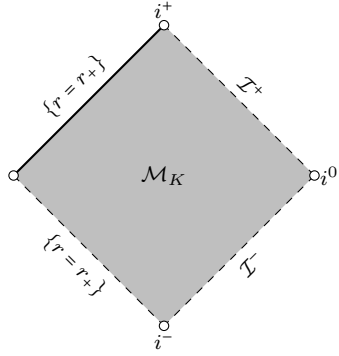
\begin{center}
\begin{tikzpicture}[scale = 1]
\fill[lightgray] (-2,0) -- (0,2) -- (2,0) -- (0,-2) -- (-2,0);
\draw[dashed] (2,0) -- (0,2) node[sloped,above,midway]{\footnotesize $\mathcal{I}^+$};
\draw[dashed] (2,0) -- (0,-2) node[sloped,below,midway]{\footnotesize $\mathcal{I}^-$};
\draw[thick] (-2,0) -- (0,2) node[sloped,above,midway]{\footnotesize $\{r = r_+\}$};
\draw[dashed] (-2,0) -- (0,-2) node[sloped,below,midway]{\footnotesize $\{r = r_+\}$};
\node at (0,0) {\footnotesize $\mathcal{M}_K$};
\path [draw = black, fill = white] (0,-2) circle (1/16);
\path [draw = black, fill = white] (2,0) circle (1/16);
\path [draw = black, fill = white] (0,2) circle (1/16);
\path [draw = black, fill = white] (-2,0) circle (1/16);
\node  at (0,2.20) {\footnotesize $i^+$};
\node at (0,-2.20) {\footnotesize $i^-$};
\node at (2.20,0) {\footnotesize $i^0$};
\end{tikzpicture}
\end{center}

Let $\delta_0$ be any sufficiently small constant so that the Hawking vector field $K$ is causal when $r \in [r_+,r_++10\delta_0]$ and so that when $r \in [r_+,r_++10\delta_0]$ we have the validity of the redshift estimate for the vector field $N$, which is defined below at the end of this section. We will arrange our choices of $\bar{t}(r)$ and $\bar{\phi}(r)$ so that $\{t^* = \tau\} \cap \{r \geq r_++\delta_0\}  = \{t = \tau\}\cap \{r \geq r_+ +\delta_0\}$.

We define a foliation $\left\{\Sigma_{\tau}\right\}_{\tau = -\infty}^{\infty}$ of $\mathcal{M}_K$ by letting $\Sigma_{\tau}  \doteq \{t^* = \tau\}.$ Note that when $r \geq r_+ + \delta_0$, $\Sigma_{\tau}$ is given by $\{t = \tau\}$. We will use $n_{\Sigma_{\tau}}$ to denote the future oriented normal vector to $\Sigma_{\tau}$. We also introduce the notation $\mathcal{R}\left(\tau_1,\tau_2\right)$ to denote the spacetime region $J^+\left(\Sigma_{\tau_1}\right) \cap J^-\left(\Sigma_{\tau_2}\right)$. In the following diagram we have depicted an example of a region $\mathcal{R}\left(\tau_1,\tau_2\right)$.
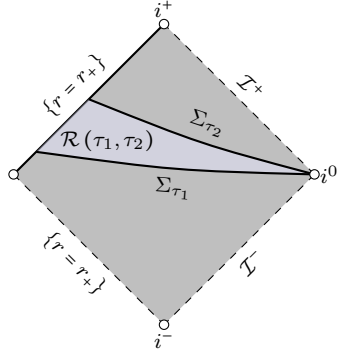
\begin{center}
\begin{tikzpicture}[scale = 1]
\fill[lightgray] (-2,0) -- (0,2) -- (2,0) -- (0,-2) -- (-2,0);
\draw[dashed] (2,0) -- (0,2) node[sloped,above,midway]{\footnotesize $\mathcal{I}^+$};
\draw[dashed] (2,0) -- (0,-2) node[sloped,below,midway]{\footnotesize $\mathcal{I}^-$};
\draw[thick] (-2,0) -- (0,2) node[sloped,above,midway]{\footnotesize $\{r = r_+\}$};
\draw[dashed] (-2,0) -- (0,-2) node[sloped,below,midway]{\footnotesize $\{r = r_+\}$};
\fill[blue!10, opacity=0.5] 
(2,0) .. controls (.15,.05) .. (-1.7,.3)  --
(-1,1) -- (-1,1) .. controls (.5,.4) .. (2,0);
\draw[thick] (2,0) .. controls (.15,.05) .. (-1.7,.3) node[sloped,below,midway]{\footnotesize $\Sigma_{\tau_1}$};
\draw[thick] (2,0) .. controls (.5,.4) .. (-1,1) node[sloped,above,midway]{\footnotesize $\Sigma_{\tau_2}$};
\node at (-.75,.45) {\footnotesize $\mathcal{R}\left(\tau_1,\tau_2\right)$};

\path [draw = black, fill = white] (0,-2) circle (1/16);
\path [draw = black, fill = white] (2,0) circle (1/16);
\path [draw = black, fill = white] (0,2) circle (1/16);
\path [draw = black, fill = white] (-2,0) circle (1/16);
\node  at (0,2.20) {\footnotesize $i^+$};
\node at (0,-2.20) {\footnotesize $i^-$};
\node at (2.20,0) {\footnotesize $i^0$};
\end{tikzpicture}
\end{center}

Lastly, we let $N$ denote a vector field which satisfies the following properties:
\begin{enumerate}
\item $N$ is timelike on all of $\mathcal{M}_K$.
\item $N$ commutes with $T$ and with $\Phi$.
\item $N = T$ for sufficiently large $r$.
\item\label{redshiftest} For any function $f$ and $r \in [r_+,r_++10\delta_0]$, we have $\mathbf{K}^N\left[f\right] \gtrsim \mathbf{J}^N_{\mu}\left[f\right]N^{\mu}$.
\end{enumerate}
Such vector fields and their corresponding estimates for solutions to wave equations were first introduced by Dafermos--Rodnianski~\cite{redshiftSchw,claylecturenotes}.

When writing out various energy norms it will be useful to introduce the following differential operators on the Kerr spacetime. We define these originally in Boyer-Lindquist coordinates:
\[\partial_{r^*} \doteq \frac{\Delta}{r^2+a^2}\partial_r,\qquad L \doteq \partial_{r^*}+T +\frac{a}{r^2+a^2}\Phi,\qquad \underline{L} \doteq \frac{r^2+a^2}{\Delta}\left(-\partial_{r^*}+T +\frac{a}{r^2+a^2}\Phi\right).\]
The vector fields $L$ and $\underline{L}$ are immediately seen to also be regular null vector fields in $\left(t^*,\phi^*,\theta,\phi\right)$ coordinates. We will use $Y$ to denote any choice of a vector field that is stationary and axisymmetric, vanishes for sufficiently large $r$, and satisfies that $Y|_{\mathcal{H}^+} = \underline{L}$. 

We let $\Omega_1$, $\Omega_2$, and $\Omega_3$ denote the standard angular momentum operators along $\mathbb{S}^2$, where we make the choice that $\Omega_1 = \Phi$. We then define
\[\left|\slashed{\nabla}\psi\right|^2 \doteq r^{-2}\sum_{i=1}^3\left|\Omega_i\psi\right|^2.\]
The vector fields $L$, $\underline{L}$, and $\{\Omega_i\}_{i=1}^3$ span the tangent space at each point of the Kerr spacetime. 

For higher order energies we introduce the following notation:
\[\tilde{\mathfrak{D}}_1 \doteq L,\qquad \tilde{\mathfrak{D}}_2 \doteq \underline{L},\qquad \tilde{\mathfrak{D}}_{2+i} \doteq r^{-1}\Omega_i.\]
(Note that we put a $r^{-1}$ weight as opposed to the corresponding definition in~\cite{blackboxlargea}.) For any multi-index $\textbf{k} = \left(k_1,\cdots,k_5\right)$ we then set
\[\tilde{\mathfrak{D}}^{\textbf{k}} \doteq \tilde{\mathfrak{D}}^{k_1}_1\cdots \tilde{\mathfrak{D}}^{k_5}_5.\]

Now we are ready to define our main energy norm.
\begin{definition}For any sufficiently regular function $f : \mathcal{M}_K \to \mathbb{C}$ and integer $k \geq 0$ we define
\[E_k\left[f\right]\left(\tau\right) \doteq \sum_{|\textbf{k}| \leq k}\int_{\Sigma_{\tau}}\left[\left|L\tilde{\mathfrak{D}}^{\textbf{k}}f\right|^2 + \left|\underline{L}\tilde{\mathfrak{D}}^{\textbf{k}}f\right|^2 + \left|\slashed{\nabla} \tilde{\mathfrak{D}}^{\textbf{k}}f\right|^2\right].\]
For any subset $\mathscr{A} \subset [r_+,\infty)$we also define 
\[E_{k, \mathscr{A}}\left[f\right]\left(\tau\right) \doteq \sum_{|\textbf{k}| \leq k}\int_{\Sigma_{\tau} \cap \{r \in \mathscr{A}\}}\left[\left|L\tilde{\mathfrak{D}}^{\textbf{k}}f\right|^2 + \left|\underline{L}\tilde{\mathfrak{D}}^{\textbf{k}}f\right|^2 + \left|\slashed{\nabla} \tilde{\mathfrak{D}}^{\textbf{k}}f\right|^2\right].\]
\end{definition}
\subsection{Currents}
We review here various current constructions that will play an important role in this paper. 
\subsubsection{Standard Currents}\label{stancurr}
Let $\left(\mathcal{M},g\right)$ be an oriented Lorentzian manifold with the corresponding Levi-Civita connection $\nabla$ and Hodge star operator $*$. Unless said otherwise, in the rest of this subsection all objects are assumed to be defined on $\mathcal{M}$ and real valued. Let $V$ denote a vector field, $w$ a function, $q$ a $1$-form, and $\psi$ a complex-valued function. We then define an associated symmetric $(0,2)$-tensor $\mathbf{T}_{\mu\nu}\left[g,\psi\right]$, $1$-form $\mathbf{J}^{V,w,q}\left[g,\psi\right]$, function $\mathbf{K}^{V,w,q}\left[g,\psi\right]$, and function $\mathbf{H}^{V,w}\left[\psi\right]$ by
\begin{equation}\label{energymomten}
\mathbf{T}_{\mu\nu}\left[g, \psi\right] \doteq {\rm Re}\left(\partial_{\mu}\psi\overline{\partial_{\nu}\psi}\right) - \frac{1}{2}g_{\mu\nu}\partial^{\gamma}\psi \partial_{\gamma}\psi,
\end{equation}
\begin{equation}\label{twistedcurdef}
\mathbf{J}^{V,w, q}_{\mu}\left[g, \psi\right] \doteq  \mathbf{T}_{\mu\nu}\left[g,\psi\right]V^{\nu} + w {\rm Re}\left(\psi\overline{\partial_{\mu}\psi}\right) + q_{\mu}\left|\psi\right|^2,
\end{equation}
\begin{equation}\label{theKdef}
\mathbf{K}^{V,w,q}\left[g,\psi\right] \doteq \pi^V_{\mu\nu}\left[g\right]\mathbf{T}^{\mu\nu}\left[g,\psi\right] + \left(\nabla^{\mu}w\right){\rm Re}\left(\psi\overline{\partial_{\mu}\psi}\right) + w\partial^{\gamma}\psi \partial_{\gamma}\psi + \left(\nabla^{\mu}q_{\mu}\right)\left|\psi\right|^2 + 2q^{\mu}{\rm Re}\left(\psi\overline{\partial_{\mu}\psi}\right),
\end{equation}
\begin{equation}\label{theHdef}
\mathbf{H}^{V,w}\left[\psi\right] \doteq V^{\mu}\partial_{\mu}\psi + w\psi.
\end{equation}
Here $\pi_{\mu\nu}^V\left[g\right]$ corresponds to $1/2$ times the deformation tensor of $V$:
\[\pi_{\mu\nu}^V\left[g\right] \doteq \frac{1}{2}\left(\mathcal{L}_Vg\right)_{\mu\nu} = \frac{1}{2}\left(\nabla_{\mu}V_{\nu} + \nabla_{\nu}V_{\mu}\right).\]
We will often omit various of the arguments of these objects when they are clear from context. In the case when $w$ and $q$ vanish we will write $\mathbf{J}^V_{\mu} = \mathbf{J}^{V,0,0}_{\mu}$, $\mathbf{K}^V = \mathbf{K}^{V,0,0}$, and $\mathbf{H}^V = \mathbf{H}^{V,0}$.

These various quantities are linked through the following divergence identity:
\begin{equation}\label{divIdentity}
\nabla^{\mu}\mathbf{J}_{\mu}  = \mathbf{K} + {\rm Re}\left(\mathbf{H}\overline{\Box_g\psi}\right).
\end{equation}
The divergence theorem then yields corresponding integral identities.
\subsubsection{Twisted currents on Kerr spacetimes}
Certain parts of our arguments will rely on some estimates established in~\cite{blackboxlargea}. These estimates are phrased in terms of so-called ``twisted'' currents. We review the corresponding formalism here. There currents will all depend on the parameter $a$ of the underlying Kerr spacetime we are studying. Following~\cite{holzegelWarnick} (see also Section A.2 of~\cite{blackboxsmalla}), we define the twisted gradient operator $\tilde{\nabla}$  by 

\[ \tilde{\nabla}f \doteq \beta\nabla\left(\beta^{-1} f\right), \qquad\tilde{\nabla}^{\dagger}f = -\beta^{-1}\nabla\left(\beta f\right) \qquad \beta \doteq \left(r^2+a^2\right)^{-1/2}.\]
It is useful to introduce
\begin{equation}\label{thisismathcalV}
\mathcal{V} \doteq  -\frac{\Box_g \beta}{\beta} = \frac{1}{\rho^2}\frac{2Mr^3+a^2r\left(r-4M\right) + a^4}{(r^2+a^2)^2}.
\end{equation}
We note (see (A.9) in ~\cite{blackboxsmalla}) that
\begin{equation}\label{lowerboundmathcalV}
\mathcal{V}\gtrsim r^{-3}.
\end{equation}
We then have the twisted modifications of the various quantities defined above:
\begin{equation}\label{energymomtentwist}
\tilde{\mathbf{T}}_{\mu\nu}\left[g, \psi\right] \doteq {\rm Re}\left(\tilde{\nabla}_{\mu}\psi\overline{\tilde{\nabla}_{\nu}\psi}\right) - \frac{1}{2}g_{\mu\nu}\left[\left(\tilde{\nabla}^{\alpha}\psi\right)\left(\tilde{\nabla}_{\alpha}\psi\right) + \mathcal{V}\left|\psi\right|^2\right].
\end{equation}
\begin{equation}\label{twistedcurdeftwist}
\tilde{\mathbf{J}}^{V,w}_{\mu}\left[g, \psi\right] \doteq  \mathbf{T}_{\mu\nu}\left[g,\psi\right]V^{\nu} + w {\rm Re}\left(\psi\overline{\tilde{\nabla}_{\mu}\psi}\right) -\frac{1}{2}\left(\nabla_{\mu}w\right) \left|\psi\right|^2,
\end{equation}
\begin{equation}\label{theKdeftwist}\begin{split}
\tilde{\mathbf{K}}^{V,w}\left[g,\psi\right] \doteq \pi^V_{\mu\nu}\left[g\right]\tilde{\mathbf{T}}^{\mu\nu}\left[g,\psi\right] + V^{\mu}\tilde{S}_{\mu}\left[\psi\right]+w{\rm Re}\left(\tilde{\nabla}^{\mu}\psi\overline{\tilde{\nabla}_{\mu}\psi}\right)+\left(\beta^{-1}\psi\right)^2\left(-\frac{1}{2}\nabla^{\mu}\left(\beta^2\nabla_{\mu}w\right) + \mathcal{V}w\beta^2\right),
\end{split}\end{equation}
\[\tilde{S}_{\mu} \left[g,\psi\right] \doteq \frac{\tilde{\nabla}^{\dagger}\left(\beta\mathcal{V}\right)}{2\beta}\left|\psi\right|^2 + \frac{\tilde{\nabla}^{\dagger}\beta}{2\beta}\tilde{\nabla}^{\nu}\psi\overline{\tilde{\nabla}_{\nu}\psi},\]
\begin{equation}\label{theHdeftwist}
\tilde{\mathbf{H}}^{V,w}\left[\psi\right] \doteq V^{\mu}\tilde{\nabla}_{\mu}\psi + w\psi.
\end{equation}
Once again, we will omit various arguments when clear from context.

Finally we have the twisted version of the fundamental divergence identity:
\begin{equation}\label{divIdentitytwist}
\nabla^{\mu}\tilde{\mathbf{J}}_{\mu}  = \tilde{\mathbf{K}} + {\rm Re}\left(\tilde{\mathbf{H}}\overline{\Box_g\psi}\right).
\end{equation}

\subsection{$(t^*,\phi^*)$-Fourier analysis}\label{FAS}
In this section we establish some conventions for Fourier analysis with respect to the $(t^*,\phi^*)$ of the Kerr star coordinates (cf.~Section 3.2 of~\cite{blackboxlargea}). This will also naturally lead to a notion of Fourier analysis with respect to the $(t,\phi)$ of Boyer--Lindquist coordinates in their domain of validity.

We let
\begin{equation}\label{defmathbbH}
\mathbb{H} \doteq \left\{ f\left(t^*,\phi^*,\theta\right) \in C^{\infty}\left(\mathbb{R}\times \mathbb{S}^2 \to \mathbb{C}\right) :  \int_{\mathbb{R} \times \mathbb{S}^2}\left|\partial^i_{t^*}\mathring{\nabla}^jf\right|^2\, dt d\mathbb{S}^2 < \infty, \forall (i,j) \in \mathbb{N}\times \mathbb{N}\right\},
\end{equation}
where $\mathring{\nabla}$ denotes the covariant derivative with respect to the round metric on $\mathbb{S}^2$. For any $f \in \mathbb{H}$ we may define the corresponding $(t^*,\phi^*)$-Fourier transform $\mathcal{F}\left[f\right]\left(\omega,m,\theta\right) : \mathbb{R} \times \mathbb{Z} \times (0,\pi)\to \mathbb{C}$ by
\begin{equation}\label{FTdef}
\mathcal{F}\left[f\right]\left(\omega,m,\theta\right) \doteq \int_{\mathbb{R}}\int_0^{2\pi}e^{i\omega t^*}e^{-im\phi^*}f\left(t^*,\phi^*,\theta\right)\, d\phi^*dt^*.
\end{equation}
We have the corresponding inversion formula:
\begin{equation}\label{FTinv}
f\left(t,\phi,\theta\right) = \mathcal{F}^{-1}\left[\mathcal{F}\left[f\right]\right]\left(t^*,\phi^*,\theta\right) \doteq \frac{1}{(2\pi)^2}\int_{\mathbb{R}}\sum_{m \in \mathbb{Z}}e^{-i\omega t^*}e^{im\phi^*}\mathcal{F}\left[f\right]\left(\omega,m,\theta\right)\, d\omega.
\end{equation}

We say that a function $\psi : \mathcal{M}_K \to \mathbb{C}$ lies in $\mathbb{H}_K$ if  for every  bounded interval $I \subset [r_+,\infty)$ and $N \in \mathbb{N}$, we have that 
\[\sum_{i+j+k \leq N}\int_{(t^*,r,\phi^*,\theta) \in \mathbb{R} \times I \times \mathbb{S}^2}\left|\partial_{t^*}^i\partial_r^j\mathring{\nabla}^kf\right|^2\, dtdrd\mathbb{S}^2 < \infty,\]
where $\mathring{\nabla}$ denotes the standard round spherical gradient acting along each $\mathbb{S}^2$. Note that $\psi \in \mathbb{H}_K$ implies that the restriction of $\psi$ to each curve of constant $r \in [r_+,\infty)$ lies in the space $\mathbb{H}$.

\begin{remark}\label{endecay}
Using Sobolev embeddings, we immediately see that, for any compact $K \subset [r_+,\infty)$,  $n\geq 1$, and $f\in \mathbb{H}_K$, we have
\begin{equation}
\lim_{t^*\to { \pm}\infty}E_{n,K}\left[f\right]\left(t^*\right) = 0.
\end{equation}
\end{remark}

We thus obtain a notion of the Fourier transform for such $\psi$:
\begin{equation}\label{FTKerrstar}
\mathcal{F}\left[\psi\right]\left(\omega,m,r,\theta\right) \doteq \int_{\mathbb{R}}\int_0^{2\pi}e^{i\omega t^*}e^{-im\phi^*}\psi\left(t^*,r,\phi^*,\theta\right)\, dt^*\, d\phi^*.
\end{equation}
For all $r > r_+$, we can also define a Fourier transform with respect to the Boyer--Lindquist coordinates $(t,\phi)$:
\begin{equation}\label{FTBL}
\mathcal{F}_{BL}\left[\psi\right]\left(\omega,m,r,\theta\right) \doteq \int_{\mathbb{R}}\int_0^{2\pi}e^{i\omega t}e^{-im\phi}\psi\left(t^*\left(t,r\right),r,\phi^*\left(\phi,r\right),\theta\right)\, dt\, d\phi.
\end{equation}
We have the following easily established relation between the two Fourier transforms:
\begin{equation}\label{relatedthetwoFTs}
\mathcal{F} = e^{i\omega \bar{t}(r)}e^{-i m \overline{\phi}(r)}\mathcal{F}_{BL},\qquad \mathcal{F}^{-1} = e^{-i\omega \bar{t}(r)}e^{i m \overline{\phi}(r)}\mathcal{F}^{-1}_{BL}
\end{equation}

For any smooth bounded function $a\left(\omega,m\right) : \mathbb{R} \times \mathbb{Z} \to \mathbb{C}$, we obtain a corresponding operator $Q_a : \mathbb{H}_K \to \mathbb{H}_K$ defined by
\begin{equation}\label{defQa}
Q_a\left[\psi\right] \doteq \mathcal{F}^{-1}\left[ a \mathcal{F}\left[\psi\right]\right]. 
\end{equation}
When $r > r_+$, in view of~\eqref{relatedthetwoFTs}, the operator $Q_a$ would be unchanged if we replaced $\mathcal{F}$ with $\mathcal{F}_{BL}$. We now use~\eqref{defQa} to define an operator $\mathcal{P}_{HT}: \mathbb{H}_K \to \mathbb{H}_K$. Heuristically speaking, $\mathcal{P}_{HT}$ selects a high frequency neighborhood of the trapped set in the case $m\neq 0$, while excluding all superradiant frequencies. The case $m=0$ can be treated by a separate argument.
\begin{definition}\label{hightrappedpart}Let $\delta_1 > 0$ and $\delta_2 > 0$ be small constants, and let $A_{\rm high} > 0$ be a constant which is sufficiently large depending on $\delta_1$ and $\delta_2$. Then let $q(x)$ is a smooth function which is identically $1$ for $|x| \leq 1$ and identically $0$ for $|x| \geq 2$, $p\left(x\right)$ be a smooth function which vanishes for $x \leq 1$ and is identically $1$ for $x \geq 2$, and $q_{\delta_1}(x)$ be a smooth function which is $1$ when $x \geq \frac{a}{2M r_+} + \delta_1$ and $0$ when $x \leq \frac{a}{2Mr_+} + (1/2)\delta_1$. Then we define an operator $\mathcal{P}_{HT}$ by
\[\mathcal{P}_{HT} \doteq Q_{a_{HT}},\]
\[a_{HT}\left(\omega,m\right) \doteq \mathbf{1}_{\left\{m \neq 0\right\}}\left(1- q\left(\frac{\omega^2+m^2}{A^2_{\rm high}}\right)\right)\left(p\left(\frac{|\omega|}{\delta_2 |m|}\right)1_{\omega m < 0} + q_{\delta_1}\left(\frac{\omega}{m}\right)\right) 
,\]
where, for $\mathcal{A} \subset \mathbb{Z}$, $\mathbf{1}_{\mathcal{A}}$ denotes the indicator function for $\mathcal{A}$.
\end{definition}

Our main results will concern solutions to the wave equation $\Box_{g_K}\psi = F$ which are defined to the future of the hypersurface $\Sigma_0$. Since such $\psi$ are not a priori defined on all of $\mathcal{M}_K$, $\mathcal{P}_{HT}\psi$ is also not well-defined a priori. To remedy this, we may use the main results of~\cite{waveKerrlargea} to define a suitable extension operator of $\psi$ to all of $\mathcal{M}_K$:
\begin{theorem}\label{extendtoeverywhere}\cite{waveKerrlargea} Let $\psi : J^+\left(\Sigma_0\right) \to \mathbb{C}$ solve~\eqref{waveEqnKerr},
satisfy that $\left(n_{\Sigma_0}\psi,\psi\right)|_{\Sigma_0}$ is smooth and compactly supported, and assume that $F$ is smooth and  compactly supported on $J^{+}(\Sigma_1)$. Then there exist functions $\mathscr{E}\psi \in \mathbb{H}_K$ and $\tilde F$ so that 
\begin{enumerate}
    \item The wave equation
    \[\Box_{g_K}\mathscr{E}\psi = F + \tilde{F}\]
    holds on all of $\mathcal{M}_K$.
    \item\label{vanishclosehornegtstar} For sufficiently negative $\tilde{t}^*$ we have that $\mathcal{E}\psi|_{t^* = \tilde{t}^*}$ is supported in $r \geq r_+ + \delta_0$.
    \item The function $\tilde{F}$ satisfies the following properties:
    \begin{enumerate}
        \item\label{fa} $\tilde{F}$ is supported to the past of $\Sigma_{-3}$ and within the region $\{r \leq r_+ + 2\delta_0\}$.
        \item\label{fb} The intersection of the support of $\tilde{F}$ and $\{r \leq r_+ + \delta_0\}$ is contained in the future of $\Sigma_{-4}$. 
        \item\label{fc} For every integer $k \geq 0$, $\tilde{F}$ satisfies that 
        \[\sum_{|\textbf{k}| \leq k}\sup_{\tau < -3}\int_{\Sigma_{\tau}}\left|\tilde{\mathfrak{D}}^{\textbf{k}}\tilde{F}\right|^2 + \int_{-\infty}^{-3}\int_{\Sigma_{\tau}}\left|\tilde{\mathfrak{D}}^{\textbf{k}}\tilde{F}\right|^2\, d\tau \lesssim_k E_k\left[\psi\right]\left(0\right).\]
    \end{enumerate}
    \item For every integer $k \geq 0$, we have that
    \begin{enumerate}
        \item \[\sup_{\tau < 0}E_k\left[ \mathscr{E}\psi\right]\left(\tau\right) \lesssim_k E_k\left[\psi\right]\left(0\right) .\]
        \item For all sufficiently large $\tilde{R}$, independent of $\psi$, we have
        \[\int_{-\infty}^0E_{k,{\tilde{R} \leq r \leq 100 \tilde{R}}}\left[ \mathscr{E}\psi\right](t)\, dt \lesssim_{k,\tilde{R}} E_k\left[\psi\right]\left(0\right).\]
    \end{enumerate}

\end{enumerate}
\end{theorem}
\begin{proof}We first discuss $k = 0$. Let $\tilde{\Sigma}$ be any Cauchy hypersurface for the Kerr exterior spacetime (such as the union of $\{t = 0\}$ and the bifurcation sphere). Then the extension procedure described in Section 13.1.1 of~\cite{waveKerrlargea} allows us to construct a solution $\tilde{\psi}$  to \eqref{waveEqnKerr} on the whole Kerr exterior spacetime so that 
\[\left\vert\left\vert \nabla_{\tilde{\Sigma}}\tilde\psi\right\vert\right\vert^2_{L^2\left(\tilde\Sigma\right)} + \left\vert\left\vert n_{\tilde\Sigma}\tilde\psi\right\vert\right\vert^2_{L^2\left(\tilde\Sigma\right)} \lesssim \left\vert\left\vert \nabla_{\Sigma_0}\psi\right\vert\right\vert^2_{L^2\left(\Sigma_0\right)} + \left\vert\left\vert n_{\Sigma_0}\psi\right\vert\right\vert^2_{L^2\left(\Sigma_0\right)}.\]
Then we let $\tilde{\chi}$ be a suitable cut-off function which is $0$ in $J^-\left(\Sigma_{-4}\right)\cap \{r \leq r_++2\delta_0\}$ and is identically $1$ for $J^+\left(\Sigma_{-3}\right) \cup\{r \geq r_++3\delta_0\}$. Set $\mathscr{E}\psi \doteq \tilde{\chi}\tilde{\psi}$. Then the desired estimates for $\mathscr{E}\psi$ and $\tilde{F}$ are an immediate consequence of the main energy boundedness and local integrated energy decay results from~\cite{waveKerrlargea} and the use of a Hardy inequality to control $r^{-1}\psi$ in $L^2$. Finally, the higher order statements follow in a similar fashion using the corresponding higher order boundedness results and integrated energy decay results from~\cite{waveKerrlargea}.
\end{proof}
\begin{remark}\label{timetranslatetheextension} It will be useful later to observe that if we let $\tau > 0$ and apply the time translation map $t \mapsto t + \tau$ to Theorem~\ref{extendtoeverywhere}, all of the implied constants will be independent of $\tau$.
\end{remark}
\begin{remark}Item~\ref{vanishclosehornegtstar} in Theorem~\ref{extendtoeverywhere} will be technically convenient for us because it allows us to use the energies $E_k$ which are defined along constant $t^*$ hypersurfaces as our basic norm to control the solution even though the foliation by constant $t^*$ hypersurfaces degenerates near the bifurcation sphere and the past event horizon. 
\end{remark}
\begin{convention}\label{rlowrhighconv} In the context of Definition~\ref{themetricclass}, we will  take $r_++(3/2)\delta_0 = r_{low}$, so that, in particular, the restriction of $\tilde{F}$ to $J^-(\Sigma_{-4})$ is supported in the region $r \in (r_{\rm low},R_{\rm high})$ 
\end{convention}

With the extension operator $\mathscr{E}$ from Theorem~\ref{extendtoeverywhere}, we may now extend the class of functions to which we can apply $\mathcal{P}_{HT}$.
\begin{definition}\label{phtwithtehextend}Let $\psi : J^+\left(\Sigma_0\right) \to \mathbb{C}$ solve~\eqref{waveEqnKerr},
satisfies that $\left(n_{\Sigma_0}\psi,\psi\right)|_{\Sigma_0}$ is smooth and compactly supported, and has that $F$ is smooth compactly supported to the future of $\Sigma_1$. Then, using the extension operator $\mathscr{E}$ from Theorem~\ref{extendtoeverywhere}, we define 
\begin{equation}
\mathcal{P}_{HT}\psi \doteq \mathcal{P}_{HT}\left(\mathscr{E}\psi\right).
\end{equation}
\end{definition}

\subsection{The interpolating operator}
Let $g \in \mathscr{A}_{\epsilon,,r_{\rm low},R_{\rm high}}$ as in Definition~\ref{themetricclass}. When we study to solutions to $\Box_g\psi = 0$, we will not know \emph{a priori} that solution $\psi$  lies in the space $\mathbb{H}_K$. It will thus be useful to introduce an ``interpolating operator'' which essentially plays the role of the  interpolating metric from Definition 11.2.2 of~\cite{waveKerrlargea}.
\begin{definition}\label{definterpolate}Let $g \in \mathscr{A}_{\epsilon,,r_{\rm low},R_{\rm high}}$  and let $\tau \geq 10$. Then let $\eta_{\tau} : \mathcal{M}_K \to \mathbb{R}$ be a cut-off function which identically $1$ on $J^+\left(\Sigma_0\right) \cap J^-\left(\Sigma_{\tau}\right)$, identically $0$ on $J^-\left(\Sigma_{-1}\right) \cup J^+\left(\Sigma_{\tau+1}\right)$, $\eta$ depends only on $t^*$, and $\eta$ satisfies $\left\vert\left\vert \eta\right\vert\right\vert_{C^N} \lesssim_N 1$. Then we define the \emph{interpolating operator} $\Box_{\tau}$ by setting
\[\Box_{\tau} \doteq \eta_{\tau}\Box_g + \left(1-\eta_{\tau}\right)\Box_{g_K}.\]
It will also be useful to define a cut-off $\tilde{\eta}_{\tau}$ which is identically $1$ on $J^+\left(\Sigma_1\right) \cap J^-\left(\Sigma_{\tau-1}\right)$ and so that ${\rm supp}\left(\tilde{\eta}_{\tau}\right) \subset \{t^* : \eta_{\tau}(t^*) = 1\}.$
\end{definition}
In view of the $C^1$ closeness of $g$ and $g_K$, it is clear that each $\Sigma_{\tau}$ is spacelike for $g$, $g_K$, and the principal symbol of $\Box_{\tau}$. Moreover the sets $J^{\pm}\left(\Sigma_{\tau}\right)$ are all the same for $g$, $g_K$, or the principal symbol of $\Box_{\tau}$. Also, for solutions $\psi$ to $\Box_{\tau}\psi 
 = 0$ it is immediate that we may define an extension operator $\mathscr{E}$ as in Theorem~\ref{extendtoeverywhere}.
\begin{corollary}\label{corextend}Theorem~\ref{extendtoeverywhere} holds if we replace $\Box_{g_K}$ everywhere by $\Box_{\tau}$. The constants in the various estimates do not depend on $\tau$.
\end{corollary}

We will generally work with the PDE $\Box_{\tau}$ instead of $\Box_g$. The following definition will be convenient for this.
\begin{definition}\label{defpsitau}Let $\psi$ be as in the statement of Theorem~\ref{theoc1pert}, and choose $\tau \geq 10$. Then we define $\tilde{\psi}_{\tau}$ to solve $\Box_{\tau}\tilde{\psi}_{\tau} = 0$ so that $\left(n_{\Sigma_0}\tilde{\psi}_{\tau},\psi_{\tau}\right)|_{\Sigma_0} = \left(n_{\Sigma_0}\psi,\psi\right)|_{\Sigma_0}$. Finally, we set $\psi_{\tau} \doteq \mathscr{E}\tilde{\psi}_{\tau}$, where $\mathscr{E}$ is the extension operator from Corollary~\ref{corextend}.
\end{definition}

\subsubsection{A useful identity for integration by parts}
We will often use the following result:

\begin{lemma}\label{thatscoolintbyparts}Let $\{x^i\}_{i=0}^3$ denote any set of local coordinates on a coordinate chart in $\mathcal{M}_K$. Let $i,j,k \in \{0,1,2,3\}$ be not necessarily distinct, and let $\partial_i$, $\partial_j$, and $\partial_k$ denote the corresponding coordinate derivatives. Then, for any function $f$, defined on the corresponding coordinate chart, we have
\[{\rm Re}\left(\left(\overline{\partial_i\partial_jf}\right)\left(\partial_kf\right)\right) = \frac{1}{2}\left[\partial_i\left({\rm Re}\left(\overline{\partial_jf}\partial_kf\right)\right) - \partial_k\left({\rm Re}\left(\overline{\partial_jf}\partial_if\right)\right) + \partial_j\left({\rm Re}\left(\overline{\partial_kf}\partial_if\right)\right)\right] \]
 In particular, if $V$ is a smooth vector field, $Q=\Box_{g_K}-\Box_{\tau}$, where $g$ is a metric satisfying Definition~\ref{themetricclass}, and $f\in \mathbb{H_K}$, 
we have
\[
\left|\int_{\mathcal{M}_K}{\rm Re}\left((Vf)(\overline{Qf})\right) \right|\lesssim \epsilon\int_{-\infty}^{\infty} E_{0, r\in [r_{low}, R_{high}]}[f](s) ds
\]
\end{lemma}
\begin{proof}
The first statement is an easily verified identity. The second follows from applying the divergence theorem, combined with Remark~\ref{endecay} and conditions (2) and (3) in Definition~\ref{themetricclass}.
\end{proof}

\section{Pseudo-differential operator estimates}\label{PDOstuff}
In this section we will collect some estimates which are a straightforward consequence of the theory of $(t^*,\phi^*)$ pseudo-differential operators. We first review the definition of a (suitable class of) $(t^*,\phi^*)$ pseudo-differential operator. The following result is standard.
\begin{theorem}\label{semitheo}Let $k$ be a non-positive integer and $A_{\rm high}$ be a sufficiently large positive constant. Suppose that $o\left(t,\omega,m\right) : \mathbb{R} \times \mathbb{R} \times \mathbb{Z} \to \mathbb{C}$ is smooth in $t$ and $\omega$ and satisfies for each $i,j \in \mathbb{Z}_{\geq 0}$ that
\begin{equation}\label{thesynbolassum}
\left|\partial_{t^*}^i\partial_{\omega}^jo\right| \lesssim_{i,j} \left(1+|\omega|+|m|+A_{\rm high}\right)^{k-j}.
\end{equation}

Then we may define a unique  linear operator $\mathcal{O} : \mathbb{H}_K \to \mathbb{H}_K$  which when acting on smooth functions $f$ which are compactly supported in $t^*$ is defined by
\begin{equation}\label{theooperator}
\mathcal{O}f\left(t^*,\phi^*,r,\theta\right) = \frac{1}{(2\pi)^2}\int_{\mathbb{R}}\sum_{m \in \mathbb{Z}}o\left(t^*,\omega,m\right)\mathcal{F}\left[f\right]\left(\omega,m,r,\theta\right) e^{-i\omega t^*}e^{im\phi^*}\, d\omega,
\end{equation}
and so that for each fixed $r$,
\begin{equation}\label{L2bound}
\left\vert\left\vert \mathcal{O}f\right\vert\right\vert_{L^2\left(\mathbb{R}_{t^*}\times \mathbb{S}^2\right)} \lesssim  \left\vert\left\vert f\right\vert\right\vert_{L^2\left(\mathbb{R}_{t^*}\times \mathbb{S}^2\right)}.
\end{equation}
Thus, for each fixed $r$, we may consider $\mathcal{O}$ acting on $L^2\left(\mathbb{R}_{t^*}\times\mathbb{S}^2\right)$.

The integer $k$ is the ``order'' of the operator $\mathcal{O}$.

We can also represent $O$ as a singular integral (for fixed $r$ and $\theta$):
\[
\mathcal{O}f\left(t^*,\phi^*,r,\theta\right) = \frac{1}{(2\pi)^2} \int_{\mathbb{R}} \int_0^{2\pi} K(t^*, \phi^*, t', \phi') f(t^*, \phi^*, r, \theta) d\phi' dt'
\]
where
\begin{equation}\label{kernel}
K(t^*, \phi^*, t', \phi') = \int_{\mathbb{R}} \sum_{m\in\mathbb{Z}} o(t^*, \omega, m) e^{i\omega(t'-t^*)} e^{im(\phi^*-\phi')} d\omega
\end{equation}

If $\mathcal{O}_1$ is order $k_1$ and $\mathcal{O}_2$ is order $k_2$, then $\left[\mathcal{O}_1,\mathcal{O}_2\right]$ will be an operator of the form~\eqref{theooperator} and will be of order $k_1+k_2-1$.
\end{theorem}

\begin{remark}\label{fromttotstar}It is sometimes more convenient to work in Boyer--Lindquist coordinates, which is defined similarly, using $\mathcal{F}_{BL}$ instead of $\mathcal{F}$. In view of~\eqref{relatedthetwoFTs} it is straightforward to translate results between Boyer--Lindquist and Kerr star coordinates.
\end{remark}

We now introduce a notation for weighted $L^2$ spaces on $\{r \geq r_+\}\cap \Sigma_{\tau}$.
\begin{definition}For any  non-negative $w\left(r,\theta\right) : [r_+,\infty) \times (0,\pi)$, we let $\mathcal{E}_{w,t^*}$ denote $L^2$ along $\Sigma_{t^*}$ with respect to the volume form multiplied by $w(r,\theta)dr d\mathbb{S}^2$. We also define the space $\mathcal{P}_w$ to denote $L^2$ along $\mathbb{R}_{t^*}\times \{r \geq r_+\} \times \mathbb{S}^2$ with respect to the volume form $w(r,\theta)dt^*drd\mathbb{S}^2$.
\end{definition}
\begin{convention}We introduce the convention that unless said otherwise, $w$ is some non-negative function defined on $(r,\theta) \in [r_+,\infty) \times (0,\pi)$.
\end{convention}
The following corollary is an immediate consequence of~\eqref{L2bound}.
\begin{corollary}\label{OonPw}Any $(t^*,\phi^*)$ pseudo-differential operator $\mathcal{O}$ as defined by Theorem~\ref{semitheo} extends to a bounded operator on $\mathcal{P}_w$.
\end{corollary}

The following remark will also be useful.
\begin{remark}\label{Planc}
Since $w$ is independent of $t^*$ and $\phi^*$, we obtain the Plancherel formula
\begin{equation}\label{Plancherel}
 \int_{-\infty}^{\infty}\| f\|_{\mathcal{E}_{w,t^*}}^2  dt^* = \int_{\mathbb{R}}\int_{r_+}^{\infty}\int_0^\pi\sum_{m \in \mathbb{Z}}|\mathcal{F} f|^2 w(r,\theta) d\theta dr d\omega
\end{equation}

As a consequence, if $o=o(\omega, m)$ is a multiplier, we have
\begin{equation}\label{mult}
\int_{-\infty}^{\infty}\| Of\|_{\mathcal{E}_{w,t^*}}^2 dt^*\leq \|o\|^2_{L^{\infty}} \int_{-\infty}^{\infty}\|f\|^2_{\mathcal{E}_{w,t^*}} dt^*.
\end{equation}
\end{remark}

We  next introduce a higher order version of $\mathcal{E}_{w,t^*}$ norm.
\begin{definition}Let $n$ be a non-negative integer and $f \in \mathbb{H}_K$. We then set
\[\left\vert\left\vert f\right\vert\right\vert_{\mathcal{E}_{w,t^*,n}} \doteq \sum_{0 \leq i+j+k \leq n}\left\vert\left\vert \ \left|\partial_{t^*}^i\partial_r^j\mathring{\nabla}^kf\right| \ \right\vert\right\vert_{\mathcal{E}_{w,t^*}},\]
where $\mathring{\nabla}$ denotes the standard round spherical gradient acting along each $\mathbb{S}^2$.

We then let $\mathcal{E}_{w,t^*,n}$ denote the corresponding completion of smooth functions under the norm $\left\vert\left\vert \cdot \right\vert\right\vert_{\mathcal{E}_{w,t^*,n}}$. 

\end{definition}

For various $(t^*,\phi^*)$ pseudo-differential operators $\mathcal{O}$ we will often establish estimates concerning the norm $\left\vert\left\vert \mathcal{O}f\right\vert\right\vert_{\mathcal{E}_{w,t^*}}$ and then desire to ``upgrade'' them to estimates involving $\left\vert\left\vert \mathcal{O}f\right\vert\right\vert_{\mathcal{E}_{w,t^*,n}}$. This is most straightforwardly done with a suitable set of differential operators which have good commutation properties with $\mathcal{O}$. The following lemma provides a convenient set of such operators.
\begin{lemma}\label{thisisusefultocommute}Let $\mathcal{O}$ be a $(t^*,\phi^*)$ pseudo-differential operator as defined in Theorem~\ref{semitheo}. We now introduce some useful operators with good commutation properties with $\mathcal{O}$:
\begin{enumerate}
\item Let $\mathring{\Delta} : \mathbb{H}_K \to \mathbb{H}_K$ denote the standard spherical Laplacian. We have $\left[\mathring{\Delta},\mathcal{O}\right] = \left[\mathring{\Delta},\partial_{t^*}\right] = \left[\mathring{\Delta},\partial_r\right] = 0$.
\item For any function $f \in \mathbb{H}_K$ we may consider $\partial_{\theta}f$ originally defined for $\theta \not\in \{0,\pi\}$. Then $\partial_{\theta}f$ extends to a function lying in $\mathcal{P}_w$. We moreover have that $\left[\partial_{\theta},\mathcal{O}\right] = \left[\partial_{\theta},\partial_{t^*}\right] = \left[\partial_{\theta},\partial_r\right]= 0$. 
\item For any function $f \in \mathbb{H}_K$ we may consider $(\sin^{-1}\theta)\partial_{\phi}f$ originally defined for $\theta \not\in \{0,\pi\}$. Then $(\sin^{-1}\theta)\partial_{\phi}f$ extends to a function lying in $\mathcal{P}_w$. We moreover have that $\left[(\sin^{-1}\theta)\partial_{\phi},\mathcal{O}\right] = \left[(\sin^{-1}\theta)\partial_{\phi},\partial_{t^*}\right]= \left[(\sin^{-1}\theta)\partial_{\phi},\partial_r\right]=0$. 
\end{enumerate}

Along with $\partial_{t^*}$ we will now use the above differential operators to provide useful estimates for the norms $\left\vert\left\vert \cdot\right\vert\right\vert_{\mathcal{E}_{w,t^*,n}}$. For all $n\geq 1$ and $f \in \mathbb{H}_K$, we have that
    \begin{align}\label{comparegood}
   & \left\vert\left\vert f\right\vert\right\vert_{\mathcal{E}_{w,t^*,n}} \sim \sum_{i+j+2k \leq n}\left\vert\left\vert \partial_{t^*}^i\partial_r^j\mathring{\Delta}^kf\right\vert\right\vert_{\mathcal{E}_{w,t^*}} 
   \\ &\qquad\qquad \qquad\qquad + \sum_{i+j+2k \leq n-1}\left\vert\left\vert \partial_{\theta}\partial_{t^*}^i\partial_r^j\mathring{\Delta}^kf\right\vert\right\vert_{\mathcal{E}_{w,t^*}} + \sum_{i+j+2k \leq n-1}\left\vert\left\vert (\sin^{-1}\theta)\partial_{\phi}\partial_{t^*}^i\partial_r^j\mathring{\Delta}^kf\right\vert\right\vert_{\mathcal{E}_{w,t^*}}.
    \end{align}
\end{lemma}
\begin{proof}Since the symbol of $\mathcal{O}$ does not depend on $\phi$ or $\theta$ the commutation properties are immediately verified.

The remaining assertions of the lemma are immediate consequences of the following facts. First, we have that for any function $h$ on $\mathbb{S}^2$:
\[\left|\mathring{\nabla}h\right|^2 = \left|\partial_{\theta}h\right|^2 + (\sin^{-2}\theta)\left|\partial_{\phi}h\right|^2.\]
Second, in view of elliptic estimates along $\mathbb{S}^2$, we have that for any non-negative integer $k$:
\[\sum_{j\leq k}\left\vert\left\vert \mathring{\Delta}^jh\right\vert\right\vert_{L^2\left(\mathbb{S}^2\right)} \sim \sum_{j \leq 2k}\left\vert\left\vert \mathring{\nabla}^jh\right\vert\right\vert_{L^2\left(\mathbb{S}^2\right)},\ \sum_{j\leq k}\left\vert\left\vert \mathring{\nabla}\mathring{\Delta}^jh\right\vert\right\vert_{L^2\left(\mathbb{S}^2\right)}+\sum_{j\leq k}\left\vert\left\vert \mathring{\Delta}^jh\right\vert\right\vert_{L^2\left(\mathbb{S}^2\right)} \sim \sum_{j \leq 2k+1}\left\vert\left\vert \mathring{\nabla}^jh\right\vert\right\vert_{L^2\left(\mathbb{S}^2\right)}.\]
\end{proof}

The following result expresses the pseudo-locality of $(t^*,\phi^*)$ pseudo-differential operators.
\begin{theorem}\label{pseudoloc}Let $\mathcal{O}$ be a $(t^*,\phi^*)$ pseudo-differential operator of order $k \leq 0$ defined by~\eqref{theooperator}, and let $f \in \mathbb{H}_K$. Then for any $t^*_0 \in \mathbb{R}$, $s > 0$, and positive integer $N$, we have that
\begin{equation}\label{local}\int_{|t^*-t^*_0| \leq s}\left\vert\left\vert \mathcal{O}f\right\vert\right\vert_{\mathcal{E}_{w,t^*, n}}^2\, dt^* \lesssim_{s,N, n} A_{\rm high}^{2k} \int_{\mathbb{R}}\frac{\left\vert\left\vert f\right\vert\right\vert_{\mathcal{E}_{w,t^*, n}}^2}{\left(1+|t^*-t^*_0|\right)^N}\, dt^*. 
\end{equation}
In particular, we have
\begin{equation}\label{L2}
\int_{\mathbb{R}}\left\vert\left\vert \mathcal{O}f\right\vert\right\vert_{\mathcal{E}_{w,t^*,n}}^2\, dt^*  \lesssim_n A_{\rm high}^{2k}\int_{\mathbb{R}}\left\vert\left\vert f\right\vert\right\vert_{\mathcal{E}_{w,t^*, n}}^2\, dt^*,
\end{equation}
and, for any interval $I \subset \mathbb{R}$ with $|I| \gtrsim 1$ and $q > 0$,
\begin{equation}\label{astartbutcandobetter}
\int_I\left\vert\left\vert \mathcal{O}f\right\vert\right\vert^2_{\mathcal{E}_{w,t^*,n}}\, dt^* \lesssim_{q, n} A_{\rm high}^{2k}\left[\int_I \left\vert\left\vert f\right\vert\right\vert^2_{\mathcal{E}_{w,t^*, n}}\, dt^*  + \sup_{t^*} \int_{t^*-q}^{t^*+q}\left\vert\left\vert f\right\vert\right\vert^2_{\mathcal{E}_{w,t^*, n}}\, ds \right].
\end{equation}

We can replace the $\left\vert\left\vert \mathcal{O}f\right\vert\right\vert^2_{\mathcal{E}_{w,t^*,n}}$ in each of these estimates with $\sum_{i+j=0}^{-k}\left\vert\left\vert \mathcal{O}T^i\Phi^jf\right\vert\right\vert^2_{\mathcal{E}_{w,t^*,n}}$ if we drop the $A_{\rm high}^{2k}$ from the right hand side.
\end{theorem}
\begin{proof} When $n = 0$, this follows from a minor extension of the proof of Theorem 1 in Chapter VI of~\cite{BigStein} , using Remark~\ref{Planc} and the singular integral representation \eqref{kernel}.

For $n > 0$, we simply use Lemma~\ref{thisisusefultocommute} and the fact that $\left[\partial_{t^*},\mathcal{O}\right]$ is a $(t^*,\phi^*)$ pseudo-differential operator of order $k$.
\end{proof}

The next corollary will be convenient later.

\begin{corollary}\label{blahblah1231234}Let $\mathcal{O}$ be an operator of order $k \leq 0$ defined by~\eqref{theooperator} so that the symbol $o$ is a real valued function of $\omega$ and $m$. Let $I_1 \subset \mathbb{R}$ be a bounded interval, let $I_2 \subset [r_+,\infty)$ be an interval. Let $f,h \in \mathbb{H}_K$ satisfy that $\tilde{r} \in I_2$ implies that $h|_{r=\tilde{r}}$ is supported on $I_1$. Then, for any $\delta > 0$ and $q \gtrsim 1$, we have that
\begin{equation}\begin{split}
\left|\int_{\mathcal{M}_K \cap \{r \in I_2\}}{\rm Re}\left(\mathcal{O}f\overline{\mathcal{O}h}\right)\right| \lesssim_{|I_1|} \delta^{-1}\int_{\mathcal{M}_K \cap \{r \in I_2\}}\left|h\right|^2 + \delta \sup_{\tilde{t}^*} \int_{\mathcal{M}_K \cap \{r \in I_2\}\cap \{t^* \in [\tilde{t}^*-q,\tilde{t}^*+q]\} }\left|f\right|^2.
\end{split}
\end{equation}
\end{corollary}
\begin{proof}Since the symbol of $\mathcal{O}$ does not depend on $t^*$ and is real valued, it is a consequence of Plancherel's theorem that 
\[\int_{\mathcal{M}_K \cap \{r \in I_2\}}{\rm Re}\left(\mathcal{O}f\overline{\mathcal{O}h}\right) = \int_{\mathcal{M}_K \cap \{r \in I_2\}}{\rm Re}\left(\mathcal{O}^2f\overline{h}\right) = \int_{\mathcal{M}_K \cap \{r \in I_2\} \cap \{t^* \in I_1\}}{\rm Re}\left(\mathcal{O}^2f\overline{h}\right).\]
We then apply Cauchy--Schwarz, use~\eqref{astartbutcandobetter} and the fact that
\[
\int_{\mathcal{M}_K \cap \{r \in I_2\} \cap \{t^* \in I_1\}} |f|^2 \lesssim |I_1|  \sup_{\tilde{t}^*} \int_{\mathcal{M}_K \cap \{r \in I_2\}\cap \{t^* \in [\tilde{t}^*-q,\tilde{t}^*+q]\} }\left|f\right|^2.
\]

\end{proof}

This next result concerns the case when we estimate $\mathcal{O}f$ in a region disjoint from the support of $f$.
\begin{theorem}\label{offthesupportitsgood}Let $\mathcal{O}$ be a $(t^*,\phi^*)$ pseudo-differential operator defined by~\eqref{theooperator}, and let $f \in \mathbb{H}_K$. Let $I_1 \subset I_2 \subset \mathbb{R}$ be two intervals so that $\left|I_2\setminus \overline{I_1}\right| \gtrsim 1$  and so that ${\rm supp}\left(f\right) \subset \mathbb{R}\setminus I_2$. Then, for every $\tilde{t} \gtrsim 1$, and non-negative integers $n$, $m$ and $\tilde{N}$, we have
\begin{equation}\label{usethislater123123}
 \sum_{i+j \leq m}\int_{I_1}\left\vert\left\vert \mathcal{O}T^i\Phi^jf\right\vert\right\vert^2_{\mathcal{E}_{w,s,n}}\, ds \lesssim A_{\rm high}^{-2\tilde{N}}\sup_{t^*}\int_{t^*-\tilde{t}}^{t^*+\tilde{t}}\left\vert\left\vert f\right\vert\right\vert^2_{\mathcal{E}_{w,s,n}}\, ds,
\end{equation}
\begin{equation}\label{L1disj}
\sum_{i+j \leq m}\int_{I_1}\left\vert\left\vert \mathcal{O}T^i\Phi^jf\right\vert\right\vert_{\mathcal{E}_{w,s,n}}\, ds \lesssim A^{-\tilde{N}}_{\rm high} \sup_{t^*} \int_{t^*-\tilde{t}}^{t^*+\tilde{t}}\left\vert\left\vert f\right\vert\right\vert_{\mathcal{E}_{w,s,n}}\, ds.
\end{equation} 
\end{theorem}
\begin{proof} When $n = 0$ this is an immediate consequence of the singular integral realization of the operator $\mathcal{O}$. See Section 2.2 of Chapter VI in~\cite{BigStein}.

For $n > 0$, we simply use Lemma~\ref{thisisusefultocommute} and the fact that $\left[\partial_{t^*},\mathcal{O}\right]$ is also a $(t^*,\phi^*)$ pseudo-differential operator.
\end{proof}

We next note the following fact about the commutator of a $(t^*,\phi^*)$ pseudo-differential operator and a smooth function $h$.
\begin{lemma}\label{commisminus1}Let $\mathcal{O}$ be an operator of order $k \leq 0$ defined by~\eqref{theooperator},  $n$ be a non-negative integer, and let $h(t) \in C^{\infty}_c\left(\mathbb{R}\right)$.  Then $\left[\mathcal{O},h\right]$ is a $(t^*,\phi^*)$ pseudo-differential operator of order $k-1$. 

Furthermore, if we let $q \geq 1$, $\tilde{N}\lesssim 1$, and $I$ any open interval which contains the support of $h$ and satisfies $|I\setminus supp(h)| \gtrsim 1$, then we have
\begin{equation}\label{onethesupportw}
\int_{\mathbb{R}}\left\vert\left\vert \left[\mathcal{O},h\right]f \right\vert\right\vert_{\mathcal{E}_{w,s, n}}^2\, ds \lesssim A_{\rm high}^{-2}\int_I \left\vert\left\vert f\right\vert\right\vert^2_{\mathcal{E}_{w,s, n}}\, ds + A_{\rm high}^{-2\tilde{N}}\sup_{t^*}\int_{t^*-q}^{t^*+q}\left\vert\left\vert f\right\vert\right\vert^2_{\mathcal{E}_{w,s, n}}\, ds,
\end{equation}
\begin{equation}\label{blahbi2j}\begin{split}
&\int_{\mathbb{R}}h^2\left\vert\left\vert \mathcal{O}f \right\vert\right\vert_{\mathcal{E}_{w,s, n}}^2\, ds \lesssim \int_{\mathbb{R}}h^2\left\vert\left\vert f \right\vert\right\vert_{\mathcal{E}_{w,s,n}}^2\, ds
\\ &\qquad +A_{\rm high}^{-2}\int_I \left\vert\left\vert f\right\vert\right\vert^2_{\mathcal{E}_{w,s, n}}\, ds + A_{\rm high}^{-2\tilde{N}}\sup_{t^*}\int_{t^*-q}^{t^*+q}\left\vert\left\vert f\right\vert\right\vert^2_{\mathcal{E}_{w,s, n}}\, ds.
\end{split}\end{equation}
\end{lemma}
\begin{proof}The fact that $\left[\mathcal{O},h\right]$ is an order $k-1$ $(t^*,\phi^*)$ pseudo-differential operator is an immediate consequence of Theorem~\ref{semitheo} and the fact that multiplication by $h(t)$ is an order $0$ $(t^*,\phi^*)$ pseudo-differential operator.

Assume now that $n = 0$. Let $\tilde{I}$ be an open interval which contains the support of $h$ and is strictly included in $I$. Then let $\tilde{\chi}$ be a cut-off function which is supported in the interval $I$ and identically $1$ on the interval $\tilde{I}$. To establish~\eqref{onethesupportw}, we write  
\[\int_{\mathbb{R}}\left\vert\left\vert \left[\mathcal{O},h\right]f \right\vert\right\vert_{\mathcal{E}_{w,s}}^2\, ds = \underbrace{\int_{\tilde{I}}\left\vert\left\vert \left[\mathcal{O},h\right]{\tilde{\chi}f} \right\vert\right\vert_{\mathcal{E}_{w,s}}^2\, ds}_{\doteq I} +\underbrace{\int_{ \tilde{I}}\left\vert\left\vert \left[\mathcal{O},h\right]{\left(1-\tilde{\chi}\right)f} \right\vert\right\vert_{\mathcal{E}_{w,s}}^2\, ds}_{\doteq II}+\underbrace{\int_{\mathbb{R}\setminus \tilde{I}}\left\vert\left\vert \mathcal{O} (h f)\right\vert\right\vert_{\mathcal{E}_{w,s}}^2\, ds}_{\doteq III}.\]
We then estimate $I$ using \eqref{astartbutcandobetter} in Theorem~\ref{pseudoloc}, combined with the fact that
\[
\sup_{t^*}\int_{t^*-q}^{t^*+q}\left\vert\left\vert \tilde\chi f\right\vert\right\vert^2_{\mathcal{E}_{w,s}}\, ds \lesssim \int_I \left\vert\left\vert f\right\vert\right\vert^2_{\mathcal{E}_{w,s}}\, ds.
\]
Finally, we estimate $II$  and $III$ using Theorem~\ref{offthesupportitsgood}.

To prove~\eqref{blahbi2j} with $n=0$ we note 
\begin{equation}\label{toprovesmthwhoknows}
\int_{\mathbb{R}}h^2(s)\left\vert\left\vert \mathcal{O}f\right\vert\right\vert^2_{\mathcal{E}_{w,s}}\, ds \lesssim\int_{\mathbb{R}}\left\vert\left\vert \mathcal{O}\left(hf\right)\right\vert\right\vert^2_{\mathcal{E}_{w,s}}\, ds + \int_{\mathbb{R}}\left\vert\left\vert \left[h,\mathcal{O}\right]f\right\vert\right\vert^2_{\mathcal{E}_{w,s}}\, ds . 
\end{equation}
The result then follows from Theorem~\ref{pseudoloc} and~\eqref{onethesupportw}.

We now consider~\eqref{onethesupportw} in the case when $n > 0$. In view of Lemma~\ref{thisisusefultocommute} we only need to discuss how to commute our estimate with $\partial^j_{t^*}$. For $n=1$, it suffices to observe that
\[\partial_{t^*}\left[\mathcal{O},h\right]f  = \left[\mathcal{O},h\right]\partial_{t^*}f - \left[\mathcal{O},\partial_{t^*}h\right]f + \left[\left[\mathcal{O},\partial_{t^*}\right],h\right]f,\]
and we can apply~\eqref{onethesupportw} to each of these terms. The case of higher $n$ follows similarly. 

Finally, we consider~\eqref{blahbi2j} in the case when $n > 0$. When $n = 1$, the estimate follows by replacing~\eqref{toprovesmthwhoknows} with
\begin{equation}\label{toprovesmthwhoknows123}\begin{split}
&\int_{\mathbb{R}}h^2(s)\left\vert\left\vert \mathcal{O}f\right\vert\right\vert^2_{\mathcal{E}_{w,s,1}}\, ds \lesssim \int_{\mathbb{R}}\left(\partial_s h\right)^2\left\vert\left\vert \mathcal{O}f\right\vert\right\vert^2_{\mathcal{E}_{w,s}}\, ds 
+ \int_{\mathbb{R}}\left\vert\left\vert h\mathcal{O}f\right\vert\right\vert^2_{\mathcal{E}_{w,s,1}}\, ds
\lesssim \\ & \int_{\mathbb{R}}\left(\partial_s h\right)^2\left\vert\left\vert \mathcal{O}f\right\vert\right\vert^2_{\mathcal{E}_{w,s}}\, ds +\int_{\mathbb{R}}\left\vert\left\vert \mathcal{O}\left(hf\right)\right\vert\right\vert^2_{\mathcal{E}_{w,s,1}}\, ds + \int_{\mathbb{R}}\left\vert\left\vert \left[h,\mathcal{O}\right]f\right\vert\right\vert^2_{\mathcal{E}_{w,s,1}}\, ds, 
\end{split}\end{equation}
and then arguing as above. The case $n > 1$ follows similarly using induction.
\end{proof}

The next lemma may be used in place of Lemma~\ref{commisminus1} when the support of $h$ is large, but $h'(t)$ has small support.
\begin{lemma}\label{usefulinthatonepart}Let $\mathcal{O}$ be an operator of order $k \leq 0$ defined by~\eqref{theooperator}, and let $h(t^*) \in C^{\infty}\left(\mathbb{R}\right)$ satisfy that ${\rm supp}\left(h'(t^*)\right) \subset I$ where $I$ is an interval so that $|I|\lesssim 1$. Then, for any $q > 0$  and non-negative integer $n$, 
\[\sum_{i+j=0}^{|k-1|}\int_{\mathbb{R}}\left\vert\left\vert \left[\mathcal{O},h\right]T^i\Phi^jf \right\vert\right\vert_{\mathcal{E}_{w,s, n}}^2\, ds \lesssim_q \sup_{t^*}\int_{t^*-q}^{t^*+q}\left\vert\left\vert f\right\vert\right\vert^2_{\mathcal{E}_{w,s, n}}\, ds,\]
\[\sum_{i+j=0}^{|k-1|}\int_{\mathbb{R}}\left\vert\left\vert \left[\mathcal{O},h\right]T^i\Phi^jf \right\vert\right\vert_{\mathcal{E}_{w,s,n}}\, ds \lesssim_q \sup_{t^*}\sqrt{\int_{t^*-q}^{t^*+q}\left\vert\left\vert f\right\vert\right\vert^2_{\mathcal{E}_{w,s, n}}\, ds}.\]
\end{lemma}
\begin{proof} We start with case of $n = 0$. We may write $\mathbb{R}\setminus I = J_1 \cup J_2$ for two infinite intervals $J_1$ and $J_2$. The function $h(t^*)$ must take constant values $h_1$ and $h_2$ on $J_1$ and $J_2$. For each $l \in \{1,2\}$ let $\chi_l$ be a smooth function supported in $J_l$ so that $\{t^* : \chi_l(t^*) \neq 1\} \cap J_l$ is contained in an interval of size $\lesssim 1$. Set $\chi_3 \doteq 1- \chi_1-\chi_2$. We now can write 
\begin{equation}\label{howtosplititupforthissillylemma}\begin{split}
\| \left[\mathcal{O},h\right]T^i\Phi^jf\|_{\mathcal{E}_{w,s}} &\leq \|\left[\mathcal{O},h-h_1\right]T^i\Phi^j\left(\chi_1f\right)\|_{\mathcal{E}_{w,s}} + \|\left[\mathcal{O},h-h_2\right]T^i\Phi^j\left(\chi_2f\right)\|_{\mathcal{E}_{w,s}}
\\ &\qquad + \|\left[\mathcal{O},h\right]T^i\Phi^j\left(\chi_3f\right)\|_{\mathcal{E}_{w,s}}
\\ &= \|\left(h-h_1\right)\mathcal{O}T^i\Phi^j\left(\chi_1f\right)\|_{\mathcal{E}_{w,s}} + \|\left(h-h_2\right)\mathcal{O}T^i\Phi^j\left(\chi_2f\right)\|_{\mathcal{E}_{w,s}} 
\\ &\qquad + \|\left[\mathcal{O},h\right]T^i\Phi^j\left(\chi_3f\right)\|_{\mathcal{E}_{w,s}}.
\end{split}\end{equation}

Applying \eqref{usethislater123123}  yields for $l\in\{1,2\}$:
\[\begin{split}
\int_{\mathbb{R}}\|\left(h-h_l\right)\mathcal{O}T^i\Phi^j\left(\chi_l f\right)\|_{\mathcal{E}_{w,s}}^2\, ds \lesssim  \sup_{t^*}\int_{t^*-q}^{t^*+q}\left\vert\left\vert f\right\vert\right\vert^2_{\mathcal{E}_{w,s}}\, ds,
\end{split}
\] 
while \eqref{astartbutcandobetter}, the fact that $\left[\mathcal{O},h\right]$ is order $k-1$, and the fact that the support of $\chi_3$ is  contained in an interval of size $\lesssim 1$  yields
\[
\int_{\mathbb{R}}\|\left[\mathcal{O},h\right]T^i\Phi^j\left(\chi_3f\right)\|_{\mathcal{E}_{w,t}}^2 dt \lesssim \int_{\mathbb{R}} \|\left(\chi_3f\right)\|_{\mathcal{E}_{w,t}}^2 dt \lesssim_q  \sup_{t^*}\int_{t^*-q}^{t^*+q}\left\vert\left\vert f\right\vert\right\vert^2_{\mathcal{E}_{w,s}}\, ds.
\]

Upgrading the estimate to the case of $n > 0$ is done in an analogous fashion to the proof of Lemma~\ref{commisminus1}.
\end{proof}

The next lemma  may be thought of as a consequence of a localized version of Plancherel's theorem. First we need a suitable definition.
\begin{definition}\label{projpair}We say that a pair of $(t^*,\phi^*)$ pseudo-differential operators $\left(\mathcal{O},\tilde{\mathcal{O}}\right)$ of order $0$ are a projecting pair if the corresponding symbols $o$ and $\tilde{o}$ do not depend on $t$, we have $\left|o\right|,\left|\tilde{o}\right| \leq 1$, and $\tilde{\mathcal{O}}\mathcal{O} = \mathcal{O}$.
\end{definition}
\begin{remark}The standard example to keep in mind for Definition~\ref{projpair} is when $o\left(\cdot,m\right)$ is a bump function of amplitude at most $1$ and supported in a certain interval and $\tilde{o}\left(\cdot,m\right)$ is a another bump function of amplitude at most one and supported in a slightly larger interval so that $\tilde{o}$ is identically $1$ on the support of $o$.
\end{remark}

Now we are ready for a version of a localized Plancherel's theorem.
\begin{lemma}\label{localizethatplancherelyay}Let $\left(\mathcal{O},\tilde{O}\right)$ be a projecting pair in the sense of Definition~\ref{projpair}, and let $ J \subset \mathbb{R} \times \mathbb{Z}$ be the support of the symbol $\tilde{o}$ of $\tilde{O}$. 

Let $\{D_i\}_{i=1}^4$ be a set of (possibly complex) first order differential operators on $\mathcal{M}_K$ so that at each point $p \in \mathcal{M}_K$, there exist complex numbers $a_i^1$, $a_i^2$, and $a_i^3$ so that $D_i = a_i^1T  +a_i^2\Phi + a_i^3$. We assume moreover that as functions on $\mathcal{M}_K$, the $\{a_i^j\}_{j=1}^3$ are independent of $t^*$ and $\phi^*$.

Let $p_i : \mathcal{M}_K \times \mathbb{R} \times \mathbb{Z} \to \mathbb{C}$ denote the corresponding symbol of $D_i$ determined by $T \mapsto -i\omega$ and $\Phi \mapsto im$. Assume that for all points in $\mathcal{M}_K$ and $\left(\omega,m\right) \in J$ we have that
\begin{equation}
{\rm Re}\left(p_1\overline{p_2}\right) \leq {\rm Re}\left(p_3\overline{p_4}\right).
\end{equation}
Then the following estimates hold: Let $h(t^*)\in C^{\infty}_c$, $q > 0$, $\delta > 0$, and $I$ be a bounded interval which contains the support of $h$ and satisfies $|I| \gtrsim 1$. Then there exists a constant $C$ so that we will have that
\begin{equation}\label{finiteinterapproxplan}\begin{split}
&\int_{-\infty}^{\infty}\int_{r_+}^{\infty}\int_{\mathbb{S}^2}h^2(t^*){\rm Re}\left(D_1\mathcal{O}f\overline{D_2\mathcal{O}f}\right) w\, d\mathbb{S}^2\, dr\, dt^* \leq
\\&\qquad \int_{-\infty}^{\infty}\int_{r_+}^{\infty}\int_{\mathbb{S}^2}h^2(t^*){\rm Re}\left(D_3\mathcal{O}f\overline{D_4\mathcal{O}f}\right) w\, d\mathbb{S}^2\, dr\, dt^* 
\\ &\qquad + C\left[\delta^{-1}\int_I\left\vert\left\vert \mathcal{O}f\right\vert\right\vert^2_{\mathcal{E}_{w,t^*}}\, dt^* +\delta \sum_{i=1}^4\int_I\left\vert\left\vert \mathcal{O}D_if\right\vert\right\vert^2_{\mathcal{E}_{w,t^*}}\, dt^*\right]
\\ &\qquad + CA_{\rm high}^{-1}\sum_{i=1}^4\left[\sup_{t^*} \int_{t^*-q}^{t^*+q}\left\vert\left\vert D_i\mathcal{O}f\right\vert\right\vert_{\mathcal{E}_{w,s}}^2\, ds + \int_I\left\vert\left\vert D_i\mathcal{O}f\right\vert\right\vert_{\mathcal{E}_{w,t^*}}^2\, dt^*\right].
\end{split}\end{equation}
The constant $C$ may be taken to be proportional to a sum of  $\left\vert\left\vert h\right\vert\right\vert_{C^N}$ for some sufficiently large $N$ and  the sum of the squares of the implied constants in~\eqref{thesynbolassum} for $i+j \leq N$. 

Lastly, we note that if none of the $D_i$ involve $T$, then the third line of~\eqref{finiteinterapproxplan} may be dropped. 

\end{lemma}
\begin{proof} Due to our assumptions on the coefficients, we have that $[D_i, \mathcal{O}] = [D_i, \tilde{\mathcal{O}}] = 0.$ We compute
\begin{equation}\begin{split}
&\int_{-\infty}^{\infty}\int_{r_+}^{\infty}\int_{\mathbb{S}^2}h^2(t^*){\rm Re}\left(D_1\mathcal{O}f\overline{D_2\mathcal{O}f}\right) w\, d\mathbb{S}^2\, dr\, dt^* =
\\ &\qquad \int_{-\infty}^{\infty}\int_{r_+}^{\infty}\int_{\mathbb{S}^2}h^2(t^*){\rm Re}\left(D_1\tilde{\mathcal{O}}\mathcal{O}f\overline{D_2\tilde{\mathcal{O}}\mathcal{O}f}\right) w\, d\mathbb{S}^2\, dr\, dt^* =
\\ &\qquad \underbrace{\int_{-\infty}^{\infty}\int_{r_+}^{\infty}\int_{\mathbb{S}^2}{\rm Re}\left(\left[h,\tilde{\mathcal{O}}\right]\left(D_1\mathcal{O}f\right)\overline{hD_2\mathcal{O}f}\right) w\, d\mathbb{S}^2\, dr\, dt^*}_{\doteq I} +
\\ &\qquad \underbrace{\int_{-\infty}^{\infty}\int_{r_+}^{\infty}\int_{\mathbb{S}^2}{\rm Re}\left(\tilde{\mathcal{O}}\left(hD_1\mathcal{O}f\right)\overline{\left[h,\tilde{\mathcal{O}}\right]D_2\mathcal{O}f}\right) w\, d\mathbb{S}^2\, dr\, dt^*}_{\doteq II} +
\\ &\qquad \int_{-\infty}^{\infty}\int_{r_+}^{\infty}\int_{\mathbb{S}^2}{\rm Re}\left(\left(\tilde{\mathcal{O}}\left(hD_1\mathcal{O}f\right)\right)\overline{\tilde{\mathcal{O}}\left(hD_2\mathcal{O}f\right)}\right) w\, d\mathbb{S}^2\, dr\, dt^* = 
\\ &\qquad I + II+
\\&\qquad \underbrace{-\int_{-\infty}^{\infty}\int_{r_+}^{\infty}\int_{\mathbb{S}^2}{\rm Re}\left(\left(\tilde{\mathcal{O}}\left(\left(D_1h\right)\mathcal{O}f\right)\right)\overline{\tilde{\mathcal{O}}\left(hD_2\mathcal{O}f\right)}\right) w\, d\mathbb{S}^2\, dr\, dt^*}_{\doteq III} +
\\ &\qquad \underbrace{-\int_{-\infty}^{\infty}\int_{r_+}^{\infty}\int_{\mathbb{S}^2}{\rm Re}\left(\left(\tilde{\mathcal{O}}\left(D_1h\mathcal{O}f\right)\right)\overline{\tilde{\mathcal{O}}\left(D_2h\right)\mathcal{O}f}\right) w\, d\mathbb{S}^2\, dr\, dt^*}_{\doteq IV} +
\\ &\qquad \underbrace{+\int_{-\infty}^{\infty}\int_{r_+}^{\infty}\int_{\mathbb{S}^2}{\rm Re}\left(\left(\tilde{\mathcal{O}}\left(D_1h\mathcal{O}f\right)\right)\overline{\tilde{\mathcal{O}}\left(D_2h\mathcal{O}f\right)}\right) w\, d\mathbb{S}^2\, dr\, dt^*}_{\doteq V}. 
\end{split}
\end{equation}
In view of Lemma~\ref{commisminus1} and Cauchy-Schwarz, we have 
\begin{equation}\label{thisisforIyeah}
\left|I\right| + \left|II\right| \lesssim A_{\rm high}^{-1}\left[\sum_{j=1}^2\sup_{t^*} \int_{t^*-q}^{t^*+q}\left\vert\left\vert D_i\mathcal{O}f\right\vert\right\vert_{\mathcal{E}_{w,s}}^2\, ds + \int_I\left\vert\left\vert D_i\mathcal{O}f\right\vert\right\vert_{\mathcal{E}_{w,s}}^2\, ds\right].
\end{equation}

In view of Plancherel's theorem and Cauchy-Schwarz, we have that
\begin{equation}\label{thisisforIIyeah}
\left|III\right| + \left|IV\right| \lesssim \delta^{-1}\int_I\left\vert\left\vert \mathcal{O}f\right\vert\right\vert^2_{\mathcal{E}_{w,s}}\, ds +\delta \sum_{i=1}^4\int_I\left\vert\left\vert \mathcal{O}D_if\right\vert\right\vert^2_{\mathcal{E}_{w,s}}\, ds.
\end{equation}

Using Plancherel's theorem again, we have that
\[V \leq\int_{-\infty}^{\infty}\int_{r_+}^{\infty}\int_{\mathbb{S}^2}{\rm Re}\left(\left(\tilde{\mathcal{O}}\left(D_3h\mathcal{O}f\right)\right)\overline{\tilde{\mathcal{O}}\left(D_4h\mathcal{O}f\right)}\right) w\, d\mathbb{S}^2\, dr\, dt^* .\]
It then remains to uncommute the operator $\tilde{\mathcal{O}}$ and the function $h$ using the analogues of the estimates~\eqref{thisisforIyeah} and~\eqref{thisisforIIyeah}. One then obtains~\eqref{finiteinterapproxplan}.

If none of the $D_i$ involve $T$, then $D_i$ will commute with $h$, and we will not have terms like $III$ and $IV$.

\end{proof}

The following is a cruder version of Lemma~\ref{localizethatplancherelyay}, which, despite having a more restricted range of validity, is sometimes more convenient to apply.
\begin{lemma}\label{crudeplanch}Let $\left(\mathcal{O},\tilde{O}\right)$ be a projecting pair in the sense of Definition~\ref{projpair}, and let $ J \subset \mathbb{R} \times \mathbb{Z}$ be the support of the symbol $\tilde{o}$ of $\tilde{O}$. 

Let $\{D_i\}_{i=1}^2$ be a set of (possibly complex) first order differential operators on $\mathcal{M}_K$ so that at each point $p \in \mathcal{M}_K$, there exists complex numbers $a_1$, $a_2$, and $a_3$ so that $D_1 =  a_1\Phi + a_2$ and $D_2 = a_3T$. We assume moreover that as functions on $\mathcal{M}_K$, the $\{a_i\}$ are independent of $t^*$ and $\phi^*$.

Let $p_i : \mathcal{M}_K\times \mathbb{R} \times \mathbb{Z} \to \mathbb{C}$ denote the corresponding symbol of $D_i$ determined by $T \mapsto -i\omega$ and $\Phi \mapsto im$. Assume there exists a large constant $A_{\rm large}$ so that for all points in $\mathcal{M}_K$ and $\left(\omega,m\right) \in J$ we have that
\begin{equation}
\left|p_1\right|^2 \lesssim  A_{\rm large}^{-2}\left|p_2\right|^2.
\end{equation}
Then the following estimates hold: Let $h(t^*) \in C^{\infty}_c$, $q > 0$, $\delta > 0$, and $I$ be a bounded interval which contains the support of $h$ and satisfies $|I| \gtrsim 1$. Then there exists a constant $C$ so that we will have that
\begin{equation}\label{crudeapproxplan}\begin{split}
&\int_{-\infty}^{\infty}\int_{r_+}^{\infty}\int_{\mathbb{S}^2}h^2(t^*)\left|D_1\mathcal{O}f\right|^2 w\, d\mathbb{S}^2\, dr\, dt^* \lesssim
\\&\qquad {\rm max}\left(A_{\rm large}^{-2},A_{\rm high}^{-2}\right)\Big[\int_{-\infty}^{\infty}\int_{r_+}^{\infty}\int_{\mathbb{S}^2}h^2(t^*)\left|D_2\mathcal{O}f\right|^2 w\, d\mathbb{S}^2\, dr\, dt^*
 \\ &\qquad + \sup_{t^*} \int_{t^*-q}^{t^*+q}\left[\left\vert\left\vert \mathcal{O}f\right\vert\right\vert_{\mathcal{E}_{w,s}}^2+\left\vert\left\vert D_1\mathcal{O}f\right\vert\right\vert_{\mathcal{E}_{w,s}}^2\right]\, ds+\int_I\left[\left\vert\left\vert \mathcal{O}f\right\vert\right\vert_{\mathcal{E}_{w,s}}^2+\left\vert\left\vert D_1\mathcal{O}f\right\vert\right\vert_{\mathcal{E}_{w,s}}^2\right]\, ds\Big].
\end{split}\end{equation}
The implied constant may be taken to be proportional to a sum of  $\left\vert\left\vert h\right\vert\right\vert_{C^N}$ for some sufficiently large $N$ and  the sum of the squares of the implied constants in~\eqref{thesynbolassum} for $i+j \leq N$. 
\end{lemma}
\begin{proof}One may essentially repeat the proof of Lemma~\ref{localizethatplancherelyay}. We omit the straightforward details.
\end{proof}

\section{Proof of Theorem~\ref{mainLinEstTheo}}\label{proofthemainlineestsec}
In this section we will give the proof of Theorem~\ref{mainLinEstTheo}. In fact, we will prove a more general result concerning solutions to the interpolating equation from Definition~\ref{definterpolate}. 
\begin{theorem}\label{actuallywehatweneedformain}Let $g \in \mathscr{A}_{\epsilon,r_{\rm low},R_{\rm high}}$, $\tau \gg 10$, and let $\Box_{\tau}$ be as in Definition~\ref{definterpolate}. Suppose that we have $\psi_{\tau} : J^+\left(\Sigma_0\right) \to \mathbb{C}$ which solves
\begin{equation}\label{taupssolve}
\Box_{\tau}\psi_{\tau} = F,
\end{equation}
 with $F$ compactly supported on $J^{+}(\Sigma_1)$, and $\left(n_{\Sigma_0}\psi,\psi\right)|_{\Sigma_0}$  smooth and compactly supported.

Set $\Psi_{\tau} \doteq \mathcal{P}_{HT}\psi_{\tau}$, $\mathscr{F} \doteq  \mathcal{P}_{HT}F + \left[\Box_{\tau},\mathcal{P}_{HT}\right]\psi_{\tau}$, and let $k \geq 0$. We then have 
\begin{equation}\label{thebetterengthing}\begin{split}
&\sup_{{t} \in (10,\tau-10)}\int_{-\infty}^{\infty}\chi_{t}(s)E_k\left[\Psi_{\tau}\right](s)\, ds  \lesssim_k 
 A_{\rm high}^{-1}\sup_{{t} > 0}\int_{-\infty}^{\infty}\chi_{t}(s)E_k\left[\psi_{\tau}\right]\left(s\right)\, ds 
\\ &+ \sum_{j=0}^k\sup_{{t} \in (10,\tau-10) }\left|\int_{\mathcal{M}_K}\tilde{\chi}^2_{t}{\rm Re}\left( T^j\mathscr{F}\overline{K\left(T^j\Psi_{\tau}\right)}\right)\right|
\\  & +\sum_{j+l=1}^k\sup_{{t} \in (10,\tau-10) }\left|\int_{\mathcal{M}_K}\tilde{\chi}_{t}^2{\rm Re}\left( T^jY^l\mathscr{F}\overline{N\left(T^jY^l\Psi_{\tau}\right)}\right)\right|+ E_k\left[\psi_{\tau}\right](0)
\\ &+\sum_{j=0}^k \sup_{{t} \in (10,\tau-10)}\left|\int_{\mathcal{M}_K}\rho \sqrt{\frac{\Delta}{\Pi}}\chi_{t}^2\xi{\rm Re}\left( T^j\mathscr{F}\overline{\left(T^j\Psi_{\tau}\right)}\right)\right|
\\ &+\sup_{{t} \in (10,\tau-10)}\int_{-\infty}^{\infty}\chi_{t}(s)E_{k-2}\left[\mathscr{F}\right]\left(s\right)\, ds+\sup_{{t} \in (10,\tau-10)}\int_{\mathcal{M}_K}\chi_{t}(s)\left|\mathscr{F}\right|^2\, ds.
\end{split}\end{equation}

\end{theorem}
\begin{remark}
When $g = g_K$, $\Box_{\tau} = \Box_{g_K}$, and $\left[\Box_{\tau},\mathcal{P}_{HT}\right] = 0$. Theorem~\ref{mainLinEstTheo} thus immediately follows from Theorem~\ref{actuallywehatweneedformain} and finite in time energy estimates.  
\end{remark}

It will be convenient to set
\[\tilde{\mathscr{F}}\doteq \mathcal{P}_{HT}\tilde{F},\]
where $\tilde{F}$ is as Corollary~\ref{corextend}.
We will generally apply our estimates to the equation
\begin{equation}\label{PsitauscrFeqn}
    \Box_{\tau}\Psi_{\tau} = \mathscr{F}+ \tilde{\mathscr{F}},
\end{equation}
for $t^* \in (10,\tau-10)$. 

\subsection{Estimates for the operator $\mathcal{P}_{HT}$}
In this section, we will use the estimates from Section~\ref{PDOstuff} to establish various useful properties of the operator $\mathcal{P}_{HT}$.

\begin{convention}We introduce the convention that in this section all implied constants may depend on the choice of $\delta_1$ and $\delta_2$ in Definition~\ref{hightrappedpart}. 
\end{convention}

The next few lemmas will be used to show that, after a bit of averaging in time, we can control $\mathcal{P}_{HT}f$ and $\Phi\mathcal{P}_{HT}f$ by $T\mathcal{P}_{HT}f$. We start with $\mathcal{P}_{HT}f$.
\begin{lemma}\label{nothingbyT}Let $f \in \mathbb{H}_K$ and $t \in \mathbb{R}$. Then
\begin{align}\label{theestforthislemm123a}\begin{split}
&\qquad \int_{-\infty}^{\infty}\int_{r_+}^{\infty}\int_{\mathbb{S}^2}\chi_{t}^2\left|\mathcal{P}_{HT}f\right|^2w\, d\mathbb{S}^2 dr ds \lesssim 
A_{\rm high}^{-2}\int_{-\infty}^{\infty}\int_{r_+}^{\infty}\int_{\mathbb{S}^2}\chi_{t}^2\left|T\mathcal{P}_{HT}f\right|^2w\, d\mathbb{S}^2drds 
  \\  &\qquad \qquad  + A_{\rm high}^{-2}\sup_{t^* \in \mathbb{R}}\int_{t^*-1/2}^{t^*+1/2}\left\vert\left\vert \mathcal{P}_{HT}f\right\vert\right\vert^2_{\mathcal{E}_{w,s}}\, ds.
\end{split}\end{align}
\end{lemma}
\begin{proof}This is an immediate consequence of Lemma~\ref{crudeplanch} where we take
\[
a_1 = 0, \quad a_2=a_3=1, \quad h=\chi_{t}, \quad\mathcal{O} = P_{HT}, \quad q=1/2, \quad A_{large} = A_{high},
\]
while $J$ is a suitable neighborhood of the support of $a_{HT}$ so that $|\omega|\gtrsim A_{high}$ on $J$.

\end{proof}

Next, we consider $\Phi\mathcal{P}_{HT}f$.
\begin{lemma}\label{phibyT}Let $f \in \mathbb{H}_K$ and $t \in \mathbb{R}$. Then
\begin{align}
   & \int_{-\infty}^{\infty}\int_{r_+}^{\infty}\int_{\mathbb{S}^2}\chi_{t}^2\left|\Phi\mathcal{P}_{HT}f\right|^2w\, d\mathbb{S}^2drdt \lesssim 
\sum_{i = -1}^{1}\int_{-\infty}^{\infty}\int_{r_+}^{\infty}\int_{\mathbb{S}^2}\chi_{t+i}^2\left|T\mathcal{P}_{HT}f\right|^2w\, d\mathbb{S}^2drdt^* 
      \\  &\qquad  + A_{\rm high}^{-2}\sup_{t^* \in \mathbb{R}}\int_{t^*-1/2}^{t^*+1/2}\left\vert\left\vert \Phi\mathcal{P}_{HT}f\right\vert\right\vert^2_{\mathcal{E}_{w,s}}\, ds.
\end{align}
\end{lemma}
\begin{proof}
This follows from Lemma~\ref{crudeplanch} and Lemma~\ref{nothingbyT}. Indeed, we apply Lemma~\ref{crudeplanch} with
\[
a_2 = 0, \quad a_1=a_3=1, \quad h=\chi_{t}, \quad\mathcal{O} = P_{HT}, \quad q=1/2, \quad A_{large} = A_{high},
\]
while $J$ is a suitable neighborhood of the support of $a_{HT}$ so that $|\omega|\gtrsim |m|$ on $J$, and control the lower order terms on the third line of \eqref{crudeapproxplan} by Lemma~\ref{nothingbyT}. Finally, since $m \neq 0$ on the support of $a_{HT}$ we may bound $\left\vert\left\vert \mathcal{P}_{HT}f\right\vert\right\vert^2_{\mathcal{E}_{w,s}}$ by $\left\vert\left\vert \Phi\mathcal{P}_{HT}f\right\vert\right\vert^2_{\mathcal{E}_{w,s}}$.

\end{proof}

The next lemma will be the key to showing the global positivity of our main physical space estimate (see Lemma~\ref{wowwhatanestimate} below).
\begin{lemma}\label{itsformallynoactuallypositive}Let $f \in \mathbb{H}_K$, $t \in \mathbb{R}$, and $h(r)$ be a positive function which satisfies $h(r) \leq \frac{a}{2Mr_+}$. Then there exists a small constant $c = c(\delta_1,\delta_2)> 0$ and a large constant $C > 0$ so that 
\begin{align}\label{blahblah2341}
&\int_{-\infty}^{\infty}\int_{r_+}^{\infty}\int_{\mathbb{S}^2}\chi_{t}^2{\rm Re}\left(\left(T\mathcal{P}_{HT}f+h(r)\Phi\mathcal{P}_{HT}f\right)\left(\overline{T\mathcal{P}_{HT}f + \frac{a}{2Mr_+}\Phi\mathcal{P}_{HT}f}\right)\right)w\, d\mathbb{S}^2drdt^* 
\\ \nonumber &\qquad \geq  c\int_{-\infty}^{\infty}\int_{r_+}^{\infty}\int_{\mathbb{S}^2}\chi_{t}^2\left[\left|T\mathcal{P}_{HT}f\right|^2+\left|\Phi \mathcal{P}_{HT}f\right|^2\right]w\, d\mathbb{S}^2drdt^*  
\\ \nonumber &\qquad - CA_{\rm high}^{-1}\sup_{t^* \in \mathbb{R}}\int_{t^*-1/2}^{t^*+1/2}\left[\left\vert\left\vert T\mathcal{P}_{HT}f\right\vert\right\vert^2_{\mathcal{E}_{w,s}}+\left\vert\left\vert \Phi \mathcal{P}_{HT}f\right\vert\right\vert^2_{\mathcal{E}_{w,s}}\right]\, ds.
\end{align}
\end{lemma}
\begin{proof}

The main observation is that we can pick $J$ a neighborhood of the support of $a_{HT}$ so that there is a constant $c=c(\delta_1, \delta_2)$ so that on $J$
\[
\left(\omega- h(r)m\right)\left(\omega - \frac{a}{2Mr_+} m\right) \geq c(\omega^2+m^2).
\]

We now apply Lemma~\ref{localizethatplancherelyay}   with
\[
D_1 = T + h(r)\Phi, \quad D_2 = T + \frac{a}{2Mr_+}\Phi, \quad D_3 = D_4 = c^{1/2} (T+i\Phi), 
\]
and $\delta = A_{\rm high}^{-1}$. We then use Lemma~\ref{nothingbyT} to control the term involving $\mathcal{P}_{HT}f$, and finally note that since $m \neq 0$  on the support of $a_{HT}$, we may replace $\mathcal{P}_{HT}f$ with $\Phi\mathcal{P}_{HT}f$ in the final estimate.

\end{proof}

In this next lemma we show that the time-average $L^2$ norm of $\mathcal{P}_{HT}f$ may be controlled by the corresponding time-average $L^2$ norms of $f$.
\begin{lemma}\label{totheoriginal}Let $f \in \mathbb{H}_K$, $t \in \mathbb{R}$, and $n$ be a non-negative integer. Then
\begin{equation}\label{Backtothebeginning}
\int_{-\infty}^{\infty}\chi_{t}^2(s)\left\vert\left\vert \mathcal{P}_{HT}f\right\vert\right\vert^2_{\mathcal{E}_{w,s,n}}\, ds \lesssim \int_{-\infty}^{\infty}\chi_{t}^2(s)\left\vert\left\vert f\right\vert\right\vert^2_{\mathcal{E}_{w,s,n}}\, ds + A_{\rm high}^{-2}\sup_{t^*\in\mathbb{R}}\int_{t^*-1/2}^{t^*+1/2}\left\vert\left\vert f\right\vert\right\vert^2_{\mathcal{E}_{w,s,n }}\, ds
\end{equation}
\end{lemma}
\begin{proof}This is a consequence of \eqref{blahbi2j} in Lemma~\ref{commisminus1}.
\end{proof}

In this final lemma, we consider the situation when $f$ is supported for $t < t_0$, and show that we can control $\mathcal{P}_{HT}f$ in $L^1_t$ norms in the region $t > t_0+1$ in terms of $L^{\infty}_t$ norms of $f$ in the region $t< t_0$.
\begin{lemma}\label{pseudolocforthewin}Let $f \in \mathbb{H}_K$, $t_0 \in \mathbb{R}$, and suppose that ${\rm supp}\left(f\right) \subset \{t^* < t_0\}$. Then, for any non-negative integers $n$ and $k$, we have
\[\sum_{i+j \leq n}\int_{t_0+1}^{\infty}\left\vert\left\vert {T^i\Phi^j}\mathcal{P}_{HT}f\right\vert\right\vert_{\mathcal{E}_{w,s}}\, ds \lesssim_k A_{\rm high}^{-k}\sup_{t^*} \left\vert\left\vert f\right\vert\right\vert_{\mathcal{E}_{w,t^*}}.\]
\end{lemma}
\begin{proof}This follows from Theorem~\ref{offthesupportitsgood}  with $I_1=(t_0+1, \infty)$ and $I_2=(t_0, \infty)$.

\end{proof}
\subsection{Physical space currents}
In this section we will construct a physical space current which, when combined with the Fourier analysis of the previous section will yield a proof of Theorem~\ref{actuallywehatweneedformain}. We note that throughout the section we will refer to the function $\tilde{F}$ which is defined in Theorem~\ref{extendtoeverywhere}.

We start by analyzing the flux term $\mathbf{J}^K_{\mu}n^{\mu}_{\Sigma_{s}}$ when $r \geq r_++\delta_0$ on the Kerr metric.

\begin{lemma}\label{whatexactlydoesjklooklike}On the Kerr metric $g_K$, we have that $r \geq r_+ + \delta_0$ implies
\begin{equation}\label{normaltoSigmawhenrislargeenough}
n_{\Sigma_{s}} = -\left(\nabla t\right)\rho\sqrt{\frac{\Delta}{\Pi}} = \rho^{-1}\sqrt{\frac{\Pi}{\Delta}}T + \rho^{-1} \frac{2Mar}{\sqrt{\Delta \Pi}}\Phi,
\end{equation}
\begin{equation}\label{JKflux}\begin{split}
\mathbf{J}^K_{\mu}\left[f\right]n^{\mu}_{\Sigma_{s}} &= \rho^{-1}\sqrt{\frac{\Pi}{\Delta}}{\rm Re}\left(\left(Tf + \frac{a}{2Mr_+}\Phi f\right)\left(\overline{Tf + \frac{2Mar}{\Pi}\Phi f}\right) \right)+\frac{1}{2}\rho \sqrt{\frac{\Delta}{\Pi}}\partial^{\gamma}f \partial_{\gamma}f
\end{split}
\end{equation}
\begin{equation}\label{JKfluxFull}\begin{split}
\mathbf{J}^K_{\mu}\left[f\right]n_{\Sigma_{s}}^{\mu} &= \frac{1}{2}\rho^{-1}\sqrt{\frac{\Pi}{\Delta}}\left|T f\right|^2 + \rho^{-1}\frac{a}{2Mr_+}\sqrt{\frac{\Pi}{\Delta}}\text{Re}\left(T f\overline{\Phi f}\right)
\\  &+ \frac{1}{2} \rho \sqrt{\frac{\Delta}{\Pi}}\left(\left[\frac{1}{\rho^2\sin^2\theta} + \frac{a^2}{\Delta\rho^2}\left(\frac{2r}{r_+} - 1\right)\right]\left|\Phi f\right|^2
+ \frac{\Delta}{\rho^2}\left|\partial_rf\right|^2 + \rho^{-2}\left|\partial_{\theta}f\right|^2\right).
\end{split}\end{equation}
\end{lemma}
\begin{proof}
    The equation~\eqref{normaltoSigmawhenrislargeenough} follows from~\eqref{thisisnablat} and the fact that when $r \geq r_++\delta_0$, $\Sigma_{s}$ is a hypersurface of constant $t$. 

Keeping in mind that $K = T + \frac{a}{2Mr_+}\Phi$, we have
\[g_K\left(K,n_{\Sigma_{s}}\right) = g_K\left(T,n_{\Sigma_{s}}\right) = -\rho\sqrt{\frac{\Delta}{\Pi}}.\]
Given this, the formula~\eqref{JKflux} is immediate, and the formula~\eqref{JKfluxFull} follows from a straightforward computation. 
\end{proof}

Now we come to the key estimate which will combine a time-averaged $\mathbf{J}^K$-energy estimate with a Lagrangian correction to eliminate the last term in~\eqref{JKflux}.
\begin{lemma}\label{wowwhatanestimate} Let $\psi_{\tau}$ be as in the statement of Theorem~\ref{actuallywehatweneedformain}. Then we have, for every $t \in (10,\tau-10)$, 
\begin{equation}\label{wowthisistheresultofavglagr}
\begin{split}
  &\int_{-\infty}^{\infty}\chi^2_t(s)\left(\int_{\Sigma_s}\mathbf{J}^K_{\mu}n^{\mu}_{\Sigma_s}\left(1-\xi\right)\right)\, ds 
  \\ &\ \  + \int_{\mathcal{M}_K}{\rho \sqrt{\frac{\Delta}{\Pi}}}\chi_t^2\xi\Bigg(\rho^{-2}\frac{\Pi}{\Delta}{\rm Re}\left(\left(T\Psi_{\tau} + \frac{a}{2Mr_+}\Phi\Psi_{\tau}\right)\left(\overline{T\Psi_{\tau} + \frac{2Mar}{\Pi}\Phi\Psi_{\tau}}\right) \right)
   \\ &\qquad \qquad \qquad  +\frac{1}{4}\Box_g\left(\rho \sqrt{\frac{\Delta}{\Pi}}\chi_{\tau}^2\xi\right)\left|\Psi_{\tau}\right|^2\Bigg)
    \\ & \leq O(\epsilon)  \int_{-\infty}^{\infty}\chi_t^2(s)E_{0,r\geq r_++\delta_0/2}\left[\Psi_{\tau}\right]\left(s\right)\, ds
   + \int_{-\infty}^{\infty}\chi^2(s)\left(\int_{\Sigma_t}\mathbf{J}^K_{\mu}n^{\mu}_{\Sigma_t}\right)\, ds 
 \\ &  + \left|\int_{\mathcal{M}_K}\tilde{\chi}_t^2{\rm Re}\left(\mathscr{F}\overline{K\Psi_{\tau}}\right)\right|+ \left|\int_{\mathcal{M}_K}{\rho \sqrt{\frac{\Delta}{\Pi}}}\chi_t^2\xi{\rm Re}\left(\mathscr{F}\overline{\Psi_{\tau}}\right)\right|
   \\ &\ \  + \left|\int_{\mathcal{M}_K}\tilde{\chi}^2_t{\rm Re}\left(\tilde{\mathscr{F}}\overline{K\Psi_{\tau}}\right)\right|+ \left|\int_{\mathcal{M}_K}{\rho \sqrt{\frac{\Delta}{\Pi}}}\chi_t^2\xi{\rm Re}\left(\tilde{\mathscr{F}}\overline{\Psi_{\tau}}\right)\right|,
\end{split}
\end{equation}
Here the currents $\mathbf{J}^K$ are defined with respect to $\Psi_\tau$ and with respect to the metric $g$.
\end{lemma}
\begin{proof}
We start by applying a $K$-energy estimate to~\eqref{PsitauscrFeqn}. For any $t \in (10,\tau-10)$ and $s\in\mathbb{R}$, we apply the identity~\eqref{divIdentity} in the region $
\mathcal{R}\left(s,t+s\right)$ with $V = K$, $w = 0$, $q = 0$, and $\varpi = 0$. Since $K$ is a Killing vector field, we have that $\mathbf{K}^K = 0$. Moreover, since $K$ is future oriented and normal to $\mathcal{H}^+$, we have that $\mathbf{J}^K_{\mu}n_{\mathcal{H}^+}^{\mu} \geq 0$. After applying the divergence theorem we thus obtain, using that $\Box_{\tau}\Psi_{\tau} = \mathscr{F}+\tilde{\mathscr{F}}$:
\begin{equation}\label{thefirstenergyest}
\int_{\Sigma_{t+s}}\mathbf{J}_{\mu}^Kn_{\Sigma_{t+s}}^{\mu} \leq \int_{\Sigma_s}\mathbf{J}_{\mu}^Kn_{\Sigma_s}^{\mu} + \int_{\mathcal{R}(s,t+s)}{\rm Re}\left(\mathscr{F}\overline{K\Psi_{\tau}}\right)+ \int_{\mathcal{R}(s,t+s)}{\rm Re}\left(\tilde{\mathscr{F}}\overline{K\Psi_{\tau}}\right).
\end{equation}
Applying Lemma~\ref{whatexactlydoesjklooklike} and that $g \in \mathscr{A}_{\epsilon,r_{\rm low},R_{\rm high}}$, we re-write~\eqref{thefirstenergyest} as 
\begin{equation}\label{thefirstenergyestexplicit}\begin{split}
&\int_{\Sigma_{t+s}}\mathbf{J}_{\mu}^Kn_{\Sigma_{t+s}}^{\mu}\left(1-\xi\right)  + O\left(\epsilon E_{0,r \geq r_+ + \delta_0/2}\left[\Psi_{\tau}\right]\left(t+s\right)\right)
\\ &\qquad +  \int_{\Sigma_{t+s}}\xi\left(\rho^{-1}\sqrt{\frac{\Pi}{\Delta}}{\rm Re}\left(\left(T\Psi_{\tau} + \frac{a}{2Mr_+}\Phi\Psi_{\tau}\right)\left(\overline{T\Psi_{\tau} + \frac{2Mar}{\Pi}\Phi\Psi_{\tau}}\right) \right)+\frac{1}{2}\rho \sqrt{\frac{\Delta}{\Pi}}\partial^{\gamma}\Psi_{\tau}\partial_{\gamma}\Psi_{\tau}\right) \leq \\ &\qquad \int_{\Sigma_s}\mathbf{J}_{\mu}^Kn_{\Sigma_t}^{\mu}  + \int_{\mathcal{R}(s,t+s)}{\rm Re}\left(\mathscr{F}\overline{K\Psi_{\tau}}\right)+ \int_{\mathcal{R}(s,t+s)}{\rm Re}\left(\tilde{\mathscr{F}}\overline{K\Psi_{\tau}}\right).
\end{split}\end{equation}
We note that in view of the smallness of $\delta_0$, we have that $K$ is timelike for $r \in (r_+,r_++2\delta_0)$ and thus that $\mathbf{J}^K_{\mu}n^{\mu}_{\Sigma_{t+s}}\left(1-\xi\right) \geq 0$.

Now we may integrate the inequality~\eqref{thefirstenergyestexplicit} in $s$ against the bump-function $\chi^2(s)$ to obtain
\begin{equation}\label{afteraveragingtheenergy}
\begin{split}
  &  \int_{-\infty}^{\infty}\chi_t^2(s)\left(\int_{\Sigma_s}\mathbf{J}^K_{\mu}n^{\mu}_{\Sigma_s}\left(1-\xi\right)\right)\, ds + O\left(\epsilon \int_{-\infty}^{\infty}\chi_t^2(s)E_{0,r\geq r_+ + \delta_0/2}\left[\Psi_{\tau}\right]\left(s\right)\, ds\right) +
  \\ &  \int_{\mathcal{M}_K\cap \{r \geq r_++\delta_0\}} \rho \sqrt{\frac{\Delta}{\Pi}}\chi_t^2\xi\Bigg(\rho^{-2}\frac{\Pi}{\Delta}{\rm Re}\left(\left(T\Psi_{\tau} + \frac{a}{2Mr_+}\Phi\Psi_{\tau}\right)\left(\overline{T\Psi_{\tau} + \frac{2Mar}{\Pi}\Phi\Psi_{\tau}}\right) \right)
   +\frac{1}{2}\partial^{\gamma}\Psi_{\tau}\partial_{\gamma}\Psi_{\tau}\Bigg)
   \\ &\leq \int_{-\infty}^{\infty}\chi^2(s)\left(\int_{\Sigma_s}\mathbf{J}^K_{\mu}n^{\mu}_{\Sigma_s}\right)\, ds +  \left|\int_{\mathcal{M}_K}\tilde{\chi}_{t}^2{\rm Re}\left(\mathscr{F}\overline{K\Psi_{\tau}}\right)\right| +  \left|\int_{\mathcal{M}_K}\tilde{\chi}_{t}^2{\rm Re}\left(\tilde{\mathscr{F}}\overline{K\Psi_{\tau}}\right)\right|.
\end{split}
\end{equation}

Our final goal is to get rid of the term proportional to $\partial^{\gamma}\Psi_{\tau}\partial_{\gamma}\Psi_{\tau}$ in the second line of~\eqref{afteraveragingtheenergy}. We consider a current $\mathbf{J}_{\mu}^{V,w,q,\varpi}$ current with $V = 0$, $w = -\frac{1}{2}\rho \sqrt{\frac{\Delta}{\Pi}}\chi_{t}^2\xi$, $q = \frac{1}{4}\nabla_{\mu}\left(\rho \sqrt{\frac{\Delta}{\Pi}}\chi_{t}^2\xi\right)$, and $\varpi = 0$. We add the corresponding divergence identity to~\eqref{afteraveragingtheenergy} and obtain~\eqref{wowthisistheresultofavglagr}.  Note that there are no boundary terms on the horizon or along $\Sigma_{t^*}$ for $t^* \to \pm \infty$ since $\xi\chi_{t}$ vanishes there.

\end{proof}

It will also be useful to have a version of Lemma~\ref{wowwhatanestimate} where we only time average the $\mathbf{J}^K$ energy estimate.
\begin{lemma}\label{onlyaverageKest}Let $\psi_{\tau}$ be as in the statement of Theorem~\ref{actuallywehatweneedformain}. Then we have, for every $t \in (10,\tau-10)$, 
\begin{equation}\label{averagedenergyyay}\begin{split}
\int_{-\infty}^{\infty}\chi_t^2(s)\left(\int_{\Sigma_s}\mathbf{J}^K_{\mu}n^{\mu}_{\Sigma_s}\right)\, ds &\leq \int_{-\infty}^{\infty}\chi^2(s)\left(\int_{\Sigma_s}\mathbf{J}^K_{\mu}n^{\mu}_{\Sigma_s}\right)\, ds + \left|\int_{\mathcal{M}_K}\tilde{\chi}^2_t{\rm Re}\left(\mathscr{F}\overline{K\Psi_{\tau}}\right)\right|
\\ &\qquad \qquad  + \left|\int_{\mathcal{M}_K}\tilde{\chi}^2_t{\rm Re}\left(\tilde{\mathscr{F}}\overline{K\Psi_{\tau}}\right)\right|.
\end{split}
\end{equation}
Here the currents $\mathbf{J}^K$ are defined with respect to $\Psi_\tau$.
\end{lemma}
\begin{proof}This follows from the proof of Lemma~\ref{wowwhatanestimate}  by integrating \eqref{thefirstenergyest} in $s$ against the bump-function $\chi^2(s)$.
\end{proof}

We next check that $\frac{2Mar}{\Pi}$ is globally dominated by $\frac{a}{2Mr_+}$ so that we can eventually apply Lemma~\ref{itsformallynoactuallypositive} to the estimate from Lemma~\ref{wowwhatanestimate}.
\begin{lemma}\label{somemono}We have that $\frac{2Mar}{\Pi} \leq \frac{a}{2Mr_+}$ for all $r \in [r_+,\infty)$.
\end{lemma}
\begin{proof}We have 
\[\frac{2Mar}{\Pi}|_{r=r_+} = \frac{a}{2Mr_+}.\]
Thus it suffices to show that $\frac{r}{\Pi}$ is non-increasing as a function of $r$. We compute
\begin{equation}\begin{split}
\frac{\partial h}{\partial r} &= \frac{(r^2+a^2)(-3r^2+a^2) + a^2\sin^2\theta\left(r^2-a^2\right)}{\Pi^2}
\\ &\leq \frac{(r^2+a^2)(-3r^2+2a^2)}{\Pi^2}
\\ &< 0.
\end{split}
\end{equation}
Here we have used the inequalities $\sin^2\theta \leq 1$, $r^2-a^2 \leq r^2+a^2$, and $r \geq r_+ > M > a$.
\end{proof}

We close the subsection with a straightforward combination of our physical space currents and Fourier analysis results.
\begin{lemma}\label{itallstartsherewow}Let $\psi_{\tau}$ be as in the statement of Theorem~\ref{actuallywehatweneedformain}. Then we have
\begin{equation}
    \begin{split}
   & {\rm sup}_{t \in (10,\tau-10)}\int_{-\infty}^{\infty}\chi_t^2(s)\left(\int_{\Sigma_s}\left[\left(1-\xi\right)\mathbf{J}^K_{\mu}n^{\mu}_{\Sigma_s} + \xi \mathbf{J}^N_{\mu}n^{\mu}_{\Sigma_s}\right]\right)\, ds \lesssim
 \\ &\qquad {\rm sup}_{t \in (10,\tau-10)}\left|\int_{\mathcal{M}_K}\tilde{\chi}_t^2{\rm Re}\left(\mathscr{F}\overline{K\Psi_{\tau}}\right)\right|+ {\rm sup}_{t \in (10,\tau-10)}\left|\int_{\mathcal{M}_K}{\rho \sqrt{\frac{\Delta}{\Pi}}}\chi_t^2\xi{\rm Re}\left(\mathscr{F}\overline{\Psi_{\tau}}\right)\right|
  \\ &\qquad + {\rm sup}_{t \in (10,\tau-10)}\left|\int_{\mathcal{M}_K}\tilde{\chi}_t^2{\rm Re}\left(\tilde{\mathscr{F}}\overline{K\Psi_{\tau}]}\right)\right|+ {\rm sup}_{t \in (10,\tau-10)}\left|\int_{\mathcal{M}_K}{\rho \sqrt{\frac{\Delta}{\Pi}}}\chi_t^2\xi{\rm Re}\left(\tilde{\mathscr{F}}\overline{\Psi_{\tau}}\right)\right|
 \\ &\qquad +A_{\rm high}^{-1} \sup_{t > 10}\int_{-\infty}^{\infty}\chi_t^2(s)E_0\left[\psi_{\tau}\right]\left(s\right)\, ds + E_0\left[\psi_{\tau}\right]\left(0\right).
    \end{split}
\end{equation}
   Here the currents are defined with respect to $\Psi$.
\end{lemma}
\begin{proof}

Note first that, using Lemma~\ref{somemono} in conjunction with Lemma~\ref{itsformallynoactuallypositive} yields
\[\begin{split}
&\int_{\mathcal{M}_K}{\rho \sqrt{\frac{\Delta}{\Pi}}}\chi_t^2\xi\Bigg(\rho^{-2}\frac{\Pi}{\Delta}{\rm Re}\left(\left(T\Psi_{\tau} + \frac{a}{2Mr_+}\Phi\Psi_{\tau}\right)\left(\overline{T\Psi_{\tau} + \frac{2Mar}{\Pi}\Phi\Psi_{\tau}}\right) \right) \\ &\geq c\int_{\mathcal{M}_K}\chi_t^2\xi\left(|T\Psi_{\tau}|^2 + |\Phi\Psi_{\tau}|^2\right) - CA_{\rm high}^{-1}\sup_{t^* \in \mathbb{R}}\int_{\mathcal{R}\left(t^*-1/2, t^*+1/2\right)}\xi\left( |T\Psi_\tau|^2 + |\Phi \Psi_\tau|^2\right).
\end{split}\]

 Due to \eqref{JKfluxFull}, there is a small constant $\tilde\delta$ so that 
\[
|T\Psi_\tau|^2 + |\Phi \Psi_\tau|^2 + \tilde\delta \mathbf{J}^K_{\mu}n^{\mu}_{\Sigma_t} \gtrsim \mathbf{J}^N_{\mu}n^{\mu}_{\Sigma_t}
\]
on the support of $\xi$.

We now add a sufficiently small constant times the estimate of Lemma~\ref{onlyaverageKest} to the estimate from Lemma~\ref{wowwhatanestimate}. Then we may apply Lemma~\ref{nothingbyT} and the fact that $\epsilon$ is small enough, to obtain that 
\begin{equation}\label{blahbakbh2342}
    \begin{split}
   & {\rm sup}_{t \in (10,\tau-10)}\int_{-\infty}^{\infty}\chi_t^2(s)\left(\int_{\Sigma_s}\left[\left(1-\xi\right)\mathbf{J}^K_{\mu}n^{\mu}_{\Sigma_t} + \xi \mathbf{J}^N_{\mu}n^{\mu}_{\Sigma_t}\right]\right)\, ds \lesssim
\\&\qquad \int_{-\infty}^{\infty}\chi^2(s)\left(\int_{\Sigma_s}\left[\left(1-\xi\right)\mathbf{J}^K_{\mu}n^{\mu}_{\Sigma_s} + \xi \mathbf{J}^N_{\mu}n^{\mu}_{\Sigma_s}\right]\right)\, ds 
\\ &\qquad + {\rm sup}_{t \in (10,\tau-10)}\left|\int_{\mathcal{M}_K}\tilde{\chi}^2_t{\rm Re}\left(\mathscr{F}\overline{K\Psi_{\tau}}\right)\right|+ {\rm sup}_{t \in (10,\tau-10)}\left|\int_{\mathcal{M}_K}{\rho \sqrt{\frac{\Delta}{\Pi}}}\chi_t^2\xi{\rm Re}\left(\mathscr{F}\overline{\Psi_{\tau}}\right)\right|
\\ &\qquad + {\rm sup}_{t \in (10,\tau-10)}\left|\int_{\mathcal{M}_K}\tilde{\chi}_t^2{\rm Re}\left(\tilde{\mathscr{F}}\overline{K\Psi_{\tau}}\right)\right|+ {\rm sup}_{t \in (10,\tau-10)}\left|\int_{\mathcal{M}_K}{\rho \sqrt{\frac{\Delta}{\Pi}}}\chi_t^2\xi{\rm Re}\left(\tilde{\mathscr{F}}\overline{\Psi_{\tau}}\right)\right|
\\ &\qquad + A^{-1}_{\rm high}\sup_{t \in\mathbb{R}}\int_{-\infty}^{\infty}\chi_t^2(s)\left(\int_{\Sigma_s} \xi \mathbf{J}^N_{\mu}n^{\mu}_{\Sigma_s}\right)\, ds.
    \end{split}
\end{equation}
To control the first term, we use Lemma~\ref{totheoriginal} and Theorem~\ref{extendtoeverywhere} concerning the extension operator $\mathscr{E}$:
\begin{equation}\label{removeHT}
\begin{split}
& \int_{-\infty}^{\infty}\chi^2(s)\left(\int_{\Sigma_s}\left[\left(1-\xi\right)\mathbf{J}^K_{\mu}n^{\mu}_{\Sigma_s} + \xi \mathbf{J}^N_{\mu}n^{\mu}_{\Sigma_s}\right]\right)\, ds \lesssim \int_{-\infty}^{\infty}\chi^2(s) E_0[\Psi_\tau](s) ds 
\\ & \lesssim \int_{-\infty}^{\infty}\chi^2(s) E_0[\psi_\tau](s) +A_{\rm high}^{-2}\sup_{t^*\in\mathbb{R}}\int_{t^*-1/2}^{t^*+1/2} E_0[\psi_\tau](s) ds \lesssim E_0[\psi_\tau](0) 
\\ & +A_{\rm high}^{-1} \sup_{t > 10}\int_{-\infty}^{\infty}\chi_t^2(s)E_0\left[\psi_{\tau}\right]\left(s\right)\, ds
\end{split}
\end{equation}
while for the last term we simply note that Theorem~\ref{extendtoeverywhere} implies that
\[
A^{-1}_{\rm high} \sup_{t \in\mathbb{R}}\int_{-\infty}^{\infty}\chi_t^2(s)\left(\int_{\Sigma_s} \xi \mathbf{J}^N_{\mu}n^{\mu}_{\Sigma_s}\right)\, ds \lesssim 
A_{\rm high}^{-1} \left(\sup_{t > 10}\int_{-\infty}^{\infty}\chi_t^2(s)E_0\left[\psi_{\tau}\right]\left(s\right)\, ds + E_0\left[\psi_{\tau}\right]\left(0\right)\right)
\]

\end{proof}

\subsection{Higher order commutation, elliptic estimates, and the redshift estimates}
Next, we observe that we may commute our estimate from Lemma~\ref{itallstartsherewow} with $T$.
\begin{lemma}\label{Tgoesthroughwow}Let $\psi_{\tau}$ be as in the statement of Theorem~\ref{actuallywehatweneedformain}.  Then, for any $k \geq 0$, we have 
\begin{equation}\label{whatactcommwouldliketoshowihope}
    \begin{split}
   & {\rm sup}_{t\in (10,\tau-10)}\int_{-\infty}^{\infty}\chi_t^2(s)\left(\int_{\Sigma_t}\left[\left(1-\xi\right)\mathbf{J}^K_{\mu}\left[T^k\Psi_{\tau}\right]n^{\mu}_{\Sigma_t} + \xi \mathbf{J}^N_{\mu}\left[T^k\Psi_{\tau}\right]n^{\mu}_{\Sigma_t}\right]\right)\, ds \lesssim
 \\ & + {\rm sup}_{t \in (10,\tau-10)}\left|\int_{\mathcal{M}_K}\tilde{\chi}_t^2{\rm Re}\left(T^k\mathscr{F}\overline{KT^k\Psi_{\tau}}\right)\right|+ {\rm sup}_{t \in (10,\tau-10)}\left|\int_{\mathcal{M}_K}{\rho \sqrt{\frac{\Delta}{\Pi}}}\chi_t^2\xi{\rm Re}\left(T^k\mathscr{F}\overline{T^k\Psi_{\tau}}\right)\right|
  \\ & + {\rm sup}_{t \in (10,\tau-10)}\left|\int_{\mathcal{M}_K}\tilde{\chi}_t^2{\rm Re}\left(T^k\tilde{\mathscr{F}}\overline{KT^k\Psi_{\tau}}\right)\right|+ {\rm sup}_{t \in (10,\tau-10)}\left|\int_{\mathcal{M}_K}{\rho \sqrt{\frac{\Delta}{\Pi}}}\chi_t^2\xi{\rm Re}\left(T^k\tilde{\mathscr{F}}\overline{T^k\Psi_{\tau}}\right)\right|
 \\ & +A_{\rm high}^{-1} \sup_{t > 10}\int_{-\infty}^{\infty}\chi_t^2(s)E_k\left[\psi_{\tau}\right]\left(s\right)\, ds + E_k\left[\psi_{\tau}\right]\left(0\right).
    \end{split}
\end{equation}
\end{lemma}
\begin{proof}

We simply apply Lemma~\ref{itallstartsherewow} and note that, since $g$ is stationary, we have $$[\Box_{\tau}, T] = [\Box_g, T] = 0$$
on the supports of $\tilde{\chi}_t$ and $\chi_t$. 

\end{proof}

Away from the horizon, the vector fields $T$ and $\Phi$ always span a timelike direction. Thus, keeping in mind also Lemma~\ref{phibyT}, we may use Lemma~\ref{Tgoesthroughwow} to control $E_{k,\{r \geq r_++\delta\}}$ for any $\delta > 2\delta_0$.
\begin{lemma}
Let $\psi_{\tau}$ be as in the statement of Theorem~\ref{actuallywehatweneedformain}. Then we have, for every  $k \geq 1$, and $\delta > 2\delta_0$ 
\begin{equation}
    \begin{split}
   & \qquad{\rm sup}_{t \in (10,\tau-10)}\int_{-\infty}^{\infty}\chi_t^2(s)E_{k,\{r \geq r_++\delta\}}\left[\Psi_{\tau}\right]\left(s\right)\, ds \lesssim_{k,\delta} \\& \qquad {\rm sup}_{t  \in (10,\tau-10)}\int_{-\infty}^{\infty}\chi_t^2(s)E_{k-2,\{r \geq r_++\delta/2\}}\left[\mathscr{F}\right]\left(s\right)\, ds
    \\& \qquad +{\rm sup}_{t \in (10,\tau-10)}\int_{-\infty}^{\infty}\chi_t^2(s)E_{k-2,\{r \geq r_++\delta/2\}}\left[\tilde{\mathscr{F}}\right]\left(s\right)\, ds 
    \\ & \qquad + \sup_{t \in (10,\tau-10)}\int_{\mathcal{M}_K \cap \{r \geq r_++\delta/2\}}\chi_t^2(s)\left|\mathscr{F}\right|^2 ds
   + \sup_{t \in (10,\tau-10)}\int_{\mathcal{M}_K \cap \{r \geq r_++\delta/2\}}\chi_t^2(s)\left|\tilde{\mathscr{F}}\right|^2 ds
 \\ &\qquad + \sup_{t \in (10,\tau-10)}\left|\int_{\mathcal{M}_K}\tilde{\chi}_t^2{\rm Re}\left(T^k\mathscr{F}\overline{KT^k\Psi_{\tau}}\right)\right|+ \sup_{t \in (10,\tau-10)}\left|\int_{\mathcal{M}_K}{\rho \sqrt{\frac{\Delta}{\Pi}}}\chi_t^2\xi{\rm Re}\left(T^k\mathscr{F}\overline{T^k\Psi_{\tau}}\right)\right|
  \\ &\qquad + \sup_{t \in (10,\tau-10)}\left|\int_{\mathcal{M}_K}\tilde{\chi}_t^2{\rm Re}\left(T^k\tilde{\mathscr{F}}\overline{KT^k\Psi_{\tau}}\right)\right|+ \sup_{t \in (10,\tau-10)}\left|\int_{\mathcal{M}_K}{\rho \sqrt{\frac{\Delta}{\Pi}}}\chi_t^2\xi{\rm Re}\left(T^k\tilde{\mathscr{F}}\overline{T^k\Psi_{\tau}}\right)\right|
 \\ &\qquad +A_{\rm high}^{-1} \sup_{t > 10}\int_{-\infty}^{\infty}\chi_t^2(s)E_k\left[\psi_{\tau}\right]\left(s\right)\, ds + E_k\left[\psi_{\tau}\right]\left(0\right),
    \end{split}
\end{equation}
where if $k = 1$, we drop the terms involving $E_{k-2}$.
\end{lemma}
\begin{proof}This is an immediate consequence of Lemma~\ref{Tgoesthroughwow}, Lemma~\ref{phibyT}, and the following elliptic estimate, which holds for any function $f$ and a sufficiently large constant $A$ and $k \geq 1$:
\begin{equation}\label{thisistheelltahtweuse}\begin{split}
&E_{k,\{r \geq r_++\delta\}}\left[f\right](t)\, \lesssim_k E_{k-2,\{r \geq r_++\delta/2\}}\left[\Box_{g_K}f\right]\left(t\right) + \int_{\Sigma_t \cap \{r \geq _+ + \delta/2\}}\left|\Box_{g_K}f\right|^2
\\ &\qquad + \sum_{l=0}^{k-1}\sum_{i+j= 0}^{k-l}E_{l,\{r_++\delta/2 \leq r \leq A\}}\left[T^i\Phi^jf\right]\left(t\right) + \sum_{j=0}^{k-1}\sum_{i= 0}^{k-j}E_{j, \{r \geq r_++ A\}}\left[T^if\right]\left(t\right),
\end{split}\end{equation}
where, if $k = 1$, we drop the first term on the right hand side of~\eqref{thisistheelltahtweuse}. 

Indeed,  we can write
\[
E_1\Psi_{\tau} = \Box_{g_K} \Psi_{\tau} + O(1) \sum_{i+j=1}^2 T^i\Phi^j \psi_{\tau} ,\quad r_++\delta/2 \leq r \leq \frac{A}2 
\]
\[
E_2\Psi_{\tau} = \Box_{g_K}\Psi_{\tau} + O(1) \sum_{i=1}^2 T^i \psi_{\tau} ,\quad \frac{A}4 \leq r 
\]
for elliptic operators $E_1$ and $E_2$  along $\Sigma_s$. 

The conclusion follows since in the region $t \in (10,\tau-10)$ we have that $\Box_{\tau} = \Box_{g_K}$.
\end{proof}

Finally, we may establish good estimates near the horizon by exploiting both the redshift multiplier and redshift commutator of~\cite{redshiftSchw,claylecturenotes}.
\begin{lemma}\label{redsfhiftit}Let $\psi_{\tau}$ be as in the statement of Theorem~\ref{actuallywehatweneedformain}.  Then, for any $k \geq 1$, we have 

\begin{equation}\label{EH}
    \begin{split}
   & {\rm sup}_{t \in (10,\tau-10)}\int_{-\infty}^{\infty}\chi_t^2(t)E_k\left[\Psi\right]\left(s\right)\, ds \lesssim_k {\rm sup}_{t \in (10,\tau-10)}\int_{-\infty}^{\infty}\chi_t^2(s)E_{k-2}\left[\mathscr{F}\right]\left(s\right)\, ds
   \\ & + \sup_{t \in (10,\tau-10)}\int_{\mathcal{M}_K }\chi_t^2(s)\left|\mathscr{F}\right|^2 + {\rm sup}_{t \in (10,\tau-10)}\left|\int_{\mathcal{M}_K}\tilde{\chi}_t^2{\rm Re}\left(T^k\mathscr{F}\overline{KT^k\Psi_{\tau}}\right)\right|
 \\ &+ \sup_{t > 0}\left|\int_{\mathcal{M}_K}{\rho \sqrt{\frac{\Delta}{\Pi}}}\chi_t^2\xi{\rm Re}\left(T^k\mathscr{F}T^k\overline{\Psi_{\tau}}\right)\right|
 +\sum_{j+l=1}^k\sup_{t \in (10,\tau-10) }\left|\int_{\mathcal{M}_K}\tilde{\chi}_t^2{\rm Re}\left(\left( T^jY^l\mathscr{F}\right)\overline{N\left(T^jY^l\Psi_{\tau}\right)}\right)\right|
 \\ &+A_{\rm high}^{-1} \sup_{t > 0}\int_{-\infty}^{\infty}\chi_t(s)E_k\left[\psi_{\tau}\right]\left(s\right)\, ds + E_k\left[\psi_{\tau}\right]\left(0\right)
 \\ &+{\sup}_{t \in (10,\tau-10)}\int_{-\infty}^{\infty}\chi_t^2(s)E_{k-2}\left[\tilde{\mathscr{F}}\right]\left(s\right)\, + \sup_{t \in (10,\tau-10)}\int_{\mathcal{M}_K }\chi_t^2(s)\left|\tilde{\mathscr{F}}\right|^2ds 
 \\ & + {\sup}_{t \in (10,\tau-10)}\left|\int_{\mathcal{M}_K}{\rho \sqrt{\frac{\Delta}{\Pi}}}\chi_t^2\xi{\rm Re}\left(T^k\tilde{\mathscr{F}}\overline{T^k\Psi_{\tau}}\right)\right|
+ {\sup}_{t \in (10,\tau-10)}\left|\int_{\mathcal{M}_K}\tilde{\chi}_t^2{\rm Re}\left(T^k\tilde{\mathscr{F}}\overline{KT^k\Psi_{\tau}}\right)\right|
 \\ & + \sum_{j+l=1}^k\sup_{t \in (10,\tau-10) }\left|\int_{\mathcal{M}_K}\tilde{\chi}_t^2{\rm Re}\left(\left( T^jY^l\tilde{\mathscr{F}}\right)\overline{N\left(T^jY^l\Psi_{\tau}\right)}\right)\right|,
    \end{split}
\end{equation}

where, if $k = 1$, we drop the terms involving $E_{k-2}$.
\end{lemma}

\begin{proof}This follows from a straightforward adaption of the now standard arguments for upgrading a degenerate boundedness statement to non-degenerate boundedness statement (see Section 3.3 of~\cite{claylecturenotes}) to the time-averaged setting.
\end{proof}
\begin{remark}If $F$ is supported in  a region $\{r \geq r_+ + \tilde{\delta}\}$ is it clear that we may drop the terms on the right hand side involving the vector field $N$ and $Y$.
\end{remark}

Finally, we may give the proof of Theorem~\ref{actuallywehatweneedformain}.
\begin{proof}
In view of Lemma~\ref{redsfhiftit}, it only remains to control the terms involving $\tilde{\mathscr{F}}$, namely the last three lines of \eqref{EH}. 

 For the terms involving $\chi_t^2$, we simply use Theorem~\ref{extendtoeverywhere}, Lemma~\ref{pseudolocforthewin} and Cauchy Schwarz. 

The last two terms are more difficult, since the support of $\tilde{\chi}_t$ can potentially be large.  Let us control the second to last term. We partition $I$, the support of $\tilde{\chi}_t$, into intervals $I_j$ with disjoint interiors so that $1/2\leq |I_j|\leq 1$. Using the fundamental theorem of calculus, we obtain
\[
\|T^k\tilde{\mathscr{F}}\|_{L^{\infty}(I_j)L^2(\Sigma_s)} \lesssim \int_{I_j} \|T^k\tilde{\mathscr{F}}(s)\|_{L^2(\Sigma_s)} + \|T^{k+1}\tilde{\mathscr{F}}(s)\|_{L^2(\Sigma_s)} ds
\]
We can thus estimate
\[\begin{split}
& \int_{I_j} \int_{\Sigma_s}\left|T^k\tilde{\mathscr{F}}\right|\left|K T^k \Psi_\tau\right| d\Sigma_s ds \leq \|T^k\tilde{\mathscr{F}}\|_{L^{\infty}(I_j)L^2(\Sigma_s)}
\int_{I_j} \|K T^k \Psi_\tau\|_{L^2(\Sigma_s)}
\\ & \lesssim \|T^k\tilde{\mathscr{F}}\|_{L^{\infty}(I_j)L^2(\Sigma_s)} \left(\int_{I_j} E_k[\Psi_{\tau}](s)ds\right)^{1/2}  
\end{split}\]

We now sum over $j$. We obtain
\begin{equation}\label{Ftilde}
\begin{split}
{\rm sup}_{t \in (10,\tau-10)}\left|\int_{\mathcal{M}_K}\tilde{\chi}_t^2{\rm Re}\left(T^k\tilde{\mathscr{F}}\overline{KT^k\Psi_{\tau}}\right)\right| \\ \lesssim \left(\int_I \|T^k\tilde{\mathscr{F}}(s)\|_{L^2(\Sigma_s)} + \|T^{k+1}\tilde{\mathscr{F}}(s)\|_{L^2(\Sigma_s)} ds\right) \left(\sup_{t^*>5}{\int_{t^*-1/2}^{t^*+1/2}}E_k[\Psi_{\tau}](s)ds\right)^{1/2}
\end{split}
\end{equation}
Using \eqref{L1disj} and Theorem~\ref{extendtoeverywhere}, we obtain
\[
\int_I \|T^k\tilde{\mathscr{F}}(s)\|_{L^2(\Sigma_s)} + \|T^{k+1}\tilde{\mathscr{F}}(s)\|_{L^2(\Sigma_s)} ds \lesssim 
A_{\rm high}^{-1} E_k\left[\psi\right]\left(0\right)\]
while \eqref{astartbutcandobetter} yields
\[
\left(\sup_{t^*>5}{\int_{t^*-1/2}^{t^*+1/2}}E_k[\Psi_{\tau}](s)ds\right)^{1/2} \lesssim \left(\sup_{t^*>5}{\int_{t^*-1/2}^{t^*+1/2}}E_k[\psi_{\tau}](s)ds\right)^{1/2}
\]
A similar argument controls the last term.

\end{proof}

\section{Proof of Theorem~\ref{theoc1pert}}\label{proofTheoc1}
In this section we will provide the proof of Theorem~\ref{theoc1pert}. We start with a useful remark.
\begin{remark}\label{axiandnotaxi}
For $\psi$ as in the statement of Theorem~\ref{theoc1pert}, we may always split $\psi = \psi_1 + \psi_2$ where $\Phi \psi_1 = 0$ and $\psi_1 = \frac{1}{2\pi}\int_0^{2\pi} \psi\, d\phi$. In view of the $L^2$-orthogonality of Fourier series, it suffices to establish Theorem~\ref{theoc1pert} separately for $\psi_1$ and $\psi_2$.
\end{remark}
For $\psi_1$, we have the following:
\begin{lemma}\label{easyifaxi}Let $\psi$ as in the statement of Theorem~\ref{theoc1pert} and assume additionally that $\Phi\psi = 0$. Then, for every positive integer $k \geq 1$, we have
\[\sup_{t^* \geq 0}E_k\left[\psi\right](t^*) \lesssim_k E_k\left[\psi\right]\left(0\right).\]
\end{lemma}
\begin{proof}This is a standard result so we will be brief in our explanation.

We note that in view of the formulas from Section~\ref{stancurr} and the fact that $g\left(\nabla t^*,\Phi\right) = 0$, we have that for any constant $\alpha \in \mathbb{R}$:
\[\mathbf{J}^{T+\alpha \Phi}_{\mu}n^{\mu}_{\Sigma_s} =\mathbf{J}^T_{\mu}n^{\mu}_{\Sigma_s}.\]
In particular, since $T$ and $\Phi$ always span a timelike direction away from the event horizon, for any function $f$ and $\delta > 0$, we have that
\[\int_{\Sigma_s}J^T_{\mu}\left[f\right]n^{\mu}_{\Sigma_s} \gtrsim_{\delta} E_{0,r \geq r_+ + \delta}\left[f\right].\]
We also know that for any integer $i \in [0,k]$, we have $\Box_g T^i\psi = 0$. We thus obtain by applying $\textbf{J}^T$-energy estimates that
\[\sum_{i=0}^k\sup_{t^* \geq 0}E_{0,r \geq r_++\delta}\left[T^i\psi\right](t^*) \lesssim E_0\left[T^i\psi\right](0).\]
Elliptic estimates (as in Lemma~\ref{ellipt} below) then yield
\[\sup_{t^* \geq 0}E_{k,r \geq r_++\delta}\left[\psi\right](t^*) \lesssim_{k,\delta} E_k\left[\psi\right](0).\]
Finally, for $\delta > 0$ sufficiently small, we may use the redshift multiplier and commutator to extend our estimate all the way to $r = r_+$ (see Section 3.3 of~\cite{claylecturenotes}).
\end{proof}

In view of Remark~\ref{axiandnotaxi} and Lemma~\ref{easyifaxi} we introduce the following convention:
\begin{convention}\label{wehavem}Unless said otherwise, throughout Section~\ref{proofTheoc1} we assume that all functions of $\phi$ have mean zero on $[0,2\pi]$ and hence have a vanishing projection onto the $m = 0$ Fourier mode. Similarly, unless said otherwise, when defining operators via~\eqref{defQa}, we will assume $m \neq 0$.
\end{convention}

We also choose the constant $\delta_0$ so that $r_+ + \delta_0 < r_{\rm low}$.

\subsection{Some Fourier decompositions}
As in the proof of Theorem~\ref{mainLinEstTheo}, we will need to define some operators via Fourier analysis with respect to $(t^*,\phi^*)$. 

Our definitions in this section will depend on a large positive constant $A_{\rm high}$. We start with an operator which corresponds to projection to the set of bounded frequencies.
\begin{definition}We define $\mathcal{P}_{\rm bound} : \mathbb{H}_K \to \mathbb{H}_K$ by $\mathcal{P}_{\rm bound} = Q_{a_{\rm bound}}$ (see~\eqref{defQa}) where
\begin{equation}\label{thisisabound}
a_{\rm bound}\left(\omega,m\right) = q\left(\frac{\omega^2+m^2}{A^2_{\rm high}}\right),
\end{equation}
where $1$ denotes the indicator function, and $q(x)$ is a smooth function which is identically $1$ for $|x| \leq 1$ and identically $0$ for $|x| \geq 2$.
\end{definition}

Our next definition will define a family of operators which involve a mild modification to the family of operators $\{P_n\}_{n=1}^{N_s}$ defined in Section 3.4.1 of~\cite{blackboxlargea}. These operators serve to project\footnote{We note that they are not literal projections in that $P_n^2 \neq P_n$ in general.} the solution to suitable subsets of superradiant frequencies.
\begin{definition}\label{thisdefinedpnhigh}For each $n = 1,\cdots,N_s$ (where $N_s$ is as in Section 3.4.1 of~\cite{blackboxlargea}) we define $\mathcal{P}_{n,{\rm high}} : \mathbb{H}_K \to \mathbb{H}_K$ by $\mathcal{P}_{n,{\rm high}} = Q_{a_{n,{\rm high}}}$ (see~\eqref{defQa}) where
\[a_{n,{\rm high}} \doteq \left(1-q\left(\frac{\omega^2+m^2}{A_{\rm high}^2}\right)\right)\overset{\musEighth}{\rho_n}\left(\frac{\omega}{am}\right),\]
and $\overset{\musEighth}{\rho_n}$ is defined as in Section 3.4.1 of~\cite{blackboxlargea}.
\end{definition}
In this paper we will not need the precise definition of the functions $\overset{\musEighth}{\rho_n}$. We will discuss later certain useful estimates associated to the operators $\mathcal{P}_{n,{\rm high}}$ (see Section~\ref{gensuper} below). However it is useful to already note here that associated to each function $\overset{\musEighth}{\rho_n} : \mathbb{R} \to \mathbb{R}$ is a bounded open interval $I_n \subset \mathbb{R}$ so that ${\rm supp}\left(\overset{\musEighth}{\rho_n}\right) \subset I_n$. Moreover, we have the following important fact which is an immediate consequence of the construction of the operators $P_n$ in Section 3.4.1 of~\cite{blackboxlargea}.
\begin{lemma}There exists a choice of functions $p$ and $q_{\delta_1}$ and constants $\delta_1$ and $\delta_2$ in Definition~\ref{hightrappedpart} so that for every sufficiently large choice of $A_{\rm high}$ we have that
\begin{equation}\label{theygivetheidentity}
\mathcal{P}_{\rm bound} + \sum_{n=1}^{N_s}\mathcal{P}_{n,{\rm high}} + \mathcal{P}_{HT} = {\rm Id},
\end{equation}
where ${\rm Id}$ stands for the identity operator on functions in $\mathbb{H}_K$ which have a vanishing projection onto the $m = 0$ Fourier mode. 
\end{lemma}

\begin{proof}
The intervals $I_n$ are chosen (see Appendix A.1.5 of of~\cite{blackboxlargea}) so that
\[
\bigcup_{n=1}^{N_s} {\rm supp}\left(\overset{\musEighth}{\rho_n}\right) = [c, d] \subset \bigcup_{n=1}^{N_s} I_j = \left(-\epsilon, \frac{1}{2Mr_+}+\epsilon\right)
\]
for some suitably small $\epsilon>0$ and $c<0$, $d>\frac{1}{2Mr_+}$. We may then pick 
\[
\delta_1 = a\left(d-(2Mr_+)^{-1}\right), \quad \delta_2 = \frac{a|c|}{2}
\]
and $p$, $q_{\delta_1}$, $\overset{\musEighth}{\rho_n}$ to be a suitable partition of unity. 
\end{proof}

\subsection{Elliptic estimates}
In this section we state a useful elliptic estimate whose straightforward proof we omit.

\begin{lemma}\label{ellipt}Let $k\geq 1$ and   $\mathcal{Q} \in  \left\{\Box_g,\Box_{g_K},\Box_{\tau}\right\}$. There exists a sufficiently large constant $A$ so that the following two elliptic estimates hold:
\begin{enumerate}
    \item Suppose that $f$ vanishes for sufficiently large $r$. Then, for every $\delta > 0$, 
    \begin{equation}\begin{split}
&E_{k,r \geq r_+ + \delta}\left[f\right](\tilde{t})\, dt \lesssim_{k,\delta} E_{k-2,r \geq r_+ + \delta/2}\left[\mathcal{Q}f\right]\left(\tilde{t}\right)+\int_{r \geq r_++\delta/2}\left|\mathcal{Q}f|_{t^* = \tilde{t}}\right|^2
\\ &\qquad + \sum_{i+j= 0}^{k+1}\int_{r_+ + \delta/2 \leq r \leq A} \left|T^i\Phi^jf|_{t^* = \tilde{t}}\right|^2 + \sum_{i= 1}^{k+1}\int_{r \geq A}\left|T^if|_{t^* = \tilde{t}}\right|^2.
\end{split}\end{equation}
\item Let $I_1,I_2 \subset (r_+,\infty)$ be open bounded intervals so that $\overline{I_2} \subset I_1$. Then 
 \begin{equation}\begin{split}
&E_{k,r \in I_2}\left[f\right](\tilde{t})\, dt \lesssim_{k,\delta,I_1,I_2} E_{k-2,r \in I_1}\left[\mathcal{Q}f\right]\left(\tilde{t}\right)+\int_{r \in I_1}\left|\mathcal{Q}f|_{t^* = \tilde{t}}\right|^2
\\ &\qquad + \sum_{i+j= 0}^{k+1}\int_{r \in I_1 \cap \{r \leq A\} }\left|T^i\Phi^jf|_{t^* = \tilde{t}}\right|^2 + \sum_{i= 0}^{k+1}\int_{r \in I_1 \cap \{r\geq A\}}\left|T^if|_{t^* = \tilde{t}}\right|^2.
\end{split}\end{equation}
\end{enumerate}
If $k = 1$ then we drop the terms involving $E_{k-2}$ from the right hand side.
\end{lemma}

\subsection{Bounded frequency estimates}
We start by providing estimates for the bounded frequencies. The main result of the section is the following:
\begin{lemma}\label{enerellipboundedfreq}Let $\psi$ be as in the statement of Theorem~\ref{theoc1pert}. Choose $\tau \gg 1$ and let $\psi_{\tau}$ be defined by Definition~\ref{defpsitau}. Set $\Psi_{\tau} \doteq \mathcal{P}_{\rm bound}\psi_{\tau}$. Let $k \geq 1$ be an integer. Then we have that
\begin{equation}\label{frompart3foruboundi120uo3j}
\begin{split}
&A_{\rm high}^{-2k-2}\sup_{t \in (10,\tau-10)}\int_{-\infty}^{\infty}\chi_t^2(s)E_{k,r \geq r_++\delta_0}\left[\Psi_{\tau}\right]\left(s\right)\, ds  
+ \int_{\mathcal{M}_K} \tilde{\chi}^2_{\tau-10,10}(s)\left|\Psi_{\tau}\right|^2r^{-4}
\\ &\qquad + A_{\rm high}^{-2k-2}\int_{-\infty}^{\infty}\tilde{\chi}^2_{\tau-10,10}(s)E_{k,\{r_{\rm low} \leq r \leq R_{\rm high}\}}\left[\Psi_{\tau}\right](s)\, ds 
\\ &\lesssim_{k}  A_{\rm high}^4\int_{-\infty}^{\infty}\chi^2_{10}(t)E_0\left[\psi_{\tau}\right](s)\, ds 
  + A_{\rm high}^{-2k-4} E_{k-1}\left[\psi\right](0) + A_{\rm high}^2 E_0\left[\psi\right](0)
\\ &\qquad + A_{\rm high}^4 \epsilon^2 \sup_{t^*} \int_{t^*-1/2}^{t^*+1/2} E_{1,\{r_{\rm low} \leq r \leq R_{\rm high}\}}\left[\psi_{\tau}\right]\left(s\right)\, ds
\\ &\qquad +\left(\epsilon + A_{\rm high}^{-2k-4}\right) \sup_{t^*} \int_{t^*-1/2}^{t^*+1/2}E_{k}\left[\psi_{\tau}\right](s)\, ds +\left(\epsilon+A_{\rm high}^{-6}\right)\sup_{t^*} \int_{t^*-1/2}^{t^*+1/2}E_{1}\left[\psi_{\tau}\right](s)\, ds.
\end{split} 
\end{equation}

\end{lemma}

We will use the following theorem, which is a consequence of the main boundedness and integrated energy decay results proved in~\cite{waveKerrlargea} (see also Theorem D.1 of~\cite{blackboxsmalla} and the discussion in Section 3.1 of~\cite{blackboxlargea}).
\begin{theorem}\label{mainresulfromwavekerrlargea}\cite{waveKerrlargea} Let   $0 < \mu \ll 1$ and $\psi : J^+\left(\Sigma_0\right) \to \mathbb{C}$ solve $\Box_{g_K}\psi = F$ where $F$ is smooth and compactly supported in the spacetime region $\{r_{\rm low} \leq r \leq R_{\rm high}\} \cap J^+\left(\Sigma_0\right)$ and where  $E_0\left[\psi\right]\left(0\right) < \infty$. Then, for every $0 < \tau_0 < \tau_1$, we have
\begin{equation}
\begin{split}
&E_0\left[\psi\right]\left(\tau_1\right)  +\int_{\mathcal{R}\left(\tau_0,\tau_1\right)} \left|\psi\right|^2r^{-4}
\\ &\qquad \lesssim_{R_{\rm high}} \mu^{-1}\int_{\mathcal{R}\left(\tau_0,\tau_1\right)} \left|F\right|^2+ \mu \int_{\tau_0}^{\tau_1}E_{0,\{r_{\rm low} \leq r \leq R_{\rm high}\}}\left[\psi\right](\tau)\, d\tau +E_0\left[\psi\right](\tau_0).
\end{split}    
\end{equation}
\end{theorem}
\begin{proof}See Remark 3.2.1 and Appendix D of~\cite{blackboxsmalla}.
\end{proof}

It will be convenient to have a time averaged version of the estimate of Theorem~\ref{mainresulfromwavekerrlargea}. Taking $\tau_0 =10$, integrating against the bump function $\chi_{t}^2(s)$, using the monotonicity of $\tilde{\chi}_{\tau,A}$ in $\tau$ and local-in-time estimates immediately leads to the following estimate.

\begin{corollary}\label{timeaveragethepart3est}
Let   $0 < \mu \ll 1$ and $\psi : J^+\left(\Sigma_0\right) \to \mathbb{C}$ solve $\Box_{g_K}\psi = F$ where $F$ is smooth and compactly supported in spacetime region $\{r_{\rm low} \leq r \leq R_{\rm high}\} \cap J^+\left(\Sigma_0\right)$ and where  $E_0\left[\psi\right]\left(0\right) < \infty$. Then, for every $\tau_1 \gg 1$, we have
\begin{equation}
\begin{split}
&\sup_{t \in (10,\tau_1)}\int_{-\infty}^{\infty}\chi_{t}^2(s)E_0\left[\psi\right]\left(s\right)\, ds  +\int_{\mathcal{M}_K} \tilde{\chi}_{\tau_1,10}^2\left|\psi\right|^2r^{-4}
\\ &\qquad \lesssim_{R_{\rm high}} \mu^{-1}\int_{\mathcal{M}_K} \tilde{\chi}^2_{\tau_1,10}\left|F\right|^2+ \mu \int_{-\infty}^{\infty}\tilde{\chi}^2_{\tau_1,10}E_{0,\{r_{\rm low} \leq r \leq R_{\rm high}\}}\left[\psi\right](s)\, ds
\\ &\qquad +\int_{-\infty}^{\infty}\chi_{10}^2(s)E_0\left[\psi\right](s)\, ds.
\end{split}    
\end{equation}
\end{corollary}

We now apply Corollary~\ref{timeaveragethepart3est} to $\mathcal{P}_{\rm bound}\psi_{\tau}$.
\begin{lemma}\label{applythatstufftotheboun}Let $\psi$ be as in the statement of Theorem~\ref{theoc1pert}. Choose $\tau \gg 1$ and let $\psi_{\tau}$ be defined by Definition~\ref{defpsitau}. Set $\Psi_{\tau} \doteq \mathcal{P}_{\rm bound}\psi_{\tau}$. Then we have that, for any small constant $\mu > 0$,
\begin{equation}\label{frompart3forubound}
\begin{split}
&\sup_{t \in (10,\tau-10)}\int_{-\infty}^{\infty}\chi_{t}^2(s)E_0\left[\Psi_{\tau}\right]\left(s\right)\, ds  
+ \int_{\mathcal{M}_K} \tilde{\chi}^2_{\tau-10, 10}(s)\left|\Psi_{\tau}\right|^2r^{-4}
\\ &\lesssim (\mu+\epsilon^2)\int_{-\infty}^{\infty}\tilde{\chi}^2_{\tau-10,10}(s)E_{1,\{r_{\rm low} \leq r \leq R_{\rm high}\}}\left[\Psi_{\tau}\right](s)\, ds  + \mu^{-1}\int_{-\infty}^{\infty}\chi^2_{10}(s)E_0\left[\Psi_{\tau}\right](s)\, ds 
\\ &\qquad + \mu^{-1} \epsilon^2 \sup_{t^*} \int_{t^*-1/2}^{t^*+1/2} E_{1,\{r_{\rm low} \leq r \leq R_{\rm high}\}}\left[\psi_{\tau}\right]\left(s\right)\, ds+ \mu^{-1}A_{\rm high}^{-2}E_0\left[\psi\right](0).
\end{split}    
\end{equation}
\end{lemma}
\begin{proof}
We have
\[\Box_{g_K}\Psi_{\tau} = \left(\Box_{g_K}-\Box_{\tau}\right)\Psi_{\tau}-\left[\mathcal{P}_{\rm bound},\Box_{\tau}\right]\psi_{\tau} + \mathcal{P}_{\rm bound}\tilde{F}.\]
Our goal is to estimate all the terms on the right hand side in $L^2$ of spacetime against $\tilde{\chi}^2_{\tau-10, 10}(s)$.

It is immediate that 
\[\int_{\mathcal{M}_K}\tilde{\chi}^2_{\tau-10, 10}(s)\left|\left(\Box_{g_K}-\Box_{\tau}\right)\Psi_{\tau}\right|^2 \lesssim \epsilon^2\int_{-\infty}^{\infty}\tilde{\chi}^2_{\tau-10, 10}(s)E_{1,\{r_{\rm low} \leq r \leq R_{\rm high}\}}\left[\Psi_{\tau}\right](s)\, ds.\]

We next turn to the term $\left[\mathcal{P}_{\rm bound},\Box_{\tau}\right]\psi_{\tau}$. We have
\[\left[\mathcal{P}_{\rm bound},\Box_{\tau}\right]\psi_{\tau} = \left[\mathcal{P}_{\rm bound},\eta_{\tau}-1\right]\left(\Box_g-\Box_{g_K}\right).\]
This implies that
\[\int_{\mathcal{M}_K}\tilde{\chi}^2_{\tau-10, 10}(t)\left|\left[\mathcal{P}_{\rm bound},\Box_{\tau}\right]\psi_{\tau}\right|^2  = \int_{\mathcal{M}_K}\tilde{\chi}^2_{\tau-10, 10}\left|\mathcal{P}_{\rm bound}\left(\left(\eta_{\tau}-1\right)\left(\Box_g-\Box_{g_K}\right)\psi_{\tau}\right)\right|^2,\]
where we have used that $\eta_{\tau} - 1$ vanishes on the support of $\tilde{\chi}^2_{\tau-10, 10}(s)$. Using this fact yet again, we may apply Theorem~\ref{offthesupportitsgood} to successfully estimate this term:
\[
\int_{\mathcal{M}_K}\tilde{\chi}^2_{\tau-10, 10}(s)\left|\left[\mathcal{P}_{\rm bound},\Box_{\tau}\right]\psi_{\tau}\right|^2 \lesssim \epsilon^2\sup_{t^*} \int_{t^*-1/2}^{t^*+1/2} E_{1,\{r_{\rm low} \leq r \leq R_{\rm high}\}}\left[\psi_{\tau}\right]\left(s\right)\, ds
\]

Lastly, we note that we may apply Corollary~\ref{corextend} and Theorem~\ref{offthesupportitsgood} to estimate the term 
\[\int_{\mathcal{M}_K}\tilde{\chi}^2_{\tau-10, 10}(s)\left|\mathcal{P}_{\rm bound}\tilde{F}\right|^2\lesssim 
A_{\rm high}^{-2} \sup_{t^*} \int_{t^*-1/2}^{t^*+1/2} \int_{\Sigma_s}|\tilde F|^2 ds \lesssim A_{\rm high}^{-2}E_0\left[\psi\right](0).\]

Thus, applying Corollary~\ref{timeaveragethepart3est} finishes the proof.

\end{proof}

In the next lemma we will use the elliptic estimate of Lemma~\ref{ellipt} to obtain an estimate for higher regularity norms of $\Psi_{\tau}$ in terms of lower regularity norms.
\begin{lemma}\label{spacetimeellip}Let $\psi$ be as in the statement of Theorem~\ref{theoc1pert}. Choose $\tau \gg 1$ and let $\psi_{\tau}$ be defined by Definition~\ref{defpsitau}. Set $\Psi_{\tau} \doteq \mathcal{P}_{\rm bound}\psi_{\tau}$. Let $k \geq 1$ be an integer,  and let $I \subset (r_++\delta_0,\infty)$ be a closed, bounded interval which contains $[r_{\rm low},R_{\rm high}]$. We denote
\[
\mathscr{W}_{\tau} = \int_{\mathcal{M}_K} \tilde{\chi}^2_{\tau-10, 10}(s)\left|\Psi_{\tau}\right|^2r^{-4} + A_{\rm high}^{-2} \sup_{t^* \in (8,\tau-8)} \int_{t^*-1/2}^{t^*+1/2} \|\Psi_{\tau}(s)\|^2_{L^2(\Sigma_s \cap \{r \in I\})} ds
\]

Then we have that
\begin{equation}\label{onetwothree}
    \begin{split}
       & \int_{-\infty}^{\infty}\tilde{\chi}^2_{\tau-10, 10}(s)E_{k,\{r_{\rm low} \leq r \leq R_{\rm high}\}}\left[\Psi_{\tau}\right](s)\, ds \lesssim_k
        \\ &\qquad \left(\epsilon + A_{\rm high}^{-2}\right) \sup_{t^*} \int_{t^*-1/2}^{t^*+1/2}E_{k}\left[\psi_{\tau}\right](s)\, ds + A_{\rm high}^{-2}E_{k-1}\left[\psi\right](0) +A_{\rm high}^{2k+2}\mathscr{W}_{\tau}.
    \end{split}
\end{equation}
 Analogously, we have
\begin{equation}\label{onetwothreetwo}\begin{split}
&\sup_{t \in (10,\tau-10)}\int_{-\infty}^{\infty}\chi_t^2(s)E_{k,r\geq r_++\delta_0}\left[\Psi_{\tau}\right]\left(s\right)\, ds  \lesssim_{k,\delta}  \\ &\qquad \left(\epsilon + A_{\rm high}^{-2}\right) \sup_{t^*} \int_{t^*-1/2}^{t^*+1/2}E_{ k}\left[\psi_{\tau}\right](s)\, ds + A_{\rm high}^{-2}E_{k-1}\left[\psi\right](0) +A_{\rm high}^{2k+2}\mathscr{W}_{\tau}.
\end{split}\end{equation}
\end{lemma}

\begin{proof} In view of Lemma~\ref{ellipt} we have
\begin{equation}
\begin{split}
&\int_{-\infty}^{\infty}\tilde{\chi}^2_{\tau-10,10}(s)E_{k,\{r_{\rm low} \leq r \leq R_{\rm high}\}}\left[\Psi_{\tau}\right](s)\, ds \lesssim_k 
\\ & \int_{\mathcal{M}_K\cap \{ r\in I\}}\tilde{\chi}^2_{\tau-10,10}\left|\Box_{\tau}\Psi_{\tau}\right|^2+\int_{-\infty}^{\infty}\tilde{\chi}^2_{\tau-10,10}(s)E_{k-2, I}\left[\Box_{\tau}\Psi_{\tau}\right](s)\, ds+
\\ & \sum_{i+j = 0}^{k+1}\int_{\mathcal{M}_K\cap \{r \in I\}}\tilde{\chi}^2_{\tau-10,10}(s)\left|T^i\Phi^j\Psi_{\tau}\right|^2.
\end{split}
\end{equation}

Let $\tilde{\chi}_1$ be a cutoff that equals $1$ on the support of $\tilde{\chi}^2_{\tau-10,10}$ with slightly larger support. Using Lemma~\ref{commisminus1}, applied to the  $0$ order operator $A_{\rm high}^{-k-1}T^i\Phi^j\tilde {\mathcal{P}}_{\rm bound}$ for a suitable projecting pair $\tilde p_{\rm bound}$ of $p_{\rm bound}$ and Convention~\ref{wehavem}, we find that 
\begin{equation}
\begin{split}
&\sum_{i+j = 0}^{k+1}\int_{\mathcal{M}_K\cap \{r \in I\}}\tilde{\chi}^2_{\tau-10,10}(s)\left|T^i\Phi^j\Psi_{\tau}\right|^2 \lesssim  A_{\rm high}^{2k+2}\int_{\mathcal{M}_K\cap \{r \in I\}}\tilde{\chi}^2_{\tau-10,10}(s)\left|\Psi_{\tau}\right|^2 
\\ &+ A_{\rm high}^{-2}\sup_{t^*} \int_{t^*-1/2}^{t^*+1/2}E_{0}\left[\Psi_{\tau}\right] 
 + A_{\rm high}^{2k}\int_{\mathcal{M}_K\cap \{r \in I\}}\tilde{\chi}_1^2 \left|\Psi_{\tau}\right|^2 
\lesssim A_{\rm high}^{2k+2}\mathscr{W}_{\tau} + A_{\rm high}^{-2}\sup_{t^*} \int_{t^*-1/2}^{t^*+1/2}E_{0}\left[\Psi_{\tau}\right].
\end{split}
\end{equation}

Arguing as in the proof of Lemma~\ref{applythatstufftotheboun}, we have that 
\begin{equation}\begin{split}
    & \int_{-\infty}^{\infty}\tilde{\chi}^2_{\tau-10,10}(t)E_{k-2, I}\left[\Box_{\tau}\Psi_{\tau}\right](s)\, ds 
    \lesssim \int_{-\infty}^{\infty}\tilde{\chi}^2_{\tau-10,10}(s)E_{k-2, I} [\mathcal{P}_{\rm bound}\tilde{F}]
    \\ &  + \int_{-\infty}^{\infty}\tilde{\chi}^2_{\tau-10,10}(s)E_{k-2, I} \left[[\Box_\tau, \mathcal{P}_{\rm bound}]\psi_\tau\right]
    \lesssim A_{\rm high}^{-2}E_{k-1}\left[\psi\right](0) \\ &   + \epsilon \sup_{t^*} \int_{t^*-1/2}^{t^*+1/2}E_{k,{r_{\rm low} \leq r \leq R_{\rm high}}}\left[\psi_{\tau}\right](s)\, ds.
\end{split}
\end{equation}

Combining these three estimates along with Lemma~\ref{applythatstufftotheboun} completes the proof of~\eqref{onetwothree}. The proof of~\eqref{onetwothreetwo} is analogous.
\end{proof}

In the next lemma we will use finite-time-energy estimates to obtain a useful estimate for $\mathscr{W}_{\tau}$.
\begin{lemma}\label{ljjkij2oi22lkj34jk} Let $\mathscr{W}_{\tau}$ be as in the statement of Lemma~\ref{spacetimeellip}. Then
\begin{equation}\label{jjiojio32oi12jnd}\begin{split}
    &\mathscr{W}_{\tau}  \lesssim \int_{\mathcal{M}_K}\tilde{\chi}^2_{\tau-10,10}(s)\left| \Psi_{\tau}\right|^2 r^{-4} + \sup_{t \in (10,\tau-10)}\int_{-\infty}^{\infty} \chi_{t}^2(s)E_0\left[\Psi_{\tau}\right](s)\, ds
    \\ &\qquad +\epsilon^2\sup_{t^*} \int_{t^*-1/2}^{t^*+1/2} E_{1,\{r_{\rm low} \leq r \leq R_{\rm high}\}}\left[\psi_{\tau}\right]\left(s\right)\, ds  + A_{\rm high}^{-2}E_0\left[\psi\right](0). 
\end{split}\end{equation}
\end{lemma}
\begin{proof}Using the fundamental theorem of calculus in the $t^*$-direction, we have
\begin{equation}\label{3ioj2oin3}\begin{split}
&\sup_{t^* \in (8,\tau-8)}\int_{t^*-1/2}^{t^*+1/2}\| \Psi_{\tau}(s) \|^2_{L^2\left(\Sigma_s\cap \{r\in I\}\right)} \lesssim \\ &\qquad \int_{\mathcal{M}_K}\tilde{\chi}^2_{\tau-10,10}(s)\left| \Psi_{\tau}\right|^2 r^{-4} + \sup_{t^* \in (\tau-10,\tau-7)}\| T \Psi_{\tau}(t) \|^2_{L^2\left(\Sigma_{t^*}\cap \{r \in I\}\right)}
\lesssim \\ &\qquad \int_{\mathcal{M}_K}\tilde{\chi}^2_{\tau-10,10}(s)\left| \Psi_{\tau}\right|^2 r^{-4} + \sup_{t^* \in (\tau-10,\tau-7)}E_0\left[\Psi_{\tau}\right](t^*).
\end{split}\end{equation}

We can find $\tilde{t} \in [\tau-11,\tau-10]$ so that
\[E_0\left[\Psi_{\tau}\right]\left(\tilde{t}\right) \lesssim \sup_{t^* \in (10,\tau-10)}\int_{-\infty}^{\infty} \chi_{t^*}^2(s)E_0\left[\Psi_{\tau}\right](s)\, ds.\]

In view of finite-in-time energy estimates, we have
\begin{equation}\begin{split}
   & \sup_{t^* \in (\tau-10,\tau-7)}E_0\left[\Psi_{\tau}\right](t^*) \lesssim E_0\left[\Psi_{\tau}\right]\left(\tilde{t}\right) + \int_{\mathcal{R}\left(\tau-11,\tau-7\right)}\left|\Box_{\tau}\left(\Psi_{\tau}\right)\right|^2
    \\ &\qquad \lesssim \sup_{t^* \in (10,\tau-10)}\int_{-\infty}^{\infty} \chi_{t^*}^2(s)E_0\left[\Psi_{\tau}\right](s)\, ds + \int_{\mathcal{R}\left(\tau-11,\tau-7\right)}\left|\Box_{\tau}\left(\Psi_{\tau}\right)\right|^2.
\end{split}\end{equation}

We have
\[\Box_{\tau}\Psi_{\tau} = -\left[\mathcal{P}_{\rm bound},\Box_{\tau}\right]\psi_{\tau} + \mathcal{P}_{\rm bound}\tilde{F}.\]
Thus, arguing just as in the proof of Lemma~\ref{applythatstufftotheboun} allows us to obtain
\begin{equation}\begin{split}
&\int_{\mathcal{R}\left(\tau-11,\tau-7\right)}\left|\Box_{\tau}\left(\Psi_{\tau}\right)\right|^2 \lesssim \epsilon^2\sup_{t^*} \int_{t^*-1/2}^{t^*+1/2} E_{1,\{r_{\rm low} \leq r \leq R_{\rm high}\}}\left[\psi_{\tau}\right]\left(s\right)\, ds
\\ &\qquad + A_{\rm high}^{-2}E_0\left[\psi\right](0),
\end{split}\end{equation}
which finishes the proof.
\end{proof}

 In this next corollary, we give our final estimate $\mathscr{W}_{\tau}$.
\begin{corollary}\label{3oij1ioj4oi2}Let $\mathscr{W}_{\tau}$ be as in the statement of Lemma~\ref{spacetimeellip}. Then we have
\begin{equation}\begin{split}
&\mathscr{W}_{\tau} \lesssim A_{\rm high}^4\int_{-\infty}^{\infty}\chi^2_{10}(s)E_0\left[\Psi_{\tau}\right](s)\, ds 
 + A_{\rm high}^4 \epsilon^2 \sup_{t^*} \int_{t^*-1/2}^{t^*+1/2} E_{1,\{r_{\rm low} \leq r \leq R_{\rm high}\}}\left[\psi_{\tau}\right]\left(s\right)\, ds
 \\ &\qquad + A_{\rm high}^2E_0\left[\psi\right](0) +\left(\epsilon+A_{\rm high}^{-6}\right)\sup_{t^*} \int_{t^*-1/2}^{t^*+1/2}E_{1}\left[\psi_{\tau}\right](s)\, ds.
\end{split}
\end{equation}

\end{corollary}
\begin{proof}

 We start by applying the estimate for $\mathscr{W}_{\tau}$ from Lemma~\ref{ljjkij2oi22lkj34jk}, then, for any $\mu > 0$, use Lemma~\ref{applythatstufftotheboun} to estimate the first two terms on the right hand side of~\eqref{jjiojio32oi12jnd}, and finally use~\eqref{onetwothree} with $k = 1$ to obtain that 
\begin{equation}\begin{split}
&\mathscr{W}_{\tau} \lesssim \mu^{-1}\int_{-\infty}^{\infty}\chi^2_{10}(s)E_0\left[\Psi_{\tau}\right](s)\, ds 
\\ &\qquad + \mu^{-1} \epsilon^2 \sup_{t^*} \int_{t^*-1/2}^{t^*+1/2} E_{1,\{r_{\rm low} \leq r \leq R_{\rm high}\}}\left[\psi_{\tau}\right]\left(s\right)\, ds+ \mu^{-1}A_{\rm high}^{-2}E_0\left[\psi\right](0)
\\ &\qquad +\left(\mu+\epsilon^2\right)\left[\left(\epsilon  + A_{\rm high}^{-2}\right) \sup_{t^*} \int_{t^*-1/2}^{t^*+1/2}E_{1}\left[\psi_{\tau}\right](s)\, ds + A_{\rm high}^{-2}E_{0}\left[\psi\right](0) +A_{\rm high}^4\mathscr{W}_{\tau}\right]
 \\ &\qquad + \epsilon^2\sup_{t^*} \int_{t^*-1/2}^{t^*+1/2} E_{1,\{r_{\rm low} \leq r \leq R_{\rm high}\}}\left[\psi_{\tau}\right]\left(s\right)\, ds  + A_{\rm high}^{-2}E_0\left[\psi\right](0). 
\end{split}
\end{equation}
We then take $\mu = \mu_0 A_{\rm high}^{-4}$ for sufficiently small $\mu_0 > 0$ to finish the proof.

\end{proof}

Now we give the proof of Lemma~\ref{enerellipboundedfreq}.
\begin{proof}We simply combine Lemma~\ref{applythatstufftotheboun} and Lemma~\ref{spacetimeellip} and use Corollary~\ref{3oij1ioj4oi2} to control the $\mathscr{W}_{\tau}$ that shows up on the right hand side of the estimates of Lemma~\ref{spacetimeellip}. 
\end{proof}

\subsection{(Generalized) large superradiant frequency estimates}\label{gensuper}

In this section we will estimate $\mathcal{P}_{n,\rm high}\psi_{\tau}$.  In what follows, we will repeatedly use, often without explicit comment, that for any $f \in \mathbb{H}_K$,
\[\int_{\mathbb{S}^2}\left|\mathcal{P}_{n,\rm high}f\right|^2\, d\mathbb{S}^2 \lesssim A_{\rm high}^{-2}\int_{\mathbb{S}^2}\left|\Phi\mathcal{P}_{n,\rm high}f\right|^2\, d\mathbb{S}^2,\]
which holds in view of the fact that on the support of $\mathcal{P}_{n,\rm high}$ we have that $|m| \gtrsim A_{\rm high}$.

Our main result of the section will be the following.
\begin{lemma}\label{finalthingfromthatblackthingy}Let $\psi$ be as in the statement of Theorem~\ref{theoc1pert}. Choose $\tau \gg 1$ and let $\psi_{\tau}$ be defined by Definition~\ref{defpsitau}. Let $1\leq n \leq N_s$, and set $\Psi_{\tau} \doteq \mathcal{P}_{n, \rm high}\psi_{\tau}$. Let $k \geq 1$ be an integer, and  $\delta$ be a small constant. Then we have that 
\begin{equation}\begin{split}
&\sup_t\int_{-\infty}^{\infty}\chi_t(s) E_{k, r \geq   R_{\rm high}}\left[\Psi_{\tau}\right](s)\, ds + \int_{-\infty}^{\infty}E_{k,r \in [r_{\rm low},R_{\rm high}]}\left[\Psi_{\tau}\right](s)\, ds 
\\ &\qquad + \int_{\mathcal{M}_K \cap \{r \in [r_{\rm low},R_{\rm high}]\}}\left|\Psi_{\tau}\right|^2 \lesssim 
\delta^{-1}E_k\left[\psi\right](0) +  \left( \delta + A_{\rm high}^{-2}\right)\sup_{t^*} \int_{t^*-1/2}^{t^*+1/2}E_k\left[\psi_{\tau}\right](s)\, ds.
\end{split}
\end{equation}
\end{lemma}

We will heavily rely on the following estimate which will be a consequence of a mild modification of a result of~\cite{blackboxlargea}.

\begin{theorem}\label{fromtheblackboxlargea}Let $n \in [1,N_s]$. Then there exists a vector field $V_n$, $w_n$, a constant $C_n$, and an operator $\mathscr{A}_n : \mathbb{H}_K\to \mathbb{R}_{\geq 0}$ which satisfy the following:
\begin{enumerate}
    \item Let $f \in \mathbb{H}_K$. Then
    \begin{enumerate}
        \item $r \leq  r_{\rm low} \Rightarrow \tilde{K}^{V_n,w_n}\left[f\right] \geq 0$.
        \item $r \geq R_{\rm high} \Rightarrow \tilde{K}^{V_n,w_n}\left[f\right] \geq 0$.
        \item\label{wheretildedeltawasfrom}  Let $\tilde{\delta} > 0$ be suitable small. We have

        \begin{enumerate}\item
        \begin{equation}\begin{split}
        &\mathscr{A}_n\left[ \mathcal{P}_{n,\rm high}f\right]+\int_{\mathcal{M}_K\cap \left\{r \in  [(1-\tilde{\delta})r_{\rm low},(1+\tilde{\delta})R_{\rm high}]\right\}}\tilde{K}^{V_n,w_n}\left[\mathcal{P}_{n,\rm high}f\right]  \gtrsim \\ &\qquad \int_{-\infty}^{\infty}E_{0,r \in  [(1-\tilde{\delta})r_{\rm low},(1+\tilde{\delta})R_{\rm high}]}\left[\mathcal{P}_{n,\rm high}f\right](s)\, ds
        \end{split}\end{equation}
 \item
 \begin{equation}\begin{split}
        &\mathscr{A}_n\left[ \mathcal{P}_{n,\rm high}f\right]+\int_{\mathcal{M}_K\cap \left\{r \in  [(1-\tilde{\delta})r_{\rm low},(1+\tilde{\delta})R_{\rm high}]\right\}}\tilde{K}^{V_n,w_n}\left[\mathcal{P}_{n,\rm high}Tf\right]
        \gtrsim \\ &\qquad \int_{-\infty}^{\infty}E_{0,r \in  [(1-\tilde{\delta})r_{\rm low},(1+\tilde{\delta})R_{\rm high}]}\left[\mathcal{P}_{n,\rm high}Tf\right](s)\, ds,
        \end{split}\end{equation}
    \item
         \begin{equation}\begin{split}
        &\mathscr{A}_n\left[ \mathcal{P}_{n,\rm high}f\right]+\int_{\mathcal{M}_K\cap \left\{r \in  [(1-\tilde{\delta})r_{\rm low},(1+\tilde{\delta})R_{\rm high}]\right\}} 
    \tilde{K}^{V_n,w_n}\left[\mathcal{P}_{n,\rm high}\Phi f\right]
        \gtrsim \\ &\qquad \int_{-\infty}^{\infty} E_{0,r \in  [(1-\tilde{\delta})r_{\rm low},(1+\tilde{\delta})R_{\rm high}]}\left[\mathcal{P}_{n,\rm high}\Phi f\right](s)\, ds
        \end{split}\end{equation}
        \end{enumerate}
          \item For the region $r \geq  R_{\rm high}$, we have
        \[\int_{\mathcal{M}_K \cap \{r \geq  R_{\rm high}\}}\tilde{K}^{V_n,w_n}\left[\mathcal{P}_{n,\rm high}f\right] \gtrsim \sum_{\left|\textbf{k}\right| = 1}\int_{\mathcal{M}_K \cap \{r \geq  R_{\rm high}\}}\left|\tilde{\mathfrak{D}}^{\textbf{k}}\mathcal{P}_{n,\rm high}f\right|^2r^{-2}\]
    \end{enumerate}
    \item The vector field $V_n$ and function $w_n$ satisfy the following:
    \begin{enumerate}
        \item $V_n$ is timelike when  $r \leq r_{\rm low}$  and $r \geq R_{\rm high}/10$. 
        \item $V_n$ commutes with $T$ and $\Phi$.
        \item For  $r \geq R_{\rm high}/10$, we have that
        \[V_n = a_t T + a_r \partial_r,\]
        for smooth functions $a_t$ and $a_r$, depending only on $r$ and which satisfy
        \[\left|a_t\right| + \left|a_r\right| + r^2\left[\left|\partial_ra_t\right| + \left|\partial_ra_r\right|\right] \lesssim 1.\]
        \item We have that $|w_n| \lesssim r^{-2}$. 
        
    \end{enumerate}
   \item\label{theanpart} The operator $\mathscr{A}_n$ satisfies the following: Let $\tilde{\delta}$ be as in~\ref{wheretildedeltawasfrom}, then we have
        \begin{equation}\label{mathscranbound}\begin{split}
        \mathscr{A}_n\left[{\mathcal{P}_{n,\rm high}}f\right]
       &\lesssim_{\tilde{\delta}}   \int_{-\infty}^{\infty}E_{0,r \in  [\left(1+\tilde{\delta}\right)r_{\rm low},\left(1-\tilde{\delta}\right)R_{\rm high}]}\left[{ \mathcal{P}_{n,\rm high}}f\right](t)\, dt 
       \\ & +
       \int_{\mathcal{M}_K \cap \left\{r \in  [\left(1+\tilde{\delta}\right)r_{\rm low},\left(1-\tilde{\delta}\right)R_{\rm high}]\right\}}\left[\left|{ \mathcal{P}_{n,\rm high}}f\right|^2 + \left|\Box_{g_K}{\mathcal{P}_{n,\rm high}}f\right|^2\right].
        \end{split}\end{equation}

\end{enumerate}
\end{theorem}
\begin{proof}See Appendix~\ref{provethatkeythm}.
\end{proof}

In the next lemma we combine Theorem~\ref{fromtheblackboxlargea} with an application of the divergence identity to produce a spacetime estimate for $\mathcal{P}_{n,\rm high}$.
\begin{lemma}\label{dividwiththatblackboxcurrent}Let $\psi$ be as in the statement of Theorem~\ref{theoc1pert}. Choose $\tau \gg 1$ and let $\psi_{\tau}$ be defined by Definition~\ref{defpsitau}. Let $n \in [1,N_s]$, and set $\Psi_{\tau} \doteq \mathcal{P}_{n, \rm high}\psi_{\tau}$. Then we have that
\begin{equation}\label{afterthefloodT}\begin{split}
&\int_{\mathcal{M}_K}\tilde{K}^{V_n,w_n}\left[T\Psi_{\tau}\right] \lesssim \int_{\mathcal{M}_K}{\rm Re}\left(\left(V_nT\Psi_{\tau} + w_nT\Psi_{\tau}\right)\overline{\Box_{g_K}T\Psi_{\tau}}\right)
\\ &\qquad +A_{\rm high}^{-2}\sup_{t^*} \int_{t^*-1/2}^{t^*+1/2}E_1\left[\psi_{\tau}\right](s)\, ds + E_1\left[\psi\right](0).
\end{split}\end{equation}

Let $\check{\chi}(r)$ be a function which is identically $1$ for $r \leq  R_{\rm high}/5$ and identically $0$ for $r \geq  R_{\rm high}/2$. Then we have that 
\begin{equation}\label{afterthefloodphi}
\int_{\mathcal{M}_K}\tilde{K}^{V_n,w_n}\left[\check{\chi}\Phi\Psi_{\tau}\right] \lesssim \int_{\mathcal{M}_K}{\rm Re}\left(\left(V_n\check{\chi}\Phi\Psi_{\tau} + w_n\check{\chi}\Phi\Psi_{\tau}\right)\overline{\Box_{g_K}\check{\chi}\Phi\Psi_{\tau}}\right).
\end{equation}
\end{lemma}
\begin{proof}We start with the proof of~\eqref{afterthefloodT}. Since $\Psi_{\tau} \in \mathbb{H}_K$, Remark~\ref{endecay} implies that for any compact $K \subset [r_+,\infty)$ we have
\begin{equation}\label{vanishtofuture}
\lim_{t^*\to \infty}E_{1,K}\left[\Psi_{\tau}\right]\left(t^*\right) = 0.
\end{equation}

In view of~\eqref{vanishtofuture} and the fact that $V_n$ is causal for large $r$, the divergence identity yields that
\begin{equation}\label{afterthefloodTalmost}\begin{split}
&\int_{\mathcal{M}_K}\tilde{K}^{V_n,w_n}\left[T\Psi_{\tau}\right] \lesssim 
\\ &\qquad \int_{\mathcal{M}_K}{\rm Re}\left(\left(V_nT\Psi_{\tau} + w_nT\Psi_{\tau}\right)\overline{\Box_{g_K}T\Psi_{\tau}}\right) + \liminf_{t^*\to -\infty}E_1\left[\Psi_{\tau}\right](t^*).
\end{split}\end{equation}
To control the final term on the right hand side of, we use Lemma~\ref{commisminus1}, Theorem~\ref{extendtoeverywhere} and the fact that
\[\liminf_{t^*\to -\infty}E_1\left[\Psi_{\tau}\right](t^*) \lesssim \liminf_{t^*\to-\infty}\int_{t^*-1/2}^{t^*+1/2}E_1\left[\Psi_{\tau}\right](s)\, ds.\]

To prove~\eqref{afterthefloodphi} we argue similarly, except that in this case we may use that, in view of the fact that $\Psi_{\tau} \in \mathbb{H}_K$, Remark~\ref{endecay} directly yields
\[\liminf_{t^*\to -\infty}E_1\left[\check{\chi}\Psi_{\tau}\right](t^*) = 0.\]
\end{proof}

Now we will use Theorem~\ref{fromtheblackboxlargea} to estimate $\mathcal{P}_{n, \rm high}\psi_{\tau}$  on a bounded but large spatial region.

\begin{lemma}\label{almosttherewowsupersuperradiatn}Let $\psi$ be as in the statement of Theorem~\ref{theoc1pert}. Choose $\tau \gg 1$ and let $\psi_{\tau}$ be defined by Definition~\ref{defpsitau}. Let $n \in [1,N_s]$, and set $\Psi_{\tau} \doteq \mathcal{P}_{n, \rm high}\psi_{\tau}$. Let $k \geq 1$ be an integer  and pick $1\gg \delta \gg \epsilon$ be a small constant. Then we have that
\begin{equation}\begin{split}
&\int_{-\infty}^{\infty}E_{k,r \in [r_{\rm low},R_{\rm high}]}\left[\Psi_{\tau}\right](s)\, ds + \int_{\mathcal{M}_K \cap \{r \in [r_{\rm low},R_{\rm high}]\}}\left|\Psi_{\tau}\right|^2 \lesssim 
\\ &\qquad  \delta^{-1}E_k\left[\psi\right](0) +  \left({ \delta} + A_{\rm high}^{-2}\right)\sup_{t^*} \int_{t^*-1/2}^{t^*+1/2}E_k\left[\psi_{\tau}\right](s)\, ds.
\end{split}
\end{equation}
\end{lemma}
\begin{proof}We first consider the case when $k = 1$.

We have
\begin{equation}\label{thefirstcommutatorthing}
\Box_{g_K}\Psi_{\tau} = \left(\Box_{g_K}-\Box_{\tau}\right)\Psi_{\tau} + \left[\mathcal{P}_{n,\rm high},\eta_{\tau}\right]\left(\Box_{g_K}-\Box_g\right)\psi_{\tau} + \mathcal{P}_{n,\rm high}\tilde{F}.
\end{equation}
Commuting with $T$ leads to 
\begin{equation}\label{psitauwithT}\begin{split}
&\Box_{g_K}T\Psi_{\tau} = \left(\Box_{g_K}-\Box_{\tau}\right)T\Psi_{\tau} + \left[\mathcal{P}_{n,\rm high},\eta_{\tau}\right]T\left(\Box_{g_K}-\Box_g\right)\psi_{\tau} 
\\ & \qquad + \mathcal{P}_{n,\rm high}T\tilde{F}+\left[\mathcal{P}_{n, \rm high},T\eta_{\tau}\right]\left(\Box_{g_K}-\Box_g\right)\psi_{\tau} - \left[T,\Box_{\tau}\right]\Psi_{\tau}.
\end{split}
\end{equation}
Let $\check{\chi}(r)$ be as in Lemma~\ref{dividwiththatblackboxcurrent}. We will also commute with $\check{\chi}(r) \Phi$. The relevant equation is 
\begin{equation}\label{psitauwithphi}\begin{split}
&\Box_{g_K}\left(\check{\chi}\Phi\Psi_{\tau}\right) = \left(\Box_{g_K}-\Box_{\tau}\right)\left(\check{\chi}\Phi\Psi_{\tau}\right) + \left[\mathcal{P}_{n,\rm high},\eta_{\tau}\right]\Phi\left(\Box_{g_K}-\Box_g\right)\left(\check{\chi}\psi_{\tau}\right) \\ &\qquad + \mathcal{P}_{n,\rm high}\left(\check{\chi}\Phi\tilde{F}\right)
- \left[\check{\chi},\Box_{\tau}\right]\Phi\Psi_{\tau} + \left[\mathcal{P}_{n,\rm high},\eta_{\tau}\right]\Phi\left[\Box_{g_K}-\Box_g,\check{\chi}\right]\psi_{\tau}.
\end{split}\end{equation}

We will need estimates for each of the three terms
\begin{equation}\label{thisisfirst}
I \doteq \int_{\mathcal{M}_K}\left|\Box_{g_K}\Psi_{\tau}\right|^2,
\end{equation}
\begin{equation}\label{thisissecond}
II \doteq \int_{\mathcal{M}_K}{\rm Re}\left(\left(V_nT\Psi_{\tau} + w_nT\Psi_{\tau}\right)\overline{\Box_{g_K}\left(T\Psi_{\tau}\right)}\right),
\end{equation}
\begin{equation}\label{thisisthird}
III \doteq \int_{\mathcal{M}_K}{\rm Re}\left(\left(V_n\check{\chi}\Phi\Psi_{\tau} + w_n\check{\chi}\Phi\Psi_{\tau}\right)\overline{\Box_{g_K}\left(\check{\chi}\Phi\Psi_{\tau}\right)}\right).
\end{equation}

 It will be important to keep in mind that when $r \geq R_{\rm high}$ we have that $\Box_{g_K}\Psi_{\tau}$ vanishes, and when $r \leq r_{\rm low}$ we have that $\Box_{g_K}\Psi_{\tau} = \mathcal{P}_{n,{\rm high}}\tilde{F}$. Moreover, $r \leq r_{\rm low}$ implies that the $t^*$ support $\tilde{F}$ is compact. Analogous statements holds for $\Box_{g_K}T\Psi_{\tau}$ and $\Box_{g_K}\left(\check{\chi}\Phi\Psi_{\tau}\right)$.

It now follows from Lemma~\ref{usefulinthatonepart},  Theorem~\ref{pseudoloc} and Theorem~\ref{extendtoeverywhere} that 
\begin{equation}\label{estforIyay}\begin{split}
&\left|I\right| \lesssim  \epsilon \int_{-\infty}^{\infty} E_{1,r\in [r_{\rm low},R_{\rm high}]}\left[\Psi_{\tau}\right](s)\, ds 
 +  \epsilon \sup_{t^*} \int_{t^*-1/2}^{t^*+1/2}E_1\left[\psi_{\tau}\right](s)\, ds + E_0\left[\psi\right](0),
\end{split}\end{equation}.

The bounds for $II$ and $III$ follow in a similar manner, with two notable exceptions. We use Lemma~\ref{thatscoolintbyparts} to estimate the first term on the RHS of \eqref{psitauwithT} and \eqref{psitauwithphi}. Moreover, we use Corollary~\ref{blahblah1231234} to estimate the third term on the RHS of \eqref{psitauwithT} and \eqref{psitauwithphi} for $r\leq r_{\rm low}$. We obtain for a fixed $\epsilon \ll\delta\ll 1$:

\begin{equation}\label{estforIIyay}\begin{split}
&\left|II\right| \lesssim  \epsilon \int_{-\infty}^{\infty} E_{1,r\in [r_{\rm low},R_{\rm high}]}\left[\Psi_{\tau}\right](s)\, ds 
+  \delta \sup_{t^*} \int_{t^*-1/2}^{t^*+1/2}E_1\left[\psi_{\tau}\right](s)\, ds+  \delta^{-1} E_1\left[\psi\right](0),
\end{split}\end{equation}
\begin{equation}\label{estforIIIyay}\begin{split}
&\left|III\right| \lesssim  \epsilon \int_{-\infty}^{\infty} E_{1,r\in [r_{\rm low},R_{\rm high}]}\left[\Psi_{\tau}\right](s)\, ds + \delta \sup_{t^*} \int_{t^*-1/2}^{t^*+1/2}E_1\left[\psi_{\tau}\right](s)\, ds \\ &\qquad + \int_{-\infty}^{\infty}E_{1,r \in [R_{\rm high}/5,R_{\rm high}/2]}\left[\Psi_{\tau}\right](s)\, ds+ \delta^{-1}E_1\left[\psi\right](0).
\end{split}\end{equation}

We may now apply Lemma~\ref{ellipt} combined with \eqref{thefirstcommutatorthing}, Theorem~\ref{pseudoloc} and Lemma~\ref{usefulinthatonepart} to establish
\begin{equation}\label{ellipfarr}\begin{split}
&\int_{-\infty}^{\infty}E_{1, r \in {[R_{\rm high}/5,R_{\rm high}/2]}}\left[\Psi_{\tau}\right](s)\, ds \lesssim \sum_{i=0}^2\int_{\mathcal{M}_K\cap \{r \in [r_{\rm low},R_{\rm high}]\}}\left[\left|T^i\Psi_{\tau}\right|^2 +|\Box_{g_K}\Psi_\tau|^2\right]
\\ &\lesssim \sum_{i=0}^2\int_{\mathcal{M}_K\cap \{r \in [r_{\rm low},R_{\rm high}]\}}\left|T^i\Psi_{\tau}\right|^2 + \epsilon \sup_{t^*} \int_{t^*-1/2}^{t^*+1/2}E_1\left[\psi_{\tau}\right](s)\, ds+ E_0\left[\psi\right](0),
\end{split}\end{equation}

\begin{equation}\label{ellipnotjustfarr}\begin{split}
&\int_{-\infty}^{\infty}E_{1,r \in [r_{\rm low},R_{\rm high}]}\left[\Psi_{\tau}\right](t)\, dt
 \lesssim \sum_{i+j=0}^2\int_{\mathcal{M}_K\cap \{r \in [\left(1-\tilde{\delta}\right)r_{\rm low},\left(1+\tilde{\delta}\right)R_{\rm high}]\}}\left[\left|T^i\Phi^j\Psi_{\tau}\right|^2 +|\Box_{g_K}\Psi_\tau|^2\right]
\\ &\lesssim \sum_{i+j=0}^2\int_{\mathcal{M}_K\cap \{r \in [(1-\tilde{\delta})r_{\rm low},(1+\tilde{\delta})R_{\rm high}]\}}\left|T^i\Phi^j\Psi_{\tau}\right|^2
 + \epsilon \sup_t \int_{t-1/2}^{t+1/2}E_1\left[\psi_{\tau}\right](s)\, ds+ E_0\left[\psi\right](0).
\end{split}\end{equation}

Now we are ready to complete the proof. Applying Theorem~\ref{fromtheblackboxlargea} to $f=\psi_{\tau}$ yields   
\[
\int_{-\infty}^{\infty}E_{0,r\in{[(1-\tilde{\delta})r_{\rm low},(1+\tilde{\delta})R_{\rm high}]}}\left[T\Psi_{\tau}\right](s)\, ds \lesssim 
\mathscr{A}_n\left[\Psi_\tau\right]+\int_{\mathcal{M}_K\cap \left\{r \in  [(1-\tilde{\delta})r_{\rm low},(1+\tilde{\delta})R_{\rm high}]\right\}}\tilde{K}^{V_n,w_n}\left[T\Psi_\tau\right]
\]
We now bound by \eqref{mathscranbound}
\[\begin{split}
& \mathscr{A}_n\left[\Psi_\tau\right] \lesssim \int_{-\infty}^{\infty}E_{0,r \in  [\left(1+\tilde{\delta}\right)r_{\rm low},\left(1-\tilde{\delta}\right)R_{\rm high}]}\left[\Psi_\tau\right](s)\, ds 
\\& + \int_{\mathcal{M}_K \cap \left\{r \in  [\left(1+\tilde{\delta}\right)r_{\rm low},\left(1-\tilde{\delta}\right)R_{\rm high}]\right\}}\left[\left|\Psi_\tau\right|^2 + \left|\Box_{g_K}\Psi_\tau\right|^2\right]
\end{split}\]

Moreover, Lemma~\ref{dividwiththatblackboxcurrent} yields
\[
\int_{\mathcal{M}_K\cap \left\{r \in  [(1-\tilde{\delta})r_{\rm low},(1+\tilde{\delta})R_{\rm high}]\right\}}\tilde{K}^{V_n,w_n}\left[T\Psi_\tau\right] \lesssim \left|II\right| + A_{\rm high}^{-2}\sup_{t^*} \int_{t^*-1/2}^{t^*+1/2}E_1\left[\psi_{\tau}\right](s)\, ds + E_1\left[\psi\right](0)
\]

We thus obtain by \eqref{estforIyay} and \eqref{estforIIyay} 
\begin{equation}\label{okherewego1}\begin{split}
& \int_{-\infty}^{\infty}E_{0,r\in{[(1-\tilde{\delta})r_{\rm low},(1+\tilde{\delta})R_{\rm high}]}}\left[T\Psi_{\tau}\right](s)\, ds \lesssim \epsilon \int_{-\infty}^{\infty} E_{1,r\in [r_{\rm low},R_{\rm high}]}\left[\Psi_{\tau}\right](t)\, dt 
 \\ &\qquad +  \left( \delta + A_{\rm high}^{-2}\right)\sup_{t^*} \int_{t^*-1/2}^{t^*+1/2}E_1\left[\psi_{\tau}\right](s)\, ds+  \delta^{-1}E_1\left[\psi\right](0).
\end{split}
\end{equation}
Now we add in a small multiple of the estimate~\eqref{ellipfarr} to \eqref{okherewego1}. We obtain
\begin{equation}\label{okherewego2}\begin{split}
&\int_{-\infty}^{\infty}E_{0,r\in[(1-\tilde{\delta})r_{\rm low},(1+\tilde{\delta})R_{\rm high}]}\left[T\Psi_{\tau}\right](s)\, ds 
\\ &\qquad + \int_{-\infty}^{\infty}E_{1,r \in [R_{\rm high}/5,R_{\rm high}/2]}\left[\Psi_{\tau}\right](s)\, ds \lesssim  \delta^{-1} E_1\left[\psi\right](0)
\\ &\qquad + \left(\delta + A_{\rm high}^{-2}\right)\int_{-\infty}^{\infty} E_{1,r\in [r_{\rm low},R_{\rm high}]}\left[\Psi_{\tau}\right](s)\, ds 
 + \left(\delta + A_{\rm high}^{-2}\right)\sup_{t^*} \int_{t^*-1/2}^{t^*+1/2}E_1\left[\psi_{\tau}\right](s)\, ds.
\end{split}
\end{equation}
Applying again Theorem~\ref{fromtheblackboxlargea}, Lemma~\ref{dividwiththatblackboxcurrent} and using~\eqref{estforIyay} and~\eqref{estforIIIyay} leads to   
\begin{equation}\label{okherewego3}\begin{split}
&\int_{-\infty}^{\infty}E_{0,r\in[(1-\tilde{\delta})r_{\rm low},(1+\tilde{\delta})R_{\rm high}]}\left[\Phi\Psi_{\tau}\right](t)\, dt \lesssim  \epsilon \int_{-\infty}^{\infty} E_{1,r\in [r_{\rm low},R_{\rm high}]}\left[\Psi_{\tau}\right](t)\, dt 
 \\ &\qquad + \int_{-\infty}^{\infty}E_{1,r \in [R_{\rm high}/5,R_{\rm high}/2]}\left[\Psi_{\tau}\right](t)\, dt
 \\ &\qquad +\left( \delta + A_{\rm high}^{-2}\right)\sup_{t^*} \int_{t^*-1/2}^{t^*+1/2}E_1\left[\psi_{\tau}\right](s)\, ds+ \delta^{-1}E_1\left[\psi\right](0).
\end{split}
\end{equation}
Now we add a small multiple of~\eqref{okherewego3} to~\eqref{okherewego2} to obtain   
\begin{equation}\label{okherewego4}\begin{split}
&\int_{-\infty}^{\infty}\left(E_{0,r\in[(1-\tilde{\delta})r_{\rm low},(1+\tilde{\delta})R_{\rm high}]}\left[\Phi\Psi_{\tau}\right]+E_{0,r\in[(1-\tilde{\delta})r_{\rm low},(1+\tilde{\delta})R_{\rm high}]}\left[T\Psi_{\tau}\right](t)\right)\, dt
\\ &\qquad + \int_{-\infty}^{\infty}E_{1,r \in [R_{\rm high}/5,R_{\rm high}/2]}\left[\Psi_{\tau}\right](t)\, dt \lesssim  
\\ &\qquad \delta^{-1} E_1\left[\psi\right](0) + \left( \delta+A_{\rm high}^{-2}\right) \sup_{t^*} \int_{t^*-1/2}^{t^*+1/2}E_1\left[\psi_{\tau}\right](s)\, ds
\\ &\qquad + \left(\epsilon  
 +A_{\rm high}^{-2}\right)\int_{-\infty}^{\infty} E_{1,r\in [r_{\rm low},R_{\rm high}]}\left[\Psi_{\tau}\right](s)\, ds.
\end{split}
\end{equation}
Finally, the proof in the case when $k = 1$ is concluded by combining this with~\eqref{ellipnotjustfarr}.

For $k > 1$, we note the following formulas, which hold for any non-negative integers $i$, $j$, and $l$:
\begin{equation}\label{psitauwithThigherk}\begin{split}
&\Box_{g_K}T^{i+1}\Psi_{\tau} = T^{i+1}\Big(\left(\Box_{g_K}-\Box_{\tau}\right)T\Psi_{\tau} + \left[\mathcal{P}_{n,\rm high},\eta_{\tau}\right]T\left(\Box_{g_K}-\Box_g\right)\psi_{\tau} 
\\ & \qquad + \mathcal{P}_{n,\rm high}T\tilde{F}+\left[\mathcal{P}_{n, \rm high},T\eta_{\tau}\right]\left(\Box_{g_K}-\Box_g\right)\psi_{\tau} - \left[T,\Box_{\tau}\right]\Psi_{\tau}\Big),
\end{split}
\end{equation}
\begin{equation}\label{psitauwithphihigherk}\begin{split}
&\Box_{g_K}\left(T^j\Phi^l\check{\chi}\Phi\Psi_{\tau}\right) = T^j\Phi^l\Big(\left(\Box_{g_K}-\Box_{\tau}\right)\left(\check{\chi}\Phi\Psi_{\tau}\right) + \left[\mathcal{P}_{n,\rm high},\eta_{\tau}\right]\Phi\left(\Box_{g_K}-\Box_g\right)\left(\check{\chi}\psi_{\tau}\right) \\ &\qquad + \mathcal{P}_{n,\rm high}\left(\check{\chi}\Phi\tilde{F}\right)
- \left[\check{\chi},\Box_{\tau}\right]\Phi\Psi_{\tau} + \left[\mathcal{P}_{n,\rm high},\eta_{\tau}\right]\Phi\left[\Box_{g_K}-\Box_g,\check{\chi}\right]\psi_{\tau}\Big).
\end{split}\end{equation}
We conclude the proof by simply repeating the estimates from above \emph{mutatis mutandis}.
\end{proof}

We close the section by adding in a suitable energy estimate in order to control a time-averaged higher order energy flux for $r\geq R_{\rm high}$  which will complete the proof of Lemma~\ref{finalthingfromthatblackthingy}.
\begin{proof}[Proof of Lemma~\ref{finalthingfromthatblackthingy}]

Let $\tilde{\chi}(r)$ be a smooth cut-off function which is identically $0$ for $r \leq  R_{\rm high}-M$ and identically $1$ for $r \geq  R_{\rm high}$. It is straightforward to see that, using Lemma~\ref{almosttherewowsupersuperradiatn}, we have, for every integer $i \in [0,k]$:
\[\Box_{\tau}\left(\tilde{\chi}T^i\Psi_{\tau}\right) = \mathcal{F}_i,\]
where
\[\int_{\mathcal{M}_K}\left|\mathcal{F}_i\right|^2 \lesssim \delta^{-1}E_k\left[\psi\right](0) +  \left(\delta + A_{\rm high}^{-2}\right)\sup_{t^*} \int_{t^*-1/2}^{t^*+1/2}E_k\left[\psi_{\tau}\right](s)\, ds,\]
and ${\rm supp}\left(\mathcal{F}_i\right) \subset [R_{\rm high}-M,R_{\rm high}]$.

In particular, carrying out an energy estimate with the $T$-vector field (which is timelike in the region $r \geq R_{\rm high} - M$) we obtain
\begin{equation}\begin{split}
\sup_{t^*} E_{0,r\geq R_{\rm high}}\left[T^i\Psi_{\tau}\right](t^*)&\lesssim \sup_{t^*} E_0\left[T^i\tilde{\chi}\Psi_{\tau}\right](t^*)
\\ &\lesssim \int_{\mathcal{M}_K}\left|\mathcal{F}_i\right|\left|T^{i+1}\tilde{\chi}\Psi_{\tau}\right| + \liminf_{t^*\to-\infty}E_0\left[T^i\tilde{\chi}\Psi_{\tau}\right](t^*)
\\ &\lesssim  \delta^{-1}E_k\left[\psi\right](0) +  \left(\delta + A_{\rm high}^{-2}\right)\sup_{t^*} \int_{t^*-1/2}^{t^*+1/2}E_k\left[\psi_{\tau}\right](s)\, ds.
\end{split}\end{equation}
In the passage to the final inequality we use Cauchy Schwarz and Lemma~\ref{almosttherewowsupersuperradiatn} to control the first term, and argue as in the proof of Lemma~\ref{dividwiththatblackboxcurrent} to control the term with the $\liminf$.

An application of Lemma~\ref{ellipt} now yields
\begin{equation}\label{far}
\sup_t\int_{-\infty}^{\infty}\chi_t(s) E_{k, r \geq   R_{\rm high}}\left[\Psi_{\tau}\right](s)\, ds \lesssim \delta^{-1}E_k\left[\psi\right](0) +  \left(\delta + A_{\rm high}^{-2}\right)\sup_{t^*} \int_{t^*-1/2}^{t^*+1/2}E_k\left[\psi_{\tau}\right](s)\, ds.
\end{equation}
The proof of Lemma~\ref{finalthingfromthatblackthingy} is then concluded by  Lemma~\ref{almosttherewowsupersuperradiatn} and \eqref{far}.
   
\end{proof}
\subsection{The large non-superradiant frequencies}
In this section we will use Theorem~\ref{actuallywehatweneedformain} to estimate $\mathcal{P}_{HT}\psi$.
\begin{lemma}\label{finalthingfromnonsupersuper}Let $\psi$ be as in the statement of Theorem~\ref{theoc1pert}. Choose $\tau \gg 1$ and let $\psi_{\tau}$ be defined by Definition~\ref{defpsitau}. Set $\Psi_{\tau} \doteq \mathcal{P}_{HT}\psi_{\tau}$. Let $k \geq 1$. Then there exists a constant $C_k$ so that we have that
\begin{equation}\begin{split}
&\sup_{t \in (10,\tau-10)}\int_{-\infty}^{\infty}\chi_{t}(s)E_k\left[\Psi_{\tau}\right](s)\, ds \lesssim A_{\rm high}^{-1}\sup_{t}\int_{-\infty}^{\infty}\chi_{t}(s)E_k\left[\psi_\tau\right](s)\, ds+
\\ &\qquad   
\epsilon \sup_{t}\int_{-\infty}^{\infty} \chi_{t}(s)E_k\left[\psi_{\tau}\right](s)\, ds + E_k\left[\psi\right](0).
\end{split}
\end{equation}

\end{lemma}
\begin{proof}We apply Theorem~\ref{actuallywehatweneedformain}. We only need to estimate the right hand side of ~\eqref{thebetterengthing} where $\mathscr{F} = \left[\Box_{\tau},\mathcal{P}_{HT}\right]$. However, the arguments for bounding these commutator errors are essentially identical to many of the previous estimates we have done (cf~proofs of  Theorem~\ref{mainLinEstTheo}, Lemma~\ref{almosttherewowsupersuperradiatn} or~\ref{applythatstufftotheboun}). 

For example, following the steps from the proof of \eqref{Ftilde}, we estimate
\[\begin{split}
& \left|\int_{\mathcal{M}_K}\tilde{\chi}^2_{t}{\rm Re}\left( T^k \mathscr{F}\overline{K\left(T^k\Psi_{\tau}\right)}\right)\right| \\ & \lesssim \left(\int_I \|T^k\mathscr{F}(s)\|_{L^2(\Sigma_s)} + \|T^{k+1}\mathscr{F}(s)\|_{L^2(\Sigma_s)} ds\right) \left(\sup_{t^*>5}{\int_{t^*-1/2}^{t^*+1/2}}E_k[\Psi_{\tau}](s)ds\right)^{1/2}.
\end{split}
\]
Using the fact that, on the support of $\tilde{\chi}_t$, we have that $\Box_\tau=\Box_g$, and thus  $$\left[\Box_{\tau},\mathcal{P}_{HT}\right] = \mathcal{P}_{HT}\left(\Box_{\tau}-\Box_{g}\right),$$
we can easily estimate using \eqref{L1disj}:
\[\begin{split}
\int_I \|T^k\mathscr{F}(s)\|_{L^2(\Sigma_s)} + \|T^{k+1}\mathscr{F}(s)\|_{L^2(\Sigma_s)} ds \lesssim 
\sup_{t^*}{\int_{t^*-1/2}^{t^*+1/2}}\|(\Box_{\tau}-\Box_{g})\psi_\tau\|_{L^2(\Sigma_s)} ds \\ \lesssim \epsilon \sup_{t}\left(\int_{-\infty}^{\infty} \chi_{t}(s)E_1\left[\psi_{\tau}\right](s)\, ds\right)^{1/2}
\end{split}\]
We thus obtain
\[
\left|\int_{\mathcal{M}_K}\tilde{\chi}^2_{t}{\rm Re}\left( T^k \mathscr{F}\overline{K\left(T^k\Psi_{\tau}\right)}\right)\right| \lesssim \epsilon \sup_{t}\int_{-\infty}^{\infty} \chi_{t}(s)E_k\left[\psi_{\tau}\right](s)\, ds
\]

\end{proof}
\subsection{Concluding the proof}
We are ready to complete the proof of Theorem~\ref{theoc1pert}. In this first lemma we use the red-shift multiplier and commutator to control higher order non-degenerate energies of $\psi_{\tau}$.
\begin{lemma}\label{combineitallwoohoo}Let $\psi$ be as in the statement of Theorem~\ref{theoc1pert}. Choose $\tau \gg 1$ and let $\psi_{\tau}$ be defined by Definition~\ref{defpsitau}. Then we have a constant $C\left(A_{\rm high},k\right)$ so that 
\begin{equation}\label{withtheredaddedin}\begin{split}
&\sup_{t^* \in (10,\tau-10)}\int_{-\infty}^{\infty}\chi_{t^*}(s)E_k\left[\psi_{\tau}\right](s)\, ds \lesssim_k 
\\ &\qquad A_{\rm high}^{-1}\sup_{t^*} \int_{-\infty}^{\infty}\chi_{t^*}(s)E_k\left[\psi_{\tau}\right](s)\, ds +  A_{\rm high}^{2k-4}\sup_{t^*} \int_{-\infty}^{\infty}\chi_{t^*}(s)E_1\left[\psi_{\tau}\right](s)\, ds
\\ &\qquad  C\left(A_{\rm high},k\right)\left(E_k\left[\psi\right](0) + \int_{-2}^{12}E_k\left[\psi_{\tau}\right](s)\, ds\right).
\end{split}\end{equation}
\end{lemma}
\begin{proof}Already from Lemmas~\ref{enerellipboundedfreq},~\ref{finalthingfromthatblackthingy}, and~\ref{finalthingfromnonsupersuper}, and picking $\epsilon$ and $\delta$ small enough, we have that 
\begin{equation}\label{withtheredaddedin123123}\begin{split}
&\sup_{t^* \in (10,\tau-10)}\int_{-\infty}^{\infty}\chi_{t^*}(s)E_{k,r \geq r_{low}}\left[\psi_{\tau}\right](s)\, ds \lesssim_{k} 
\\ &\qquad A_{\rm high}^{-1}\sup_{t^*} \int_{-\infty}^{\infty}\chi_{t^*}(s)E_k\left[\psi_{\tau}\right](s)\, ds +  A_{\rm high}^{2k-4}\sup_{t^*} \int_{-\infty}^{\infty}\chi_{t^*}(s)E_1\left[\psi_{\tau}\right](s)\, ds
\\ &\qquad + C\left(A_{\rm high},k\right)\left(E_k\left[\psi\right](0) + \int_{-2}^{12}E_k\left[\psi_{\tau}\right](s)\, ds\right).
\end{split}\end{equation}

The full statement then follows from a straightforward adaption of the now standard arguments for upgrading a degenerate boundedness statement to non-degenerate boundedness statement (see Section 3.3 of~\cite{claylecturenotes}) to the setting where we average against the bump function $\chi_{t^*}$.
\end{proof}

Finally we will use finite-in-time energy estimates to complete the proof of Theorem~\ref{theoc1pert}.
\begin{proof}Let $\psi$ be as in the statement of Theorem~\ref{theoc1pert}. 

It is immediate from finite-in-time energy estimates imply that for every integer $k \geq 1$, there exists a constant $C_k$ (independent of $\tau$) so that $-\infty < s < t^* < \infty$ imply
\begin{equation}\label{finiteintime}
E_k\left[\psi_{\tau}\right](t^*) \leq e^{C_k\left(t^*-s\right)}E_k\left[\psi_{\tau}\right](s).
\end{equation}
In particular, for every $t^*$, we may define $\tilde{t}$ so that
\[E_k\left[\psi_{\tau}\right]\left(\tilde{t}\right) = {\rm inf}_{s \in [t^*-1/2, t^*]}E_k\left[\psi_{\tau}\right]\left(s\right),\]
and then obtain that
\begin{equation}\label{fromtheaveragetotherealthing}
E_k\left[\psi_{\tau}\right]\left(t^*\right) \lesssim_k E_k\left[\psi_{\tau}\right]\left(\tilde{t}\right) \lesssim \int_{t^*-1/2}^{t^*} E_k\left[\psi_{\tau}\right]\left(s\right)\, ds.
\end{equation}
Another consequence of~\eqref{finiteintime} (and the extension procedure of Corollary~\ref{corextend}) is that 
\begin{equation}\label{oij32rji2r3i}
\sup_{t^* \in (-2,12)}E_k\left[\psi_{\tau}\right](t^*) \lesssim_k E_k\left[\psi\right](0),
\end{equation}
\begin{equation}\label{92u3jingiuo2}
\sup_{t^* \in ( \tau-10,\tau+1)}E_k\left[\psi_{\tau}\right](t^*) \lesssim_k E_k\left[\psi_{\tau}\right](\tau-10).
\end{equation}
In view of the definition of $\psi_{\tau}$, the fact that uniform boundedness of non-degenerate holds on the exact Kerr metric, and~\eqref{92u3jingiuo2}, we have that
\begin{equation}\label{3ij2onoi1324}
\sup_{t^* \geq  \tau-10}E_k\left[\psi_{\tau}\right](t^*) \lesssim_k E_k\left[\psi_{\tau}\right]\left(\tau-10\right).
\end{equation}
We first consider the case when $k = 1$. Taking $A_{\rm high}$ sufficiently high and combining~\eqref{fromtheaveragetotherealthing},~\eqref{oij32rji2r3i},~\eqref{3ij2onoi1324} and Lemma~\ref{combineitallwoohoo} yields
\[\sup_{t^* \in (0,\tau-10)}E_1\left[\psi_{\tau}\right](t^*) \lesssim_k E_1\left[\psi\right](0).\]
This concludes the proof, since for any given $t^* > 0$, if $\tau \gg t^*$, then $\psi_{\tau}(t^*) = \psi(t^*)$. 

Having now estimated $\sup_{t^*} E_1\left[\psi_{\tau}\right](t^*)$, the estimate of Lemma~\ref{combineitallwoohoo} becomes
\begin{equation}\label{withtheredaddedin123}\begin{split}
&\sup_{t^* \in (10,\tau-10)}\int_{-\infty}^{\infty}\chi_{t^*}(s)E_k\left[\psi_{\tau}\right](s)\, ds \lesssim_k  A_{\rm high}^{-1}\sup_{t^*} \int_{-\infty}^{\infty}\chi_{t^*}(s)E_k\left[\psi_{\tau}\right](s)\, ds 
\\ &\qquad  + C\left(A_{\rm high},k\right)\left(E_k\left[\psi\right](0) + \int_{-2}^{12}E_k\left[\psi_{\tau}\right](s)\, ds\right).
\end{split}\end{equation}

We may now consider the case $k > 1$. 
Taking $A_{\rm high}$ sufficiently high and combining~\eqref{fromtheaveragetotherealthing},~\eqref{oij32rji2r3i},~\eqref{3ij2onoi1324} and~\eqref{withtheredaddedin123} yields
\[\sup_{t^* \in (0,\tau-10)}E_k\left[\psi_{\tau}\right](t^*) \lesssim_k E_k\left[\psi\right](0).\]
This concludes the proof, since for any given $t^* > 0$, if $\tau \gg t^*$, then $\psi_{\tau}(t^*) = \psi(t^*)$.

\end{proof}
\appendix
\section{Metrics in $\mathscr{A}_{\epsilon,r_{\rm low},R_{\rm high}}$ with stable trapping}\label{stabletrapapp}
Choose $M > 0$, $a \in (-M,M)$, and set, as usual, $r_+ \doteq M + \sqrt{M^2-a^2}$. Motivated by the metric ansatz from~\cite{CarterSeparate} (with a mild modification to match with Boyer--Lindquist coordinates) for any smooth function $H(r) : (r_+,\infty) \to (0,\infty)$ we consider metrics of the form:
\begin{equation}\label{themetricforsep}
g_H \doteq  -\frac{H(r)\Delta }{\rho^2}\left(dt - a \sin^2\theta d\phi\right)^2 + \frac{\sin^2\theta}{\rho^2}\left(adt - (r^2+a^2) d\phi\right)^2 + \frac{\rho^2}{H(r) \Delta}dr^2 + \rho^2 d\theta^2,
\end{equation}
where $\left(t,r,\theta,\phi\right) \in \mathbb{R} \times (r_+,\infty) \times \mathbb{S}^2$ and $\Delta$ and $\rho^2$ take their usual definitions from the Kerr metric. A short calculation shows that, when $H = 1$,~\eqref{themetricforsep} is the same as the Kerr metric expressed in Boyer--Lindquist coordinates~\eqref{metric}. 

In this next lemma we check that, for suitable $H$, we have that $g_H \in \mathscr{A}_{\epsilon,,r_{\rm low},R_{\rm high}}$, where $\mathscr{A}_{\epsilon,,r_{\rm low},R_{\rm high}}$ is defined in Definition~\ref{themetricclass}.
\begin{lemma}\label{Hisgoodandall}Let $r_+ < r_{\rm low} < R_{\rm high} < \infty$ and $\epsilon > 0$. Suppose that $H$ satisfies that, for a suitable small constant c, we have
\begin{enumerate}
    \item $r \not\in (r_{\rm low},R_{\rm high}) \Rightarrow H = 1$.
    \item $\sup_{r \in [r_{\rm low},R_{\rm high]}}\left[\left|H(r)-1\right| + \left|H'(r)\right|\right]  \leq c\epsilon$.
\end{enumerate}
Then $g_H \in \mathscr{A}_{\epsilon,r_{\rm low},R_{\rm high}}$.
\end{lemma}
\begin{proof}If $\theta$ is not close to $0$ or $\pi$ then the $C^1$ closeness of $g_H$ to the Kerr metric is immediate. We thus only need to show that $g_H$ is $C^1$ close to $g$ for $\theta$ close to $0$ or $\pi$. The two cases are symmetric so we will only discuss $\theta = 0$.

To construct a regular coordinate system near $\theta = 0$, one simply introduces the new coordinate system $(t,r,x,y)$ where $x = \theta \cos\phi$ and $y = \theta \sin\phi$. A short computation shows that in these new coordinate system, $g_H$ is $C^1$ close to the Kerr metric.
\end{proof}
\begin{convention}From now on, unless said otherwise, we assume that $H$ satisfies the hypothesis of Lemma~\ref{Hisgoodandall}.
\end{convention}

We next record the formulas for $\det\left(g_H\right)$ and $g_H^{-1}$. 
\begin{lemma}\label{eqnsofgeodflow}We have
\[{\rm det}\left(g_H\right) = -\rho^4\sin^2\theta,\]
\begin{equation}\begin{split}
&g_H^{-1} = \left(\frac{a^2\sin^2\theta}{\rho^2} - \frac{(r^2+a^2)^2}{H\Delta \rho^2}\right)\partial_t\otimes\partial_t + \left(\frac{1}{\sin^2\theta\rho^2} - \frac{a^2}{H \Delta \rho^2}\right)\partial_{\phi}\otimes\partial_{\phi} 
\\ &\qquad \qquad + \left(\frac{a(r^2+a^2)}{H \Delta \rho^2}-\frac{a}{\rho^2} \right)\left(\partial_t\otimes\partial_{\phi}+ \partial_{\phi}\otimes\partial_t\right) + \frac{H\Delta}{\rho^2}\partial_r\otimes\partial_r + \rho^{-2}\partial_{\theta}\otimes\partial_{\theta}.
\end{split}\end{equation}
\end{lemma}

It follows from the results of~\cite{CarterSeparate} that the Hamilton--Jacobi equation associated to the geodesic flow of the metric $g_H$ is separable, and thus the geodesic flow is integrable. This leads to the following lemma.
\begin{lemma}Let $\left(t(s),\phi(s),r(s),\theta(s)\right)$ denote a null geodesic with respect to the metric $g_H$. Then there exist real valued constants $E$, $L$, and $Q$ so that 
\begin{equation}\label{Econs}
\left(\frac{H\Delta}{\rho^2} - \frac{a^2\sin^2\theta}{\rho^2}\right)\frac{dt}{ds} - a\left(\frac{H\Delta  \sin^2\theta}{\rho^2} - \frac{ \sin^2\theta \left(r^2+a^2\right)}{\rho^2}\right)\frac{d\phi}{ds} = -E,
\end{equation}
\begin{equation}\label{Lcons}
\left(\frac{(r^2+a^2)^2\sin^2\theta}{\rho^2} - \frac{a^2\sin^4\theta H \Delta}{\rho^2}\right)\frac{d\phi}{ds}+ a\left(\frac{H\Delta  \sin^2\theta}{\rho^2} - \frac{ \sin^2\theta \left(r^2+a^2\right)}{\rho^2}\right)\frac{dt}{ds} = L,
\end{equation}
\begin{equation}\label{thetacons}
\rho^4\left(\frac{d\theta}{ds}\right)^2 = Q - a^2E^2\sin^2\theta - \frac{L^2}{\sin^2\theta}+2aEL,
\end{equation}
\begin{equation}\label{thercons}
\rho^4\left(\frac{dr}{ds}\right)^2 = \left(\left(r^2+a^2\right)E - aL\right)^2  - Q \Delta H.
\end{equation}

In particular, we may consider the class of null geodesics where $\theta(s) = \pi/2$ for all $s$. For this class of geodesics, $E$ and $L$ may be arbitrary real numbers, we must have 
\begin{equation}\label{QEL} Q = \left(aE - L\right)^2,
\end{equation}
and the equation~\eqref{thercons} becomes 
\begin{equation}\label{equitorialplanenullgeod}
r^4\left(\frac{dr}{ds}\right)^2 = \left(\left(r^2+a^2\right)E - aL\right)^2  - \left(aE-L\right)^2 \Delta H
\end{equation}

\end{lemma}
\begin{proof}We omit the proof as the arguments are standard.
\end{proof}

When $H = 1$, the equations of Lemma~\ref{eqnsofgeodflow} govern the null geodesic flow of the Kerr black hole. These null geodesics have been thoroughly studied in~\cite{chand}. In the following lemma we record some useful facts derived in~\cite{chand} in the course of analyzing of the geodesic flow of Kerr.
\begin{lemma}\label{fromchand}For any $(D,r) \in \mathbb{R} \times (r_+,\infty)$ we define
\[P_1\left(D,r\right) \doteq \left(\left(r^2+a^2\right)-aD\right)^2 - \left(a-D\right)^2\Delta.\]
Then there exists exactly two choices for $\left(D,r\right)\in \mathbb{R} \times (r_+,\infty)$  which we denote by $\{\left(D_i,r_i\right)\}_{i=1}^2$  so that 
\begin{equation}\label{trappedcond}
P_1\left(D_i,r_i\right) = 0,\qquad \partial_rP_1\left(D_i,r_i\right) = 0.
\end{equation}

We moreover have that 
\begin{equation}\label{aminusdinotzero}
D^2_i = 3r_i^2 + a^2,\qquad D_i > 0.
\end{equation}
\end{lemma}
\begin{proof}These facts are proved directly or are easy consequences of the analysis of Section 61 of Chapter 7 of~\cite{chand}.
\end{proof}
\begin{remark}
An immediate consequence of~\eqref{trappedcond} is that the null geodesics which have $D = D_i$ will satisfy $\frac{dr}{ds} = 0$ if $r\left(0\right) = r_i$; however, as is well known, this circular orbit is unstable to small perturbations of $r\left(0\right)$.
\end{remark}

In this next lemma we show that for (a wide class of) suitable functions $H$ there exist (non-axisymmetric) null geodesics which experience stable trapping. 
\begin{lemma}There exists a $1$-parameter family of functions $H_{\epsilon}(r)$ so that for each $\epsilon > 0$ sufficiently small the following holds true: There exist $\left(r_0,\theta_0,\phi_0\right) \in (r_+,\infty) \times \mathbb{S}^2$,  real numbers $E_0$, $L_0$, and $Q_0$, and $\tilde{\epsilon} > 0$ (where $\tilde{\epsilon}$ depends on $\epsilon$)  so that any null geodesic $\gamma(s) = \left(t(s),r(s),\theta(s),\phi(s)\right)$ which satisfies
\begin{equation}\label{closeto}
\left|r(0) - r_0\right| + \left|\theta(0)-\theta_0\right| + \left|\phi(0)-\phi_0\right| +  |E-E_0| + |Q-Q_0| + |L-L_0|   \leq \tilde{\epsilon}
\end{equation}
will also satisfy
\[\sup_{s \in \mathbb{R}}\left|r(s) - r_0\right| \leq \epsilon.\]
Furthermore, null geodesics which satisfy~\eqref{closeto} are not axisymmetric.
\end{lemma}
\begin{proof}
 Fix a choice of $\left(D_i,r_i\right)$ from Lemma~\ref{fromchand}. We pick $r_0=r_i$, $\theta_0 = \pi/2$,  $Q_0 = (aE_0-L_0)^2$, $L_0/E_0 = D_i$, and $\phi_0$ arbitrary.

It will be convenient to set $D \doteq L/E$. We introduce the function 
\[P_H(r,D) \doteq \left(\left(r^2+a^2\right)-aD\right)^2 - \left(a-D\right)^2\Delta H\]
Note that \eqref{QEL} and \eqref{closeto} imply
 that~\eqref{thercons} becomes  
\[\left(\frac{dr}{ds}\right)^2 =r^{-4}E^2P_H\left(D,r\right)  + C(\tilde\epsilon, r,Q,E,L),\]
 where we denote by $C(\tilde\epsilon, r,Q,E,L)$ a continuous function of $\tilde\epsilon$, $r,Q,E,L$  with $C(0, r_0,Q_0,E_0,L_0)=0$.

We now choose $A$ sufficiently large depending on $M$, $D_i$, and $r_i$. Let $h(x)$ be a smooth function so that ${\rm supp}\left(h\right) \subset (-2,2)$, $h\left(\pm 1\right) = A$, and $h\left(0\right) = -A$. Then, for all $\epsilon > 0$ sufficiently small (depending, among other things, on $A$), we set
\[H_{\epsilon}\left(r\right) \doteq 1 + \epsilon^2h\left(\frac{r-r_i}{\epsilon}\right).\]
Note that we have that $H_{\epsilon}$ satisfies the hypothesis of Lemma~\ref{Hisgoodandall}.

In view of~\eqref{aminusdinotzero}, we have that
\begin{equation}\label{io32ijo3}
{\rm inf}_{r \in [r_i-\epsilon,r_i+\epsilon]}\left(a-D\right)^2\Delta \gtrsim 1.
\end{equation}
Keeping~\eqref{io32ijo3} in mind, in view of~\eqref{trappedcond} and the definition of $h$, we have that
\[P_{H_{\epsilon}}\left(r_i\pm\epsilon\right) \lesssim -\epsilon^2 A,\qquad P_{H_{\epsilon}}\left(r_i\right) \gtrsim \epsilon^2 A.\]

The proof of the lemma then follows,  after picking $\tilde\epsilon\ll\epsilon$,  from the following two facts:

\begin{enumerate}
    \item Along any null geodesics, $r(s)$ is continuous in $s$ and can only take values where the right hand side of~\eqref{thercons} is positive.
    \item The right hand side of~\eqref{thercons} is uniformly continuous in $r$ with respect to the parameters $E$, $L$, and $Q$ (and hence also continuous with respect to changes in the initial conditions).
\end{enumerate}

\end{proof}

\section{Proof of Theorem~\ref{fromtheblackboxlargea}}\label{provethatkeythm}
In this section we briefly explain how Theorem~\ref{fromtheblackboxlargea} follows from~\cite{blackboxlargea}; essentially Theorem~\ref{fromtheblackboxlargea} is an immediate consequence of (the proofs of) Theorem 3.3 and Proposition 3.3.4 from~\cite{blackboxlargea}. 

We start by noting that the Fourier analysis in~\cite{blackboxlargea} is carried out in the context of functions which are Schwartz in $(t,\theta,\phi) \in \mathbb{R} \times \mathbb{S}^2$ for each fixed $r$, while the Fourier analysis in this paper is done with respect to functions lying in the space $\mathbb{H}_K$. However, since the operators defined in~\cite{blackboxlargea} immediately extend to $\mathbb{H}_K$ by standard pseudo-differential operator theory, this difference plays no role in the definition of the operators $P_n$ or in the proof of either Theorem 3.3 or Proposition 3.3.4 from~\cite{blackboxlargea}.

We next note that our operators $\mathcal{P}_{n,\rm high}$ (see Definition~\ref{thisdefinedpnhigh}) differ from the operator $P_n$ defined in~\cite{blackboxlargea} via the symbol being multiplied by an additional smooth cut-off in Fourier space. The result is the Fourier support of a function $\mathcal{P}_{n,\rm high}f$ is contained in the Fourier support of $P_nf$. Since the proofs of Theorem 3.3 and Proposition 3.3.4 in fact rely on applying Plancherel's theorem and then establishing inequalities for each fixed frequency $\left(\omega,m\right)$ lying in the Fourier support of $\mathcal{P}_n$, it is immediate that we may replace $P_n$ with the operator $\mathcal{P}_{n, \rm high}$ in Theorems 3.3 and Proposition 3.3.4 in~\cite{blackboxlargea}.

We next note that while Theorem 3.3 and Proposition 3.3.4 in~\cite{blackboxlargea} assume $k \geq 1$, it is immediate from the proof and the discussion of the previous paragraph that we can take $k = 0$ as long the last three terms of the right had side of the estimate from Theorem 3.3 are replaced with their $k = 1$ version.

We now observe that an inspection of the proof of Theorem 3.3 when applied to our operators $\mathcal{P}_{n,\rm high}$ replacing $P_n$ shows that integrals containing the integral of the final three terms on the right hand of the main estimate of Theorem 3.3 can be restricted to $r \in [\tilde{r},\tilde{R}]$ for $\tilde{r}$ sufficiently close to $r_+$ and $\tilde{R}$ sufficiently large.  specifically, we note that the source of the final three terms on the right hand side of Theorem 3.3 arises from the final two terms $\Delta r^{-5}|u|^2$ and $\Delta \iota_{\rm elliptic}\left(V_0+\omega^2\right)|u|^2$ which appear on the right hand side of the estimates of Propositions A.4.3, A.4.5, and A.4.7. On the one hand, in view of the facts that we have replaced $P_n$ with $\mathcal{P}_{n,\rm high}$ and that $m^2 \gg 1$ on the support of $\mathcal{P}_{n,\rm high}$, it is clear that the term proportional to $\Delta r^{-5}|u|^2$ is easily absorbed in the available positive terms. On the other hand, by definition, the indicator function $\iota_{\rm elliptic}$ is supported in a bounded range of $r$, and also, Propositions A.4.3, A.4.5, and A.4.7 all explicitly state that the term involving $\Delta \iota_{\rm elliptic}\left(V_0+\omega^2\right)|u|^2$ may be assumed to be supported in the region $r \geq r_{\rm close}$ for suitable $r_{\rm close} > r_+$.

With the previous two paragraphs in mind,  we pick $V_n$ and $w_n$ to be the vector fields and scalar functions from Theorem 3.3 in~\cite{blackboxlargea}, while the operator $\mathscr{A}_n\left[\psi\right]$ in Theorem~\ref{fromtheblackboxlargea} may be defined as to be the sum of the three terms on the second line of the right hand of Theorem 3.3 in~\cite{blackboxlargea} with $\tau_0 = -\infty$, $\tau_1 = \infty$, k=1, and $\psi$ replaced by $\check{\chi}(r)\psi$, where $\check{\chi}(r)$ is identically $0$ for $r \leq \left(1+\tilde{\delta}\right)r_{\rm low}$ and  $r \geq \left(1-\tilde{\delta}\right)R_{\rm high}$, and identically $1$ for $r \in [\left(1+2\tilde{\delta}\right)r_{\rm low}, \left(1-2\tilde{\delta}\right)R_{\rm high}]$ . Since $\psi$ and $F$ will all lie in $\mathbb{H}_K$ there is no problem taking the limits $\tau_0 \to -\infty$ and $\tau_1 \to \infty$. The above discussion then justifies every part of Theorem~\ref{fromtheblackboxlargea} except for the estimate~\eqref{mathscranbound}.   We remark that, even though the estimate in Theorem 3.3 of~\cite{blackboxlargea} is stated for $\mathscr{A}_n\left[\psi\right]$, it is actually proved for $\mathscr{A}_n\left[P_{n} \psi\right]$, see (A.114) in~\cite{blackboxlargea}.

We now claim that the estimate~\eqref{mathscranbound} for $\mathscr{A}_n$ is a consequence of Proposition 3.3.4 in~\cite{blackboxlargea} with $k = 0$ applied to $\check{\chi}\psi$. Indeed it suffices to observe that for any function $\tilde\psi$, $\Box_{g_K}\tilde\psi = F$ implies that $\Box_{g_K}\left(\check{\chi}\tilde\psi\right) = F + \check{F}$, where $\check{F}$ involves $\tilde\psi$ and first derivatives of $\tilde\psi$ and is supported in the region $r \in [(1+\tilde{\delta})r_{\rm low},(1+2\tilde{\delta})r_{\rm low}] \cup \left[\left(1-2\tilde{\delta}\right)R_{\rm high},\left(1-\tilde{\delta}\right)R_{\rm high}\right]$.

\bibliographystyle{plain}

\end{document}